\documentclass[prx,superscriptaddress,longbibliography]{revtex4-2}

\usepackage[utf8]{inputenc}

\usepackage[T1]{fontenc}
\usepackage{lmodern}

\usepackage{qcircuit}
\usepackage{braket}
\usepackage{amsfonts,amsmath,amssymb,amsthm}
\usepackage{mathtools}

\usepackage{bm}
\usepackage{graphicx}
\usepackage{ascmac}
\usepackage{mathrsfs}
\usepackage{float}
\usepackage{url}
\usepackage{natbib}
\usepackage{algorithm}
\usepackage{algpseudocode}

\theoremstyle{plain}
\newtheorem{thm}{Theorem}

\newtheorem{lem}[thm]{Lemma}
\newtheorem{cor}[thm]{Corollary}
\newtheorem*{thm*}{Theorem}
\newtheorem*{lem*}{Lemma}
\newtheorem*{cor*}{Corollary}

\theoremstyle{definition}
\newtheorem{dfn}{Definition}

\theoremstyle{remark}
\newtheorem{rem}[thm]{Remark}
\newtheorem*{rem*}{Remark}

\usepackage[colorlinks=true, allcolors=blue]{hyperref}

\newcommand{\green}[1]{\textcolor{black}{#1}}

\hypersetup{breaklinks=true}

\newcommand{\beginsupplement}{%
  \setcounter{section}{0}%
  \renewcommand{\thesection}{S\arabic{section}}%
  \setcounter{thm}{0}%
  \renewcommand{\thethm}{S\arabic{thm}}%
  \setcounter{dfn}{0}%
  \renewcommand{\thedfn}{S\arabic{dfn}}%
  \setcounter{figure}{0}%
  \renewcommand{\thefigure}{S\arabic{figure}}
  \renewcommand{\theHfigure}{S.\arabic{figure}}%
  \setcounter{table}{0}%
  \renewcommand{\thetable}{S\arabic{table}}%
}

\begin{document}

\title{Copy-scarce learning of ground states}

\author{Kaito Wada}
\email{wadakaito.q@gmail.com}
\affiliation{International Center for Elementary Particle Physics, University of Tokyo, 7-3-1 Hongo, Bunkyo-ku, Tokyo 113-0033, Japan}
\affiliation{Graduate School of Science and Technology, Keio University, Hiyoshi 3-14-1, Kohoku, Yokohama 223-8522, Japan}

\author{Nobuyuki Yoshioka}
\affiliation{International Center for Elementary Particle Physics, University of Tokyo, 7-3-1 Hongo, Bunkyo-ku, Tokyo 113-0033, Japan}

\begin{abstract}
% Probing multiple properties of ground states is central to understanding quantum matter, yet conventional readout can be prohibitively costly because it consumes many independently prepared copies.
% Here we establish the fundamental dynamical cost that any protocol must incur in the worst case to estimate multiple
% observables or reconstruct a full classical description of the state from a scarce initial supply of ground-state copies. 
% For unique gapped ground states, we measure this cost by elapsed Hamiltonian time and derive a common inequality capturing the distinguishability supplied by initial copies and accumulated through controlled dynamics, even when copies may be consumed.
% We construct parallel catalytic readouts that finally return all copies nearly unchanged and require no coherent preparation circuit or its inverse.
% In identified regimes, they attain the minimum elapsed Hamiltonian time implied by the inequality up to poly-logarithmic factors, with only logarithmic dependence on inverse return accuracy. 
% Catalytic return therefore incurs no leading-order penalty.
% The resulting copy--dynamics landscape connects copy-scarce and copy-only learning, identifying where optimality remains unresolved.
% Even under an optimistic preparation-cost model favoring repeated preparation, numerical resource estimates for low-order fermionic correlations show reductions reaching a factor of about three million in end-to-end elapsed Hamiltonian time relative to a task-specialized copy-consuming protocol.

Probing multiple properties of ground states is central to understanding quantum matter, yet conventional readout can be prohibitively costly because it consumes many independently prepared copies.
Here, in copy-scarce regimes, we establish optimal trade-offs between the number of initial ground-state copies and elapsed Hamiltonian-evolution time required for estimating multiple observables or reconstructing a full classical description of a unique gapped ground state.
We derive a unifying inequality that quantifies how well ground states can be distinguished using initial copies and controlled dynamics, even when all copies may be consumed.
It yields worst-case lower bounds on elapsed Hamiltonian time that any protocol must incur to perform either task.
We construct parallel catalytic readouts that return all initial copies nearly unchanged while attaining these bounds up to poly-logarithmic factors. 
Our protocols require no coherent preparation circuit or its inverse, and their elapsed Hamiltonian time depends only logarithmically on the inverse of the desired return error.
Thus, optimal copy--dynamics trade-offs can be attained together with catalytic return.
Numerical benchmarks further show substantial reductions in end-to-end elapsed Hamiltonian time relative to a task-specialized copy-consuming protocol, even under an optimistic preparation-cost model.
Our results reveal a copy--dynamics resource landscape connecting copy-scarce learning to copy-only endpoints realized by known protocols for shadow tomography and pure-state tomography.

\end{abstract}

\maketitle

\section{Introduction}

Preparing a many-body ground state does not by itself reveal its physical properties. 
The state prepared on a quantum device remains quantum rather than an accessible classical description, and one then needs to perform quantum measurements to extract meaningful information.
Conventional measurements generally disturb the state and only yield a classical outcome with large uncertainty; in order to reveal correlations~\cite{huang2020predicting,PRXQuantum.6.010336}, reduced density matrices~\cite{bonet2020nearly,PhysRevLett.127.110504,wan2023matchgate}, configuration weights~\cite{koyluouglu2026measuring}, or full classical descriptions~\cite{haah2016sample,o2016efficient,pelecanos2025mixed}, it is typically required to repeat preparation of many identical fresh copies~\cite{Zhang2022computingground,Chakraborty2024implementingany,wada2025trade}.
More sophisticated measurements enable copy-efficient protocols~\cite{10.1145/3188745.3188802,buadescu2021improved, chen2022exponential,huang2021information,huang2022quantum}, but still require repeated state preparation or the demanding preparation of an extensive-size stock of copies before a joint measurement~\cite{chen2024optimal}.
When preparing each nontrivial ground state requires substantial resources~\cite{lee2023evaluating}, measurement protocols that consume the entire fresh copies are prohibitive in practice.
%it is in practice prohibitive to repeatedly consume fresh copies.
%is prohibitive in practice.

% A common resource in ground-state problems is the controlled time evolution of a Hamiltonian that defines the state.
% Many leading approaches to preparing ground states are naturally quantified by Hamiltonian-evolution time~\cite{PhysRevResearch.6.033147,wzb3-dbg9}.
% On readout sides, recent dynamics-assisted protocols have shown that one can learn physical properties from only a single ground-state copy~\cite{chen2025catalytic}, whereas conventional measurement protocols fail with such a limited number of copies.

Controlled dynamics offer a route to extracting information even from a single supplied copy while preserving or restoring the state.
This possibility has been explored through protective measurements~\cite{aharonov1993measurement} and quantum state restoration~\cite{farhi2010quantum}.
Here we focus on controlled time evolution under the Hamiltonian that defines the ground state, a resource also used in ground-state preparation~\cite{albash2018adiabatic,ge2019faster,dong2022ground,PhysRevResearch.6.033147,wzb3-dbg9}.
With this access, Refs.~\cite{farhi2010quantum,chen2025catalytic} achieve optimal time scaling for estimating a single observable with catalytic return, i.e., leaving the supplied state essentially undisturbed after measurement.
Still, these results
% establish that dynamics-assisted readout is possible from a single supplied copy, but 
leave open how much dynamics is fundamentally required to complete collective learning tasks from one copy.
% \red{Furthermore, it is natural and tempting to deepen the question to ask how such dynamical resource requirement can be compensated by supply of target copies.}
Furthermore, it is natural and tempting to ask how supplying additional copies can reduce this dynamical cost.

%\red{Here we further deepen the question by asking the fundamental tradeoff relationship between the dynamical and static resource, by assuming that}
% \red{Consider a problem beyond one ground-state copy setup by assuming that%, 
% } a moderate-size stock of $P$ copies may be prepared before readout and then jointly processed.

Beyond the single-copy setup, we consider a moderate-size stock of $P$ ground-state copies prepared before readout and then jointly processed.
Copy-consuming approaches can exploit this stock through $P$-way parallel single-copy measurements and, in some settings, more powerful joint measurements~\cite{PRXQuantum.6.010336,chen2024optimal,pelecanos2026debiased}, while coherent protocols may likewise process all $P$ copies in parallel.
Our question can concretely %therefore 
be refined as follows: given this stock and controlled access to the system Hamiltonian, what is the minimum elapsed Hamiltonian-evolution time required to learn $M$ physical properties or reconstruct the ground state classically?
We refer to this dependence of the dynamical cost on the initial stock size as the \textit{copy--dynamics trade-off}.
Must any optimal learning protocol consume the stock, or can the same dynamical cost be attained while all supplied copies are returned nearly unchanged?
% Recent dynamics-assisted protocols have shown catalytic return for a single ground-state copy~\cite{chen2025catalytic}, but 
Whether this return is compatible with the optimal copy--dynamics trade-off for multi-observable estimation and full state reconstruction has remained unknown.

\begin{figure}[tb]
 \centering
 \begin{tabular}{c}
 \includegraphics[scale=0.63]{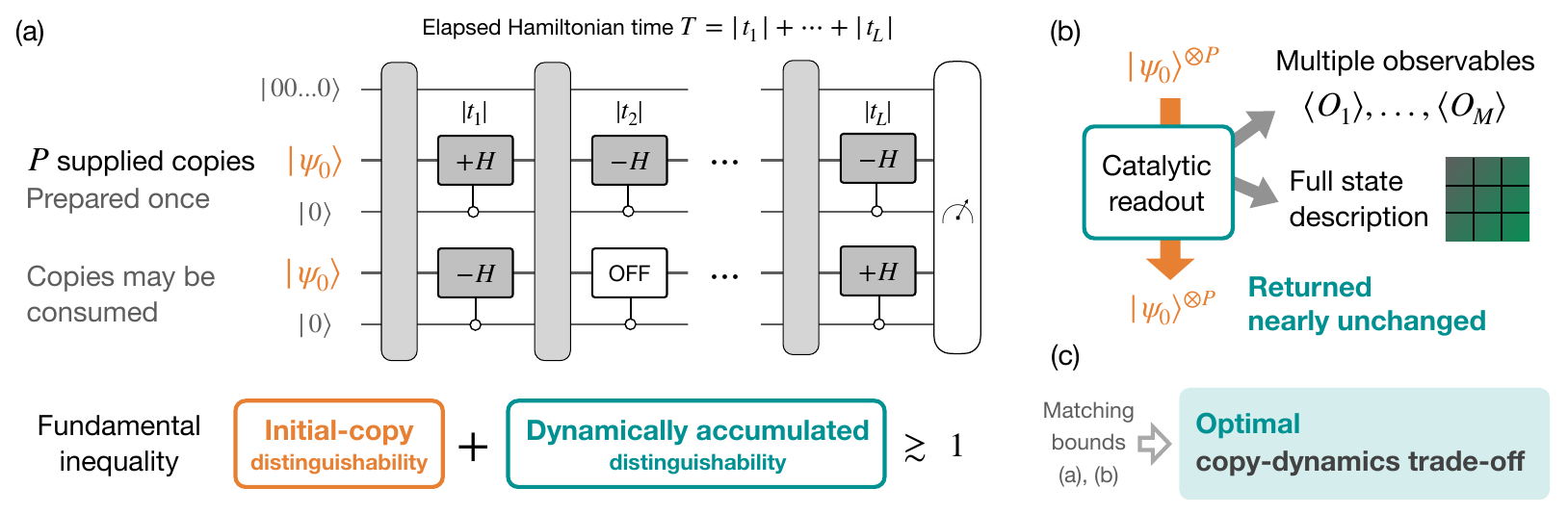}
 \end{tabular}
 \caption{
\textbf{Optimal copy--dynamics trade-off with catalytic return.}
(a) The circuit represents a general measurement protocol under our access model.
$P$ copies of the unique ground state $|\psi_0\rangle$ of a Hamiltonian $H$ are prepared before readout.
During each interval $|t_\ell|$, each system may
independently undergo controlled forward ($+H$), backward ($-H$), or no Hamiltonian evolution (OFF).
Grey vertical boxes represent arbitrary $H$-independent operations, including joint and adaptive processing; all supplied copies may be consumed.
The readout cost is the elapsed Hamiltonian time $T=\sum_{\ell=1}^{L}|t_\ell|$.
% , counting each parallel segment once rather than once per system.
The inequality below the circuit requires the initial-copy and dynamically accumulated contributions to jointly reach a threshold for successful learning of a computational-basis distribution.
It yields lower bounds on the elapsed Hamiltonian time for
multi-observable estimation and full state reconstruction.
(b) Using the $P$-parallel Hamiltonian access and observable block-encoding unitaries, our catalytic readout simultaneously estimates the expectation values of $M$ possibly noncommuting observables.
Its nearly unbiased version also reconstructs a full classical description of the ground state.
Both protocols return all $P$ supplied copies nearly unchanged; the inverse return error enters the elapsed Hamiltonian time only through a logarithmic factor.
(c) In copy-scarce regimes, catalytic upper bounds match lower bounds that allow copy consumption, up to poly-logarithmic factors.
Thus, catalytic return imposes no leading-order penalty on
the optimal copy--dynamics trade-off.}
 \label{fig:overview}
\end{figure}

In this work, we first establish a static--dynamic distinguishability inequality that constrains copy--dynamics trade-offs.
Consider learning the computational-basis probability distribution of a $d$-dimensional ground state to $\ell_1$ error $\omega$.
For sufficiently small $\omega>0$, we prove that any measurement protocol that succeeds with probability at least $2/3$ for every allowed input, using $P\geq 1$ supplied copies and $P$-parallel controlled Hamiltonian evolution (Fig.~\ref{fig:overview}(a)), must obey
\begin{equation}
    \underbrace{
        \frac{P\omega^2}{d}
    }_{\text{initial-copy distinguishability}}
    +
    \underbrace{
        \frac{P\omega\Delta T}{d}
    }_{\text{dynamically accumulated distinguishability}}\gtrsim 1,
    \label{eq:static_dynamic}
\end{equation}
where $\Delta$ is the spectral gap and $T$ is the worst-case elapsed Hamiltonian time.
The distinguishability both from copies and dynamics is evaluated by a common quantity; the proof allows arbitrary joint and adaptive processing, and further all supplied copies may be consumed.
Through suitable reductions, it yields lower bounds for both multi-observable estimation and full state reconstruction.

Combining the resulting task-specific lower bounds with explicit readout protocols below, we establish nearly optimal copy--dynamics trade-offs in copy-scarce ground-state learning.
% the dynamical limit of learning from a finite stock of costly ground-state copies.
% Combining the resulting task-specific lower bounds with explicit readout protocols, we establish the dynamical limit of learning from a finite stock of costly ground-state copies.
% by combining the resulting task-specific lower bounds with explicit readout protocols. 
We introduce a coherent parallel readout protocol that uses $P$ supplied copies to estimate the expectation values of $M$ possibly noncommuting observables simultaneously to additive error $\varepsilon$. 
Its worst-case elapsed controlled Hamiltonian-evolution time is $\tilde{\mathcal{O}}({\Delta}^{-1} (\varepsilon^{-1}{\sqrt M}/P+1))$, with improved $M$ dependence for structured observable sets. 
The same readout framework also enables full ground-state tomography: its nearly unbiased version reconstructs a classical description of a $d$-dimensional ground state to trace-distance error $\eta$ in the elapsed Hamiltonian time $\tilde{\mathcal{O}}(\Delta^{-1}(\eta^{-1}{d}/{P}+1))$. 
Together, the upper and lower bounds establish the optimal copy--dynamics trade-off up to poly-logarithmic factors for each supplied-copy count $P$ in the respective copy-scarce regimes specified below. 
In particular, for multi-observable estimation, this worst-case near-optimality holds simultaneously in $M,\varepsilon,P$ and $\Delta$.

Crucially, these nearly-optimal protocols 
% attain these costs while returning 
return all $P$ supplied ground-state copies nearly unchanged, even though the lower bounds allow all copies to be consumed (Fig.~\ref{fig:overview}(b)). 
The inverse return error enters the elapsed Hamiltonian time only through a logarithmic factor.
Consequently, the optimal copy--dynamics trade-offs can attain together with catalytic return in copy-scarce regimes, up to poly-logarithmic factors.
For sufficiently large number of copies, (collective) shadow tomography~\cite{10.1145/3188745.3188802,buadescu2021improved,huang2021information,grier2024sample,chen2024optimalshadow,jeronimo2026dimension} and pure-state tomography~\cite{hayashi1998asymptotic,scharnhorst2025optimal,pelecanos2025mixed} enable learning from the supplied copies alone, but consume them in the protocol.
Our results connect these copy-only endpoints to the copy-scarce regimes; the resulting copy--dynamics landscape (Fig.~\ref{fig:phase_diagram}) identifies where catalytic return incurs no leading-order penalty and where the optimal trade-off remains unresolved.
% We note that these protocols require neither coherent access to a ground-state preparation circuit nor its inverse, unlike previous coherent methods based on quantum amplitude estimation and quantum gradient estimation~\cite{PhysRevA.75.012328,PhysRevLett.129.240501}.

Our numerical resource estimates further show that the catalytic readout can substantially reduce the end-to-end elapsed Hamiltonian time even against task-specialized copy-consuming protocols under a deliberately optimistic preparation-cost model. 
We explicitly evaluate the constant prefactors and logarithmic overheads of the catalytic readout for representative learning tasks relevant to quantum chemistry, many-body quantum physics, and quantum information,
and include ground-state preparation cost for both approaches.
For complete two-body fermionic reduced density matrix (RDM) estimation on $m=20$ selected modes of a $200$-mode system at half filling, the reductions reach approximately $2\times10^3$-fold at complex entrywise error $\epsilon=10^{-2}$ and $3\times10^6$-fold at $\epsilon=10^{-4}$. 
Comparisons with global-Clifford classical shadows for stabilizer-state fidelity prediction and direct sampling for full computational-basis distribution learning likewise reveal regimes favorable to our catalytic readout.

% We then quantify when the catalytic readout reduces the end-to-end elapsed Hamiltonian time compared with repeated ground-state preparation and measurement 
% for representative learning tasks relevant to quantum chemistry, many-body physics, and quantum information.
% Our numerical resource estimates explicitly evaluate the constant prefactors and logarithmic overheads in our catalytic readout cost, include ground-state preparation cost in both approaches, and compare against task-specialized copy-consuming protocols under the same parallelism.
% Even against an optimistic model of preparation cost with task-specialized protocols, substantial reductions emerge over broad parameter ranges. 
% For estimating a two-body fermionic reduced density matrix (RDM) on 4--25 selected modes under global particle-number symmetry, the reduction exceeds $10^3$-fold at entrywise estimation error $10^{-2}$ and 
% $10^6$-fold at $10^{-4}$.
% Benchmarks for Clifford-state fidelity prediction and full
% computational-basis distribution estimation also identify
% regimes in which catalytic readout reduces the end-to-end elapsed Hamiltonian time.

Overall, these results show that our catalytic readout can attain the optimal copy--dynamics trade-off up to poly-logarithmic factors while returning the supplied ground states nearly unchanged (Fig.~\ref{fig:overview}(c)), and can substantially reduce the end-to-end elapsed Hamiltonian time over broad ranges in our numerical benchmarks.
This combination motivates a shift from repeatedly preparing and measuring fresh copies to ground-state learning workflows that prepare a finite stock of ground states once and retain them after catalytic readout with controlled Hamiltonian dynamics.

\section{Main results}
\subsection{Setup: copy-scarce ground-state access}
Throughout this paper, we focus on the unique ground state $\ket{\psi_0}$ of a system Hamiltonian $H$ with a known spectral gap $\Delta>0$.
For a predetermined set of possibly noncommuting observables $O_1,...,O_M$ with the operator norm at most one, our primary task is estimating all the expectation values 
\begin{equation}
    \mu_j:=\langle \psi_0|O_j|\psi_0 \rangle,~~~\bm{\mu}=(\mu_1,...,\mu_M),~~~j=1,...,M.
\end{equation}
Measurement protocols considered in this paper are supplied with $P$ ground-state copies, controlled access to the Hamiltonian evolution $e^{-itH}$ for arbitrary $t\in\mathbb{R}$, and controlled block-encoding unitaries of observables.
This includes the time-reversal operation $e^{+itH}$, but requires neither coherent ground-state preparation circuit nor its inverse unlike previous coherent approaches~\cite{PhysRevA.75.012328,PhysRevLett.129.240501,PRXQuantum.6.020308}; our readout is therefore compatible with dissipative ground-state preparation~\cite{PhysRevResearch.6.033147,wzb3-dbg9}.

The $P$ system registers can be evolved simultaneously.
% In each parallel time-evolution segment with time duration $t$,
% the protocol can apply a controlled or uncontrolled evolution of time $\pm t$ to each of the $P$ system registers while maintaining the $P$-way parallel structure.
% Any protocol in this access model alternately applies the parallel segment and a fixed unitary gate acting on the system registers and a finite number of ancilla qubits.
If a protocol uses $L$ parallel time-evolution segments and its $l$-th segment has time duration $|t_l|$, we define its elapsed controlled Hamiltonian-evolution time as
\begin{equation}
    T=\sum_{\ell=1}^L |t_\ell|.
    \label{eq:elpsed_htime_main}
\end{equation}
Thus, a parallel segment of time $t_l$ contributes $|t_l|$, rather than $P|t_l|$ (Fig.~\ref{fig:overview}(a)).
All upper and lower bounds below use this same access model.
The black-box and adaptive access model is specified in Methods~\ref{sec:methods-accessmodel}.

\subsection{Static--dynamic distinguishability inequality}

Copy--dynamics trade-off lower bounds on both
multi-observable estimation and full state reconstruction
are derived from Eq.~\eqref{eq:static_dynamic} for a simpler task: estimating the computational-basis probability distribution of a gapped ground state.
The key to Eq.~\eqref{eq:static_dynamic} is to tightly measure the distinguishability arising from the supplied copies and controlled dynamics by using a common quantity.
This quantity is represented by the average real part of the overlaps between the protocol states; the average is taken over hard-to-distinguish Hamiltonian inputs labeled by hidden bits differing only in one bit.
We explicitly construct a gapped-Hamiltonian family whose hidden bits are difficult to distinguish both from the supplied ground-state copies and through controlled dynamics.
% We explicitly construct a gapped-Hamiltonian family whose ground-state distributions encode $\Theta(d)$ hidden bits through weak probability biases induced by a fixed estimation error $\omega$.
This family enables an $\omega$-accurate distribution estimate to decode the hidden bits. Through this reduction and information-theoretic analysis, we obtain the inequality.

The static term $P\omega^2/d$ reflects the standard-quantum-limit-like scaling of statistical resolution from the supplied ground-state copies, each of which has weak probability biases induced by hidden bits and given $\omega$; without dynamics, successful learning requires $P=\Omega(d/\omega^2)$.
The dynamical term $P\omega\Delta T/d$ reflects coherent accumulation of input Hamiltonian differences, giving a Heisenberg-like $1/\omega$ scaling in the required Hamiltonian time; here, $\Delta T$ measures the evolution time in units of the inverse spectral gap.
The initial contribution from the copies and the additional increase allowed by the dynamics must together reach a decoding threshold for successful learning.
The Hamiltonian-family construction, how to measure distinguishability, and the proof are given in Methods~\ref{sec:method-keylb} and \green{the
Supplementary Materials}.

% For this family, the $P$ supplied ground-state copies give an initial value of at most $\mathcal{O}(P\omega^2/d)$ for the quantity.
% Controlled Hamiltonian evolution can increase the same quantity by at most $\mathcal{O}(P\omega\Delta T/d)$ over elapsed Hamiltonian time $T$, even with arbitrary joint and adaptive processing.
% An $\omega$-accurate distribution estimate allows the hidden bits to be decoded, on average, with a constant advantage over
% random guessing; successful decoding therefore requires the final value of this quantity to be bounded below by a positive constant.

\begin{figure}[tb]
 \centering
 \begin{tabular}{c}
 \includegraphics[scale=0.75]{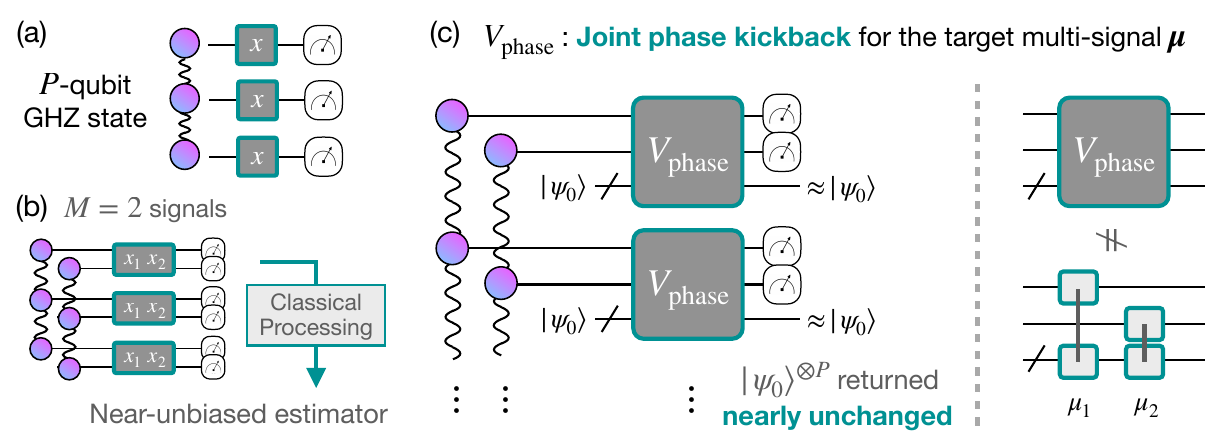}
 \end{tabular}
 \caption{
 \textbf{Catalytic parallel multi-signal readout.}
(a) Standard entanglement-enhanced sensing with one signal.
The boxes apply $S_{\tau x}:={\rm diag}(1,e^{i\tau x})$ in parallel to the qubits of a distributed $P$-qubit GHZ state, producing the amplified collective phase $P\tau x$ in Eq.~\eqref{eq:main_text_ghz}.
(b) Multi-signal extension ($M=2$ shown).
The $P$ devices share $M$ GHZ states, one per signal, which
simultaneously acquire the collective phases $P\tau x_j$.
Single-qubit measurements over several trials at different
$\tau$ values, followed by efficient classical post-processing of the
measurement data, resolve phase periodicity and yield simultaneous signal
estimates; we further develop its nearly unbiased version.
(c) Catalytic implementation for $x_j=\mu_j=\langle\psi_0|O_j|\psi_0\rangle$.
Each of the $P$ devices separably implements the joint multidimensional phase-kickback operation $V_{\mathrm{phase}}$ using controlled Hamiltonian evolution and observable block-encoding unitaries, while returning one ground-state copy nearly unchanged.
The possibly noncommuting observables are encoded jointly,
rather than through separate single-observable phase-kickback
operations (right).
After GHZ-state distribution, only local operations and single-qubit measurements are required within each device with one ground-state copy,
followed by classical post-processing of the measurement data. Additional ancillary registers are omitted.
 }
 \label{fig:catalytic_paralle_readout}
\end{figure}

\subsection{Parallel catalytic readout via joint multidimensional phase kickback}
Our protocol can be considered as a multidimensional extension of entanglement-enhanced sensing~\cite{giovannetti2006quantum,giovannetti2011advances}, which we call \textit{parallel multi-signal readout}.
Consider first a single physical signal $x$ (Fig.~\ref{fig:catalytic_paralle_readout}(a)). 
Applying a phase sensing unitary with parameter $\tau$ to each qubit of a distributed $P$-qubit GHZ state in parallel produces
\begin{equation}
    |{\rm GHZ}_P(\tau x)\rangle:= \frac{\ket{0}^{\otimes P}+e^{iP\tau x}\ket{1}^{\otimes P}}{\sqrt{2}}.
    \label{eq:main_text_ghz}
\end{equation}
The signal therefore appears as the collective phase shift $P\tau x$.
This enhanced-sensitivity can be attained by using only single-qubit measurements followed by classical post-processing of the measurement data~\cite{PhysRevA.102.042613}; after sharing the GHZ state, this primitive process can be performed in parallel by $P$ local quantum devices.
For a multi-signal problem, the local devices share $M$ GHZ states and simultaneously imprint the signals $x_1,\ldots,x_M$ into phases of distinct GHZ states (Fig.~\ref{fig:catalytic_paralle_readout}(b)).

For our problem, the signal components are the expectation values $x_j=\mu_j$.
It therefore remains to realize the desired phase shift without knowing $\bm{\mu}$.
We then construct a joint multidimensional phase-kickback operation $V_{\rm phase}$ from controlled Hamiltonian evolution and observable block-encoding unitaries such that
\begin{align}
        V_{\rm phase}^{\otimes P}\left(|{\rm GHZ}_P\rangle^{\otimes M}\otimes|\psi_0\rangle^{\otimes P}\right)
        \simeq \left[\bigotimes_{j=1}^M |{\rm GHZ}_P(\tau \mu_j) \rangle \right]\otimes|\psi_0\rangle^{\otimes P}.
        \label{eq:catalytic_multisignal_map}
\end{align}
Ancillary registers and a global phase are omitted.
Equation~\eqref{eq:catalytic_multisignal_map} imprints the expectation values of possibly noncommuting observables as phase signals on distinct GHZ states simultaneously, while returning the supplied ground-state copies nearly unchanged (Fig.~\ref{fig:catalytic_paralle_readout}(c)).
The returned ground states can therefore be reused to generate other phase-shifted GHZ states.
Collecting the single-qubit measurement data for logarithmically many different $\tau$ values resolves phase periodicity and yields the following result.
\begin{thm}
    [Parallel catalytic readout, informal]\label{thm_main:parallel_catalytic_readout}
    For any $P\geq 1$, there exists a $P$-parallel measurement protocol that with high probability at least $1-\delta$, produce $\varepsilon$-approximate estimates of $\mu_j$ simultaneously and returns a joint state $\delta$-close to the initial input copies in trace-distance.
    Its elapsed controlled Hamiltonian-evolution time is 
    \begin{equation}\label{eq:thm_main:parallel_catalytic_readout}
        T^{\rm cat}_{\rm obs}=\mathcal{\widetilde{O}}\left[\frac{1}{\Delta}\left(\frac{\Gamma_{\rm obs}}{P\varepsilon}+1\right)\right],~~~\Gamma_{\rm obs}=\left\|\sum_{j=1}^MO_j^2\right\|^{1/2}\leq \sqrt{M}.
    \end{equation}
    Here, $\widetilde{\mathcal{O}}$ hides poly-logarithmic factors, including those in $\delta^{-1}$ and the intrinsic rank of $\sum_j O_j^2$.
\end{thm}
\noindent
Notably, this measurement protocol returns all the supplied ground-state copies with controllable small return accuracy. 
The return accuracy $\delta$ only logarithmically contributes to the elapsed Hamiltonian time.

Our protocol attains the elapsed Hamiltonian-time scaling by integrating multiple quantum-enhanced mechanisms.
The factor $\Gamma_{\rm obs}$ is the root-sum-square operator norm of the observable set.
This yields a square-root dependence on $M$ even in the worst-case observable set due to the joint phase kickback rather than performing $M$ phase-sensing unitary independently (Fig.~\ref{fig:catalytic_paralle_readout}(c)).
The Heisenberg-like scaling $1/\varepsilon$, together with the linear $P$ improvement, comes from the extended entanglement-enhanced readout, whereas the inverse gap $1/\Delta$ sets the time required to suppress transitions between the ground and excited sectors.
Our protocol first attains all of these scalings simultaneously, which turns out to be essentially worst-case optimal in the regime specified later.

Reference~\cite{chen2025catalytic} established catalytic readout of a single ground-state expectation value in elapsed Hamiltonian time $\widetilde {\mathcal{O}}(1/(\Delta\varepsilon))$. 
For $1\leq P\leq M$, allocating $M$ such readouts into $P$ copies gives elapsed Hamiltonian time $\widetilde {\mathcal{O}}(M/(P\Delta\varepsilon))$.
That work suggested that combining its catalytic readout with quantum gradient estimation~\cite{gilyen2019optimizing,PhysRevLett.129.240501} might enable estimation of $M$ observables in elapsed Hamiltonian time $\widetilde {\mathcal{O}}(\sqrt M/(\Delta\varepsilon))$ with single-copy input.
We explicitly realize this together with further improved $M$ dependence for structured observable sets by using our joint phase-kickback primitive in \green{the Supplementary Materials}.
Still, independently allocating observables into $P\in [1,M]$ such single-copy protocols only yields $\widetilde {\mathcal{O}}(\sqrt{M/P}/(\Delta\varepsilon))$ in the worst case, rather than the $\sqrt M/P$ dependence achieved in Eq.~\eqref{eq:thm_main:parallel_catalytic_readout}.
Comparisons with other coherent estimation methods are given in Methods~\ref{sec:methods-accessmodel}.

Structured sets of observables give a smaller elapsed Hamiltonian time.
The most basic example is the computational-basis projectors, which reveal occupation-basis configuration weights~\cite{koyluouglu2026measuring}; this set gives $\Gamma^2_{\rm config}\leq 1$.
For many-body fermionic systems, the complete low-order correlations within a selected orbital subspace are described by $k$-body RDMs~\cite{bonet2020nearly,PhysRevLett.127.110504,wan2023matchgate}; 
for $m$ selected orbitals under global particle-number symmetry, $\Gamma^2_{k\rm RDM} \sim \max_{\ell}[\ell^k (m-\ell)^k/(k!)^2]$ holds with the particle number $\ell$ in the selected space.
A randomly-generated size-$M$ catalog of classically tractable stabilizer states, related to a canonical learning task where global-Clifford classical shadows are especially effective~\cite{huang2020predicting}, 
realizes $\Gamma_{\rm stab}^2\sim \log(2^N)/\log(2^N/M)$ except for exponentially small probability in the number of qubits $N$.
As an extreme example, one wants to recover the full classical description of $d$-dimensional ground states.
This can be done by taking a complete set of $\mathcal{O}(d^2)$ observables specifying the matrix entries of ground-state density matrix, leading to $\Gamma^2_{\rm tom}={d}$~\cite{van2023quantum}.
These tasks also admit efficient copy-consuming approaches tailored to their structure, thereby motivating concrete numerical comparisons below.

\begin{figure}[tb]
 \centering
 \begin{tabular}{ccc}
 \includegraphics[scale=0.47]{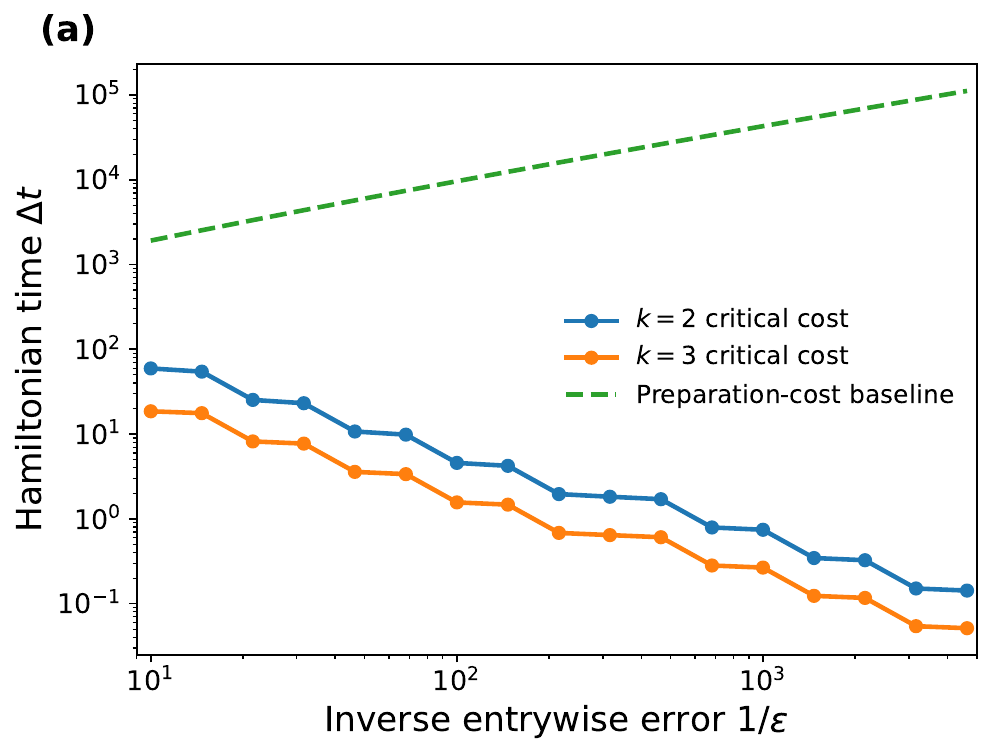}&~~~&\includegraphics[scale=0.47]{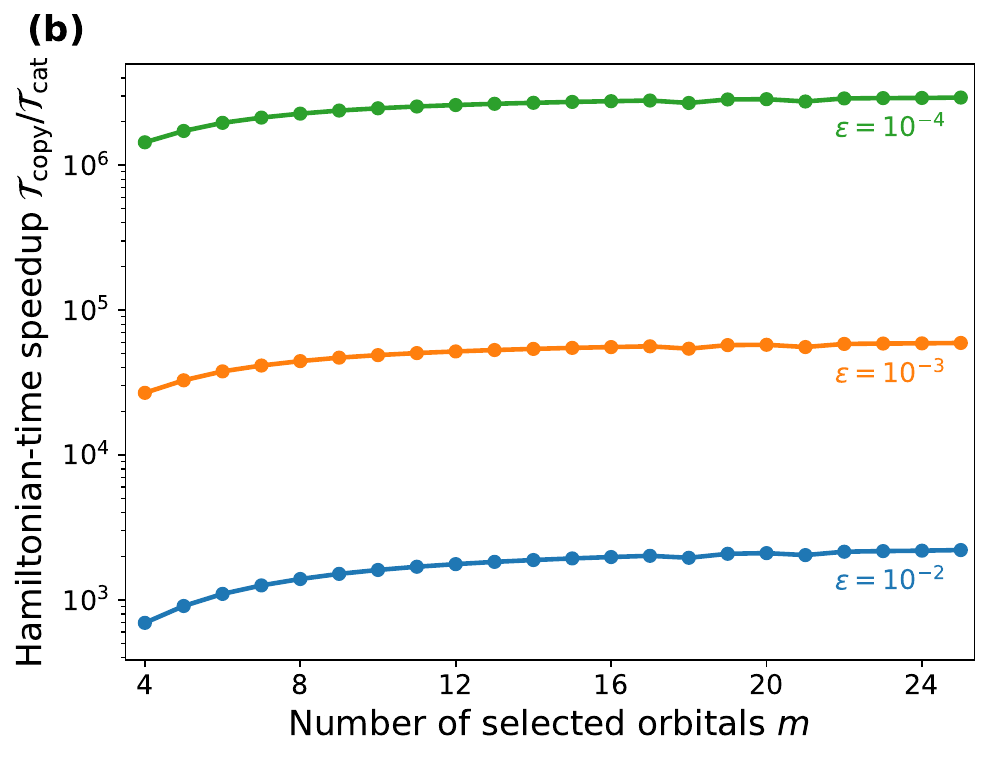}
 \end{tabular}
 \caption{\textbf{End-to-end catalytic-readout advantage in selected-orbital RDM learning.}
We provide numerical Hamiltonian-time estimates for learning selected-orbital $k$-RDM of a gapped ground state on $N=200$ fermionic modes with total particle number $\eta=100$, compared with particle-number-preserving orbital-rotation shadows~\cite{low2022classical,koizumi2026provably}.
(a) Critical one-copy preparation cost $\Delta t_{\rm crit}$ in Eq.~\eqref{eq:critical_preparation_cost_main} for the full 2- and 3-RDMs
of $m=20$ selected modes (solid curves with markers), versus the inverse additive error $1/\varepsilon$ for each complex RDM entry.
The dashed curve shows the \textit{very optimistic} preparation-cost baseline $\Delta t_{\rm prep}$ described in the main text.
Catalytic readout is favorable when this baseline lies above the corresponding critical-cost curve.
(b) End-to-end elapsed Hamiltonian-time speedup $\mathcal T_{\rm copy}/\mathcal T_{\rm cat}$ for full selected-orbital
2-RDM learning as the number of modes $m$ varies, at $\varepsilon=10^{-2},10^{-3},10^{-4}$.
Both total times include the ground-state preparation cost and are evaluated using the same preparation-cost model as in (a); ratios above one make catalytic readout favorable.
We refer \green{the Supplementary Materials} for further details.}
 \label{fig:num_exp_rdm}
\end{figure}

\begin{figure}[htb]
 \centering
 \begin{tabular}{ccc}
 \includegraphics[scale=0.47]{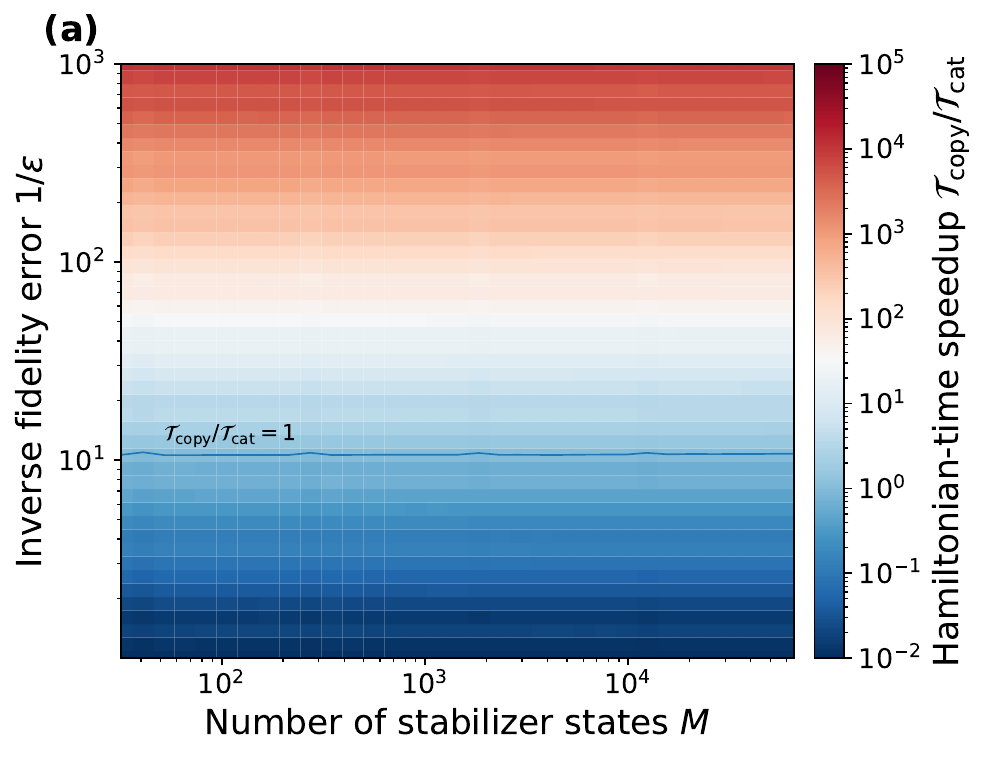}&~~~&\includegraphics[scale=0.47]{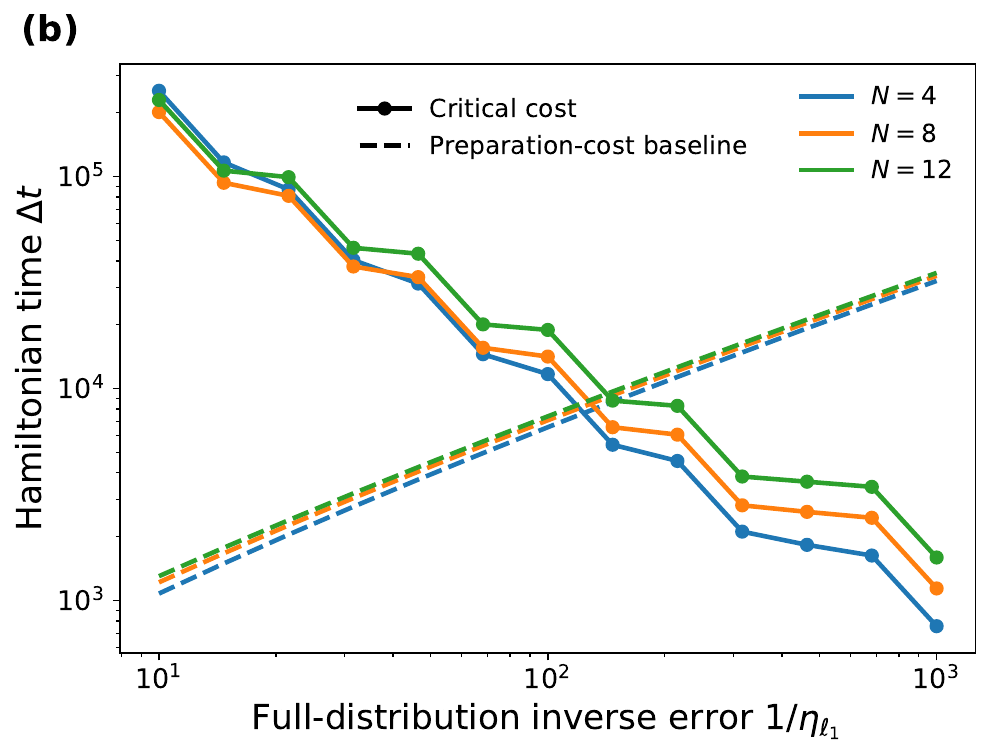}
 \end{tabular}
\caption{\textbf{Precision-dependent crossovers between copy-consuming protocols and catalytic readout.}
(a) Stabilizer-state fidelity prediction on $N=100$ qubits.
Color shows the end-to-end elapsed Hamiltonian-time ratio
$\mathcal T_{\rm copy}/\mathcal T_{\rm cat}$ for simultaneously estimating the fidelities with $M$ random stabilizer states, relative to global-Clifford classical shadows~\cite{huang2020predicting}, as a function of catalog size $M$ and inverse additive-error $1/\varepsilon$.
Catalytic readout is faster when the ratios above one, whereas ratios below one make the copy-consuming approach favorable; the line represents the ratio one.
(b) Full computational-basis probability distribution learning on $N=4,8,12$ qubits, versus inverse $\ell_1$ error $1/\eta_{\ell_1}$.
Colors identify $N$, solid curves with markers show the critical one-copy preparation cost $\Delta t_{\rm crit}$ relative to direct sampling, and dashed curves show the corresponding \textit{very optimistic} preparation-cost baseline
$\Delta t_{\rm prep}$.
For each $N$, catalytic readout is favorable when the dashed curve lies above its solid counterpart.
We refer \green{the Supplementary Materials} for further details.}
 \label{fig:num_exp_STAB_DIST}
\end{figure}

\subsection{From repeated ground-state preparation to reusable coherent readout}

We now see when the catalytic readout reduces the end-to-end elapsed Hamiltonian time relative to task-specialized copy-consuming approaches with repeated preparation for the structured observable sets above.
Let $t_{\rm prep}$ be the elapsed Hamiltonian time required to prepare one ground-state copy, and let $N_{\rm copy}$ be the number of fresh copies required by the comparator.
With the same $P$-way parallelism for both approaches, their total times are $\mathcal T_{\rm cat}=t_{\rm prep}+T_{\rm obs}^{\rm cat}$ and $\mathcal T_{\rm copy}=\lceil N_{\rm copy}/P\rceil t_{\rm prep}$.
The crossover occurs at the critical one-copy preparation cost, expressed in inverse-gap units as
\begin{equation}
    \Delta t_{\rm crit}
    :=
    \frac{\Delta T_{\rm obs}^{\rm cat}}
    {\lceil N_{\rm copy}/P\rceil-1}.
    \label{eq:critical_preparation_cost_main}
\end{equation}
This threshold does not require specifying a ground-state preparation method; catalytic readout is favorable whenever $t_{\rm prep}>t_{\rm crit}$.

We evaluate the constant prefactors and logarithmic overheads of the catalytic readout and set $P=20$ throughout.
To define a \textit{very optimistic} baseline of $t_{\rm prep}$, we focus on dissipative ground-state preparation and combine an extrapolated rapid-mixing fit for the one-dimensional transverse-field Ising model under bulk dissipation with the second-order simulation cost of a single-jump Lindbladian~\cite{PhysRevResearch.6.033147,wzb3-dbg9}.
We set the simulation-cost prefactor to one and omit multiple-jump and filter-implementation overheads, deliberately favoring the copy-consuming approach. 
The following comparison excludes other gate costs and classical post-processing.
We see typical parameter scaling of the critical cost and detailed task setups at \green{the Supplementary Materials}.

We first consider learning the complete low-order fermionic correlation within a selected orbital subspace of a globally number-conserving ground state.
Specifically, our targets are the full selected-orbital 2- and 3-RDMs for $N=200$ fermionic modes at half filling $\eta=100$.
As the comparator, we use particle-number-preserving orbital-rotation
shadows specialized to fermionic RDM estimation~\cite{koizumi2026provably,low2022classical}.
For $m=20$ selected modes, the critical cost lies well below the preparation-cost baseline, with the separation increasing at higher precision (Fig.~\ref{fig:num_exp_rdm}(a)).
For full 2-RDM learning at $m=20$, the end-to-end Hamiltonian-time reduction reaches approximately $2\times10^3$-fold at complex entrywise error $\varepsilon=10^{-2}$ and $3\times10^6$-fold at $\varepsilon=10^{-4}$; large reductions persist across the range of $m=4$--$25$ (Fig.~\ref{fig:num_exp_rdm}(b)).

We also benchmark two tasks for which the selected copy-consuming protocols are
especially effective: stabilizer-state fidelity prediction using
global-Clifford classical shadows~\cite{huang2020predicting} and full
computational-basis distribution learning using direct sampling
(Fig.~\ref{fig:num_exp_STAB_DIST}).
For size-$M$ random stabilizer-state catalogs on $N=100$ qubits, global-Clifford shadows remain competitive at $\varepsilon=10^{-1}$, whereas the catalytic readout yields approximately $10^2$-fold Hamiltonian-time reductions at $\varepsilon=10^{-2}$, with weak catalog-size dependence over the presented range (Fig.~\ref{fig:num_exp_STAB_DIST}(a)).
For full computational-basis distribution learning on 4, 8 and 12 qubits, direct sampling is preferable at lower precision, whereas the catalytic readout becomes favorable at smaller $\ell_1$ error $\eta_{\ell_1}$ (Fig.~\ref{fig:num_exp_STAB_DIST}(b)).
% Additional parameter regions and local-Pauli benchmarks are provided in \green{the Supplementary Materials}.

\subsection{Fundamental copy--dynamics trade-off}
% We next ask whether other parallel measurement protocols can further reduce the scaling of the elapsed Hamiltonian time and thus the critical time.
% To answer this, we below provide fundamental lower bounds on the worst-case elapsed Hamiltonian time.
% They can be obtained from the static-dynamic distinguishability inequality Eq.~\eqref{eq:static_dynamic} under the same lower-bound model in which one may consume all supplied copies.
% % The lower-bound model allows all possible measurement protocols, which may use arbitrary ancilla qubits, arbitrary joint operation across $P$ systems, adaptive controls, controlled Hamiltonian evolution and its time reversal and arbitrary final destructive measurements.

Beyond these task-specific comparisons, we ask how much Hamiltonian dynamics is fundamentally required when ground-state copies are scarce.
The static--dynamic distinguishability inequality in Eq.~\eqref{eq:static_dynamic} yields worst-case lower bounds on the elapsed Hamiltonian time for both multi-observable estimation and full state reconstruction, even when all supplied copies may be consumed; see Methods~\ref{sec:method-keylb} for this reduction.

\subsubsection{Multi-observable estimation}
\begin{thm}[Copy--dynamics trade-off lower bound, informal]\label{thm_main:lowerbound_obsest}
    For a $d$-dimensional Hamiltonian with the unique ground state and gap $\Delta$,
    any $P$-parallel measurement protocol must, in the worst case, use the elapsed controlled Hamiltonian-evolution time 
    \begin{equation}\label{eq_main:obs_LB_full}
    \Omega\!\left(
        \frac{1}{P\Delta}
        \min\!\left\{
            \frac{1}{\varepsilon^2},
            \frac{\sqrt{M}}{\varepsilon},\frac{\sqrt{d}}{\varepsilon}
        \right\}
        \left[
            1-CP\varepsilon^2
        \right]_+
    \right),~~~[x]_{+}:=\max\{0,x\},
    \end{equation}
    to estimate all ground-state expectation values of arbitrary $M$ bounded-norm observables within additive error $\varepsilon$ with high probability. 
    Here, $C$ is a universal constant, and the hard observable set can be chosen mutually commuting.
\end{thm}
The lower bound holds for any unrestricted protocols that may completely consume all supplied copies; no catalytic return is required in the lower-bound derivation.
Nevertheless, our catalytic protocol in Theorem~\ref{thm_main:parallel_catalytic_readout} attains the same leading Hamiltonian-time scaling as the unrestricted lower bound $\Omega(\Delta^{-1}\varepsilon^{-1}{\sqrt M}/{P})$ in the copy-scarce $P\lesssim \sqrt{M}/\varepsilon$, high-precision $\varepsilon\lesssim 1/\sqrt{M}$ regime with $M\leq  d$.
That is, the catalytic readout realizes the worst-case optimal copy--dynamics trade-off, up to poly-logarithmic factors, while leaving the costly ground-state copies available for subsequent use; notably, this catalytic return imposes no leading-order penalty in the optimal trade-off throughout the matched regime.
% Theorem~\ref{thm_main:lowerbound_obsest} follows by embedding the distribution-learning task 
% into a commuting set of target observables; see Methods~\ref{}.

Beyond the matched branch, the full lower-bound expression contains three additional regimes. 
At lower precision, the $1/\varepsilon^{2}$ branch becomes independent of $M$, whereas the $\sqrt{d}/\varepsilon$ branch limits the growth with the number of observables.
The positive-part factor causes the bound to vanish once $P\varepsilon^{2}$ becomes sufficiently large, marking the crossover to a copy-rich regime in which the hard commuting observables can be learned by consuming the supplied copies.
We return to the complete resource landscape after establishing the corresponding limits for learning a full classical description of the ground state.

\subsubsection{Full ground-state reconstruction}

The same absence of a leading-order penalty by catalytic returns remains preserved even for reconstructing a full classical description of ground states.
% Our catalytic readout can reconstruct a complete classical description of the ground state.
% We apply its nearly unbiased version to a complete set of $\mathcal{O}(d^2)$ observables specifying the matrix entries of $\ket{\psi_0}\!\bra{\psi_0}$ and obtain the following upper-bound result.
\begin{thm}
    [Parallel catalytic ground-state tomography, informal]
    \label{thm_main:parallel_catalytic_readout_tom}
    For any $P\geq 1$, there exists a $P$-parallel measurement protocol that with high probability at least $1-\delta$, produce a classical description of $d$-dimensional $\ket{\psi_0}\bra{\psi_0}$ with trace-distance error $\eta$ and returns a joint state $\delta$-close to the initial input copies in trace-distance.
    Its elapsed controlled Hamiltonian-evolution time is
    \begin{equation}
        T^{\rm cat}_{\rm tom}=\mathcal{\widetilde{O}}\left[\frac{1}{\Delta}\left(\frac{d}{P\eta}+1\right)\right].
        \label{eq:tomography_upper}
    \end{equation}
    Here, $\widetilde{\mathcal{O}}$ hides poly-logarithmic factors, including those in $\delta^{-1}$.
    Furthermore, any $P$-parallel measurement protocol, even one that consumes all supplied copies, must in the worst case use the elapsed controlled Hamiltonian-evolution time
    \begin{equation}
        \Omega\!\left(
            \frac{d}{P\Delta\eta}
            \left[1-C\frac{P\eta^2}{d}\right]_{+}
        \right)
        \label{eq:tomography_lower}
    \end{equation}
    to solve the same problem with constant success probability, where $C$ is a universal constant.
\end{thm}

This catalytic protocol is obtained from a nearly unbiased extension of our catalytic observable readout.
Our result has the same $d/\eta$ dependence as the optimal query complexity of pure-state tomography with state-preparation unitaries~\cite{van2023quantum}, while achieving $P$-parallel speedup and catalytic return through Hamiltonian dynamics.
Combining the upper and lower bounds, we establish that the worst-case elapsed Hamiltonian time for learning a full classical description is $\widetilde\Theta(\Delta^{-1}\eta^{-1}{d}/{P})$ in the copy-scarce regime where $P\lesssim d/\eta$; again, this near-optimal copy--dynamics trade-off can be attained catalytically.

The lower bound vanishes due to the positive-part factor when $P\gg d/\eta^2$ and this is consistent with the established result on the conventional pure-state tomography that consumes $\Theta(d/\eta^2)$ copies~\cite{
hayashi1998asymptotic,scharnhorst2025optimal,pelecanos2025mixed}.
% Theorem~\ref{thm_main:parallel_catalytic_readout_tom} follows from the static--dynamic distinguishability inequality because any classical description with small trace distance also determines the computational-basis distribution accurately.

\begin{figure}[tb]
 \centering
 \begin{tabular}{ccc}
 \includegraphics[scale=0.55]{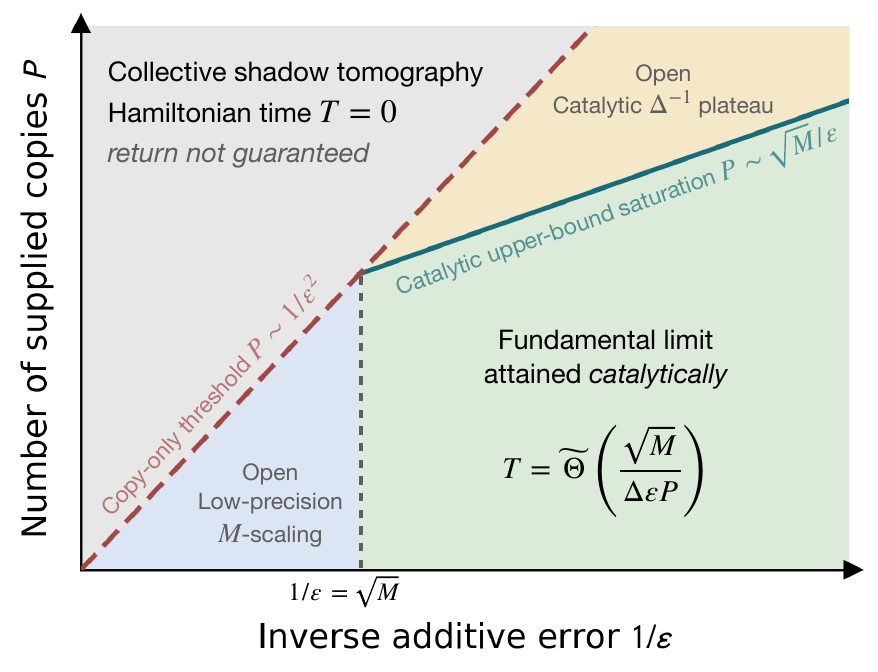}&~~~&\includegraphics[scale=0.55]{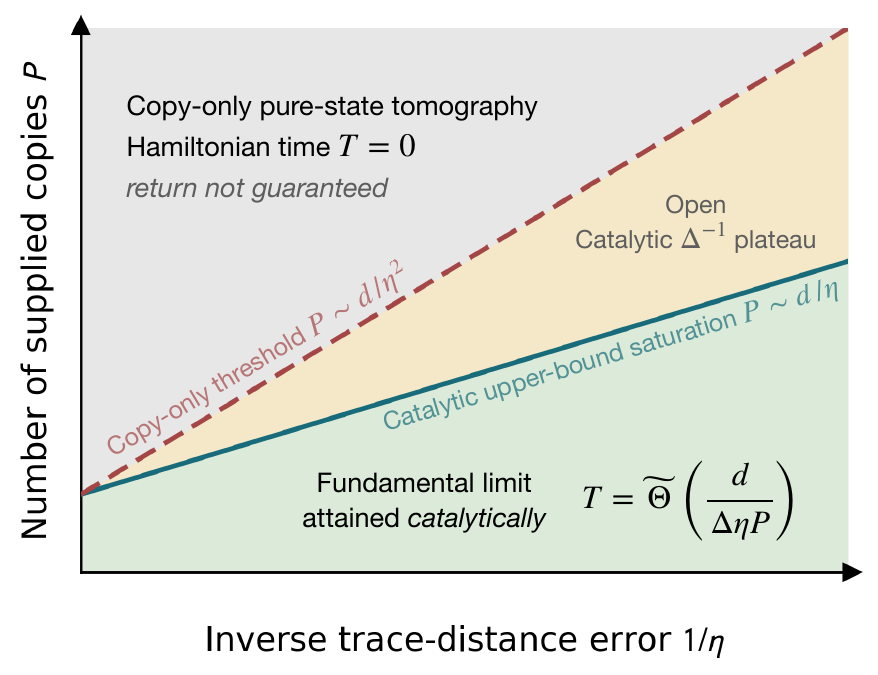}\\
 (a) Multiple observables estimation &~~~&(b) Full ground-state tomography
 \end{tabular}
 \caption{\textbf{Copy--dynamics landscape for ground-state learning.}
 Our catalytic protocols are available for all regions and return all supplied copies nearly unchanged. 
 Green regions indicate where our catalytic readout matches lower bounds for arbitrary protocols, including those that may consume all supplied copies, up to poly-logarithmic factors.
 Beige regions denote the unmatched catalytic $1/\Delta$ plateau from our catalytic upper bounds, and blue region denotes the unmatched $M$-scaling in the low-precision estimation.
 Gray regions represent the conventional copy-only protocols: no elapsed Hamiltonian time is required, but catalytic return is not necessarily guaranteed.
 (a) Multiple observables estimation in the $(P,1/\varepsilon)$ plane at fixed $M\lesssim d$ and $\Delta$. 
 The collective shadow-tomography protocol for the gray region is information-theoretic, and its computational efficiency remains open~\cite{jeronimo2026dimension}.
 (b) Full ground-state tomography in the $(P,1/\eta)$ plane at fixed $d$ and $\Delta$.
 Constants and poly-logarithmic factors are suppressed; detailed comparisons are given in Methods~\ref{sec:landscape-method}.
 }
 \label{fig:phase_diagram}
\end{figure}

\section{Outlook}

We conclude by examining the copy--dynamics landscape in Fig.~\ref{fig:phase_diagram}, which brings together our catalytic upper bounds, unrestricted lower bounds, and known copy-only protocols.
Our catalytic protocols apply throughout this landscape, with the green matched regimes, up to poly-logarithmic factors, established above.
Beyond the hard commuting observables used in our lower bound, recent collective shadow tomography can estimate arbitrary known sets of bounded observables using 
a joint measurement on $\widetilde{\mathcal{O}}(\varepsilon^{-2})$ copies (hiding poly-logarithmic dependence on $M$)~\cite{jeronimo2026dimension}.
Together with the copy-only pure-state tomography~\cite{hayashi1998asymptotic,scharnhorst2025optimal,pelecanos2025mixed}, this identifies the gray regions where learning requires no Hamiltonian evolution, $T=0$, but catalytic return is not guaranteed.
Between the matched and copy-only regimes, the present catalytic upper bounds saturate at the elapsed Hamiltonian time $\widetilde {\mathcal{O}}(1/\Delta)$, whereas the unrestricted lower bounds fall below this scaling. 
These intermediate beige regions span $\sqrt M/\varepsilon\lesssim P\lesssim\varepsilon^{-2}$ for high-precision observable estimation and $d/\eta\lesssim P\lesssim d/\eta^{2}$ for tomography, ignoring constants and logarithmic factors.
A different gap appears on the low-precision side of Fig.~\ref{fig:phase_diagram}(a): the precision-limited branch in Eq.~\eqref{eq_main:obs_LB_full} is independent of $M$, while the present worst-case catalytic upper bound retains a square-root dependence on $M$ (blue region).
These unresolved regions pose a common question: what is the complete copy--dynamics trade-off with and without catalytic return? 
Answering this requires improved protocols or stronger lower bounds.
For the blue region, a promising direction is to develop conceptually new approaches that integrates ideas from shadow tomography and Heisenberg-limited coherent readout.

Our numerical benchmarks also motivate more refined comparisons with copy-consuming approaches. 
The present comparisons use task-specialized single-copy measurements; 
% (see \green{Methods} for comparisons to few-copy measurement protocols \red{[check]})
the collective shadow-tomography bound in Fig.~\ref{fig:phase_diagram}(a) is information-theoretic and does not itself provide a gate-efficient implementation with joint measurements on fixed $P$ copies~\cite{jeronimo2026dimension}.
Quantifying the benefits of joint measurements on up to fixed $P$ copies with gate-efficient implementations would sharpen those task-specific comparisons under the same $P$-way parallelism~\cite{chen2026instance}.
Beyond the dissipative-preparation baseline considered here, which does not require any state with high overlap with the target state, another direction is to evaluate preparation costs when trial states with moderate ground-state overlap are available~\cite{ge2019faster,lin2020near,dong2022ground}.
Comparing these costs, including their dependence on initial overlap and required preparation accuracy, with the critical preparation times would clarify when catalytic reuse remains favorable against stronger copy-consuming strategies.

% Further benchmarks in concrete physical settings would clarify the benefits of catalytic readout, with tensor-network states, purified Gibbs states, and steady states of suitably engineered dynamics as candidate targets. 
% Identifying suitable gapped parent Hamiltonians for these targets would provide a route to probing such states without relying on potentially costly repeated state preparations.
Further investigation in physically relevant settings would clarify the benefits of dynamics-assisted readout, with tensor-network states, (purified) Gibbs states, and steady states of suitably engineered dynamics as candidate targets.
Identifying efficient dynamical access such as Hamiltonians beyond state copies for these targets would provide a route to probing such states without relying on potentially costly repeated state preparations.
Related progress includes efficient observable estimation for Gibbs states using Hamiltonian dynamics~\cite{chen2026efficient} and dissipative adiabatic measurements of steady states~\cite{zhang2020dissipative}.
Since our present benchmarks focus on the elapsed Hamiltonian time, further work should also quantify gate and ancilla counts and optimize the circuits for observable block encoding and other protocol components.
For distributed implementation, the GHZ states used in our protocol are problem-independent quantum resources: their form does not depend on the target Hamiltonian or observable set. 
They can therefore be prepared and shared among quantum devices in advance, after which readout requires only local operations and classical post-processing of the measurement data (Fig.~\ref{fig:catalytic_paralle_readout}).
This separation motivates distributed architectures in which GHZ generation and distribution are performed separately from problem-specific catalytic readout.

\section{Methods}

\subsection{$P$-parallel access model}
\label{sec:methods-accessmodel}

We assume that the operator norm $\|H\|$ of a system Hamiltonian $H$ or its upper bound is known; we can replace $\|H\|$ with its upper bound.
The excited-state energies of $H$ are assumed to be separated from the ground-state energy by a spectral gap $\Delta>0$. 
The target observables are assumed to be accessed by controllable block-encoding unitaries.
This is compatible with various standard description of observables, e.g., sparse matrices and linear combinations of Paulis~\cite{Low2019hamiltonian,10.1145/3313276.3316366}.

The access model in this paper, we call it \textit{$P$-parallel access model}, is defined as follows.
% An algorithm
A measurement protocol is supplied with $P$ copies of $\ket{\psi_0}$ and controlled access to the Hamiltonian
time evolution $e^{-itH}$ for arbitrary $t\in\mathbb{R}$.
% Note that this includes the time reversal $e^{+iHt}$.
The protocol knows that $\ket{\psi_0}$ is the unique ground state of the given black-box Hamiltonian $H$ with a known spectral gap $\Delta$, but 
% it receives no additional information on the inputs other than the copies and the  Hamiltonian evolution oracle.
it receives no instance-specific classical description of $H$, its eigenbasis, or its ground state.
That is, the only access to the black-box Hamiltonian $H$ and its ground state $\ket{\psi_0}$ is through the controlled time-evolution oracle and the $P$ supplied ground-state copies.
We assume no prior knowledge of a procedure for preparing a quantum state having sufficiently large overlap with the ground state.

In Fig.~\ref{fig:overview}(a), 
% The $P$ system registers can be evolved simultaneously.
% Specifically, 
one parallel time evolution segment of time $t_l$ applies a controlled or uncontrolled evolution of time $\pm t_l$ to each of the $P$ system registers in parallel: the corresponding entire Hamiltonian is given by
\begin{equation}
    \sum_{a=1}^P C_{a}\otimes H^{(a)},
\end{equation}
where $H^{(a)}$ is the target system Hamiltonian $H$ with sign $\pm 1$ acting on the $a$-th system register, and $C_{a}$ is either zero, the one-qubit identity, or a one-qubit control projector $|1\rangle\langle 1|$ acting on an ancilla qubit.
To ensure the parallel structure, $C_{a}$ and $C_{a'}$ act on distinct ancilla qubits for all $a\neq a'$.
Between the time-evolution segments, the protocol may apply arbitrary $H$-independent quantum operations, including joint operations across all systems, arbitrary ancillas, intermediate measurements and classical feed-forward.
Since the number of segments is unrestricted (but finite) and local evolution can be switched on and off, any finite adaptive control are included in our model.
With the above definition, we define the elapsed controlled Hamiltonian-evolution time as in Eq.~\eqref{eq:elpsed_htime_main}.
% Suppose that an algorithm uses $L$ sequential such segments with times
% $t_1,\ldots,t_L\in\mathbb{R}$. We measure its cost by the
% \textit{sequential total evolution time}
% \begin{equation}
%     T
%     :=
%     \sum_{\ell=1}^{L}|t_\ell|.
%     \label{eq:sequential-total-evolution-time}
% \end{equation}
% A parallel time evolution segment of time $t_\ell$
% contributes $|t_\ell|$, rather than $P|t_\ell|$, to
% Eq.~\eqref{eq:sequential-total-evolution-time}. 
% Thus, $T$
% measures the elapsed Hamiltonian evolution time when the
% $P$ system registers are operated in parallel.
If one instead counts the evolution time summed over all $P$ registers, the resulting cost is at most $PT$ in general.

Our access model differs from that for query-efficient observable estimation protocols based on (standard) quantum amplitude estimation (QAE)~\cite{brassard2000quantum,PhysRevA.75.012328,PhysRevA.102.022408,wrkp-qd33} and quantum gradient estimation (QGE)~\cite{jordangradest,gilyen2019optimizing,van2021quantum,PhysRevLett.129.240501,van2023quantum,PRXQuantum.6.020308,4c6g-zx6c}.
These protocols are formulated with access to a target-state preparation unitary and its inverse (i.e., the ground state $|\psi_0\rangle$ in our problem), while our model requires neither.
The supplied copies in our model may therefore be obtained through dissipative or measurement-and-feedforward preparation~\cite{PhysRevResearch.6.033147,wzb3-dbg9,mao2023measurement}, without requiring coherent access to the preparation process or its inverse.

There are two variants of QAE we need to mention here particularly.
Nondestructive QAE~\cite{harrow2020adaptive,rall2021faster,rall2023amplitude,cornelissen2023sublinear} provides a related form of catalytic readout.
Given one copy of a target state and reflections about this state and a target subspace, it estimates the specified single amplitude and restores the input state with high probability.
The required state reflection $2\ket{\psi_0}\bra{\psi_0}-\bm{1}$ may be obtained by time-$\pi$ evolution under a Hamiltonian $\bm{1}-\ket{\psi_0}\bra{\psi_0}$; this access model is a special case considered in this paper.
Parallel amplitude estimation (PAE)~\cite{oshio2025near} implements single-signal phase shifts for amplitude (more precisely, probability on the specified subspace) using quantum signal processing of a standard amplitude-estimation oracle and its inverse, and combines them with parallel GHZ probes to achieve near-Heisenberg scaling in the total number of queries.
Our readout shares this parallel sensing architecture and integrates catalytic return, $P$-parallel readout, and a quantum speedup in $M$ to attain the near-optimal copy--dynamics trade-off (see Methods~\ref{sec:method_joint_kickback} and~\ref{sec:method_multireadout}).

\subsection{Key ideas behind the lower bounds}
\label{sec:method-keylb}

To prove Eq.~\eqref{eq:static_dynamic}, we construct a Hamiltonian family whose hidden bits are difficult to distinguish both from the supplied ground-state copies and through controlled dynamics even with arbitrary joint and adaptive processing. 
Let $m=\lfloor(d-1)/2\rfloor$ and $\bm{b}\in\{0,1\}^{m}$. 
Starting from the uniform distribution on $2m$ outcomes, we encode each $r$-th bit in a weak relative probability bias proportional to $\omega$ between the index pair $(r,r+m)$ for $r\in [m]$, obtaining a probability distribution $p_{\bm{b}}$.
For its amplitude-encoded quantum state $|p_{\bm{b}}\rangle$, we define
\begin{equation}
H_{\bm{b}}:=2\Delta\left[
        \bm{1}
        -\frac{1}{2}
        \left(
            \ket{0}\!\bra{p_{\bm{b}}}
            +
            \ket{p_{\bm{b}}}\!\bra{0}
        \right)
    \right],
\end{equation}
where the reference state $|0\rangle$ is orthogonal to the $2m$ basis states.
Each $H_{\bm{b}}$ has the unique ground state $(|0\rangle+|p_{\bm{b}}\rangle)/\sqrt{2}$ and spectral gap $\Delta$.

Now, we introduce a measure of distinguishability inspired by quantum adversary methods~\cite{ambainis2000quantum,hoyer2005lower,de2019quantum}.
Since the structure of $H_{\bm{b}}$ enables an $\omega$-accurate distribution estimate to decode a constant fraction of hidden bits only through classical post-processing, it suffices to focus on the hidden-bit decoding task.
Let $|\Phi_{\bm{b}}^{(\ell)}\rangle$ denote an arbitrary protocol state including all ancillas within our access model after the $\ell$-th parallel time-evolution segment (see Fig.~\ref{fig:overview}(a)).
We use the adversary progress quantity
\begin{equation}
\mathcal{D}_{\ell}:=1-\mathbb E_{\bm{b},r}
\operatorname{Re}\langle\Phi_{\bm{b}}^{(\ell)}|
\Phi_{\bm{b}^{(r)}}^{(\ell)}\rangle,
\end{equation}
where $\bm{b}^{(r)}$ differs from $\bm{b}$ only in bit $r$, and the average is uniform over $\bm{ b}$ and $r\in [m]$. 
This quantity measures the average distinguishability of the hard input pairs after the $\ell$-th segment.
If the protocol decodes a constant fraction of hidden bits with high probability, then an information-theoretic argument shows that the final progress $\mathcal{D}_{L}$ must be above a threshold even when the initial copies are consumed.

We then bound the initial progress and its change under controlled dynamics.
The weak probability biases make ground states with neighboring bits close, giving $\mathcal{D}_{0}=\mathcal{O}(P\omega^2/d)$ for their $P$ copies.
The difference between neighboring Hamiltonians appears only within the two-dimensional subspace spanned by $|p_{\bm{b}}\rangle-|p_{\bm{b}^{(r)}}\rangle$ and the reference state; moreover, the subspaces supporting the difference vectors are mutually orthogonal for different bit positions $r$.  
Averaging over neighboring inputs tightly quantifies the dynamical change by exploiting these hardness structures on $H_{\bm{b}}$, giving the maximal possible change from $\mathcal{D}_0$ to $\mathcal{D}_L$ by $\mathcal{O}(P\omega\Delta T/d)$, even under arbitrary joint and adaptive processing.
Combining the required threshold for the final progress with these bounds on the initial contribution and its dynamical change yields the static--dynamic distinguishability inequality.

To derive the multi-observable lower bound, we consider the above hard Hamiltonian family in dimension $s$ and embed it into a $d$-dimensional system, where we can take $s=\Theta(\min\{d,M,\varepsilon^{-2}\})$.
On this system, we construct $M$ known, mutually commuting diagonal observables of norm at most one, fixed independently of the hidden bits.
Their expectation values in the ground state of the embedded Hamiltonian encode the computational-basis distribution of the original $s$-dimensional ground state through a Hadamard transformation~\cite{van2021quantum}.
Thus, any protocol that estimates these ground-state expectation values with simultaneous additive error $\varepsilon$ yields, by classical post-processing, an estimate of the original distribution with $\ell_1$ error at most $\omega=\sqrt{s}\varepsilon$. 
Substituting this into Eq.~\eqref{eq:static_dynamic} with dimension $s$ yields Eq.~\eqref{eq_main:obs_LB_full}, with the initial-copy contribution $P\omega^2/s=P\varepsilon^2$.

For tomography, the diagonal entries of a classical density-matrix estimate give an estimate of the computational-basis distribution. 
Contractivity of trace distance under a measurement channel implies that a density-matrix estimate with trace-distance error $\eta$ gives a distribution estimate with $\ell_1$ error at most $2\eta$, so applying Eq.~\eqref{eq:static_dynamic} with $\omega=2\eta$ yields Eq.~\eqref{eq:tomography_lower}.

Both reductions preserve the spectral gap, copy count, and elapsed Hamiltonian time in the $P$-parallel access model, and the resulting bounds apply even when all supplied copies may be consumed. 
The positive-part factors $[\cdot]_+$ come from subtracting the upper bound on the initial-copy contribution from the required decoding threshold. 
When this difference is nonpositive, the inequality imposes no positive lower bound on the worst-case evolution time.
The complete reductions are given in \green{the Supplementary Materials}.

\subsection{Joint multidimensional phase kickback}
\label{sec:method_joint_kickback}

The joint phase-kickback operation $V_{\rm phase}$ in Eq.~\eqref{eq:catalytic_multisignal_map} aims to mimic the ideal local action
\begin{equation}\label{eq:phasekickback-method}
    |\boldsymbol b\rangle|\psi_0\rangle
\longmapsto
e^{i\tau\boldsymbol b\cdot\boldsymbol\mu}
|\boldsymbol b\rangle|\psi_0\rangle,
\qquad \boldsymbol b\in\{0,1\}^{M}.
\end{equation}
For $M=1$, this is ordinary phase kickback for eigenphase $e^{i\tau b\mu_1}$ via a controlled unitary that has $\ket{\psi_0}$ as the corresponding eigenstate.
For general $M$, the implementation therefore requires a $\boldsymbol b$-controlled unitary whose ground-state eigenphase encodes all $M$ expectation values jointly as in $e^{i\tau \bm{b}\cdot \bm{\mu}}$.
Applying the ideal operation in parallel on the $P$ devices gives Eq.~\eqref{eq:catalytic_multisignal_map} with no approximation error.

To approximate this, we combine joint multi-observable encoding~\cite{van2023quantum,PRXQuantum.6.020308,4c6g-zx6c} with energy filtering through an operator Fourier transformation (OFT)~\cite{chen2025efficient,PhysRevResearch.6.033147,chen2025catalytic}.
For a Hermitian operator $A$, the OFT circuit constructs a block encoding of a discretized and truncated approximation to 
\begin{equation}
    \hat{A}_f:=\frac{1}{\sqrt{2\pi}}\int_{-\infty}^{\infty} {\rm d}t~f(t)e^{iHt}Ae^{-iHt}=\sum_{i,j=0}\ket{\psi_i}\bra{\psi_i}A\ket{\psi_j}\bra{\psi_j}\hat{F}(E_j-E_i),
\end{equation}
where $\ket{\psi_i}$ is the eigenstate of the system Hamiltonian $H$ with eigenenergy $E_i$, and $\hat{F}$ is the Fourier transformation of $f$.
One can design a normalized filter function $f(t)$ that preserves zero energy difference and suppresses energy differences of magnitude at least $\Delta$, yielding
\begin{equation}
\hat{A}_f\simeq \langle\psi_0|A|\psi_0\rangle \Pi_0+\hat{A}_f^{\perp},
~~~
\Pi_0=|\psi_0\rangle\langle\psi_0|,
~~~
\hat{A}_f^{\perp}\Pi_0=\Pi_0\hat{A}_f^{\perp}=0.
\end{equation}
Filtering thus makes the supplied ground state an approximate eigenstate with eigenvalue $\langle\psi_0|A|\psi_0\rangle$.
Each query to the filtered block encoding uses one query to $A$ and controlled evolution under $H$ on the time scale $\Delta^{-1}$, up to logarithmic precision factors.

To obtain the improvement with $M$ in Eq.~\eqref{eq:thm_main:parallel_catalytic_readout}, we adapt joint multi-observable encoding developed for applications of QGE~\cite{van2023quantum,PRXQuantum.6.020308,4c6g-zx6c}.
It provides a $\bm{b}$-controlled block-encoding unitary of a Hermitian operator $C(\bm{b})$ that approximate $\sum_j(b_j-1/2)O_j$, up to a normalization factor $1/2\sigma$ governed by $\Gamma_{\mathrm{obs}}$, on all but a small fraction of control bit strings.
This uses no time evolution on the system Hamiltonian.
Allowing these exceptional strings enables matrix-norm concentration and provides such a normalization factor improved from naive $M$.
We apply the OFT coherently to $C(\boldsymbol b)$ and then use block-Hamiltonian simulation~\cite{Low2019hamiltonian,10.1145/3313276.3316366} of the resulting Hamiltonian $\hat{C}_f(\boldsymbol b)$ with time $2\sigma\tau$; this controlled unitary defines $V_{\rm phase}$.
On the well-approximated control strings, its action with the ground-state input approximates Eq.~\eqref{eq:phasekickback-method} up to an global phase.

We prove that the parallel application of the above operations provides the target shift in Eq.~\eqref{eq:catalytic_multisignal_map} on the $P$-qubit GHZ states, while retaining the reduced normalization and logarithmic Hamiltonian-time overhead in $P$.
The collective phase accumulation in GHZ states is therefore compatible with the above phase kickback construction even with the exceptional control strings; the overall coherent error can be efficiently controlled without losing the parallel advantage.
The QGE papers~\cite{van2023quantum,PRXQuantum.6.020308,4c6g-zx6c} provide a related error analysis, but extending it to our setting is not straightforward because their initial superposition over control bit strings is separable, whereas our protocol uses GHZ probes entangled across devices.
The complete encoding, filtering construction, and error analysis are given in \green{the Supplementary Materials}.

\subsection{Catalytic parallel multi-signal readout}
\label{sec:method_multireadout}

Since the above phase kickback leaves the ground-state registers nearly unchanged, we reuse the same $P$ registers to generate all phase-shifted GHZ states required for readout, without additional ground-state copies.
The joint operation in Eq.~\eqref{eq:catalytic_multisignal_map} generates $M$ phase-shifted GHZ states simultaneously.
Single-qubit measurements on these states in appropriate $X$ or $Y$ base, followed by classical post-processing, yield $M$ binary outcomes, whose statistics encode the sine or cosine of the corresponding collective phases $P\tau\mu_j$.
These measurements are performed locally at each device, so the GHZ qubits need not be brought together again after distribution.

To resolve phase periodicity, we process the measurement data for each signal component using a multistage procedure developed for single-signal phase estimation~\cite{PhysRevA.92.062315,PhysRevA.102.042613}.
We use $K=\mathcal O(\log(1/\varepsilon))$ stages with $\tau_k=2^{k-1}/P$, repeat each of $X$ and $Y$ measurements $v_k=\mathcal O(K-k+1)$ times, and estimate the signals from coarse to fine scales.
Although reusing the same ground-state registers produce a correlated measurement outcomes across stages and signal components, choosing the approximation error in Eq.~\eqref{eq:catalytic_multisignal_map} sufficiently small keeps the full data distribution close in total variation distance to the ideal product distribution. 
With this error control, repeating the entire procedure $\Upsilon_{\rm med}=\mathcal O(\log(M/\delta))$ times and taking coordinate-wise medians yields simultaneous $\varepsilon$-accurate estimates with probability at least $1-\delta$.
Summing the elapsed Hamiltonian time of the parallel phase-kickback calls gives
\begin{equation}
T_{\rm obs}^{\rm cat}
=\sum_{k=1}^{K}2v_k\Upsilon_{\rm med}\,
\widetilde{\mathcal O}\!\left(\frac{\sigma|\tau_k|+1}{\Delta}\right)
=\widetilde{\mathcal O}\!\left[
\frac{1}{\Delta}\left(\frac{\Gamma_{\rm obs}}{P\varepsilon}+1\right)
\right],
\end{equation}
where $\sigma$ is derived in the previous subsection, thereby recovering Eq.~\eqref{eq:thm_main:parallel_catalytic_readout}.
The additive term comes from the optimal block-Hamiltonian simulation whose query complexity is the sum of evolution time and simulation error contributions.

The copy-return guarantee in Theorems~\ref{thm_main:parallel_catalytic_readout} and~\ref{thm_main:parallel_catalytic_readout_tom} applies to the joint post-measurement state of all $P$ ground-state registers conditioned on the complete record of GHZ-register measurements.
For every measurement record outside a subset with probability at most $\delta$, this conditional state is within trace distance $\delta$ of the initial $P$-copy state $\ket{\psi_0}^{\otimes P}$.
No post-selection on particular measurement outcomes is required. 
The return accuracy can be specified separately from $\varepsilon$ and $\delta$, and contributes only logarithmic factors to the elapsed Hamiltonian time.

In the application to ground-state tomography, we use a common error metric, trace-distance error.
A naive conversion to this metric from the above element-wise additive error gives a loose dimension dependence~\cite{van2023quantum}.
Following the approach of Ref.~\cite{van2023quantum}, we combine nearly unbiased estimation with random-matrix concentration.
We obtain the required nearly-unbiased estimates from our catalytic readout by further introducing random known-phase shifts and symmetrizing the classical post-processing.
The resulting estimator is close in total variation distance to an auxiliary estimator with independent coordinates, bounded errors and controllably small biases, thereby allowing random-matrix concentration to give a sharp error bound and hence Eq.~\eqref{eq:tomography_upper}.
The complete readout construction and error analysis are given in \green{the Supplementary Materials}.

% \subsection{Numerical simulation}
% All comparisons consider only sequential total controlled Hamiltonian-evolution
% time.
% Encoding observables, GHZ preparation, measurement circuit overhead in copy-consuming protocols, and classical
% post-processing are not included here.
% \red{[$\Delta T_{\rm obs}^{\rm cat}$ is independent of $\Delta$][$P$ scaling, now we fixed 20][how to derive Gamma in each observable set; supple?][how to evaluate catalytic readout; supple?][We see typical parameter scaling of the critical cost and detailed task setups at \green{Methods and the Supplementary Materials}.][how to evaluate the comparators][why is the comparator chosen][add few-copy joint measurements?]}

\subsection{Regions in the copy--dynamics landscape}
\label{sec:landscape-method}

We describe the asymptotic regimes shown in Fig.~\ref{fig:phase_diagram}.
Hereafter, UB denotes the worst-case upper bound of elapsed Hamiltonian time achieved by our catalytic protocol, whereas LB denotes the worst-case lower bound for arbitrary $P$-parallel protocols, which may consume all supplied copies.
We remark that our catalytic protocols are available for all regions and return all supplied copies nearly unchanged.
All nonzero upper and lower bounds contain the same overall inverse-gap factor $1/\Delta$; 
the spectral gap therefore sets the Hamiltonian-time scale.
In the gray regions, a separate copy-only protocol attains zero elapsed Hamiltonian time, but does not guarantee the catalytic return.

\subsubsection{Multi-observable estimation}

We focus on a typical case $M\lesssim d$.
In the green region where $\varepsilon\lesssim M^{-1/2}$ and $P\lesssim {\sqrt{M}}/{\varepsilon}$,
the upper and lower bounds are
$T_{\rm obs,UB}^{\rm cat}=\widetilde{\mathcal{O}}\!\left(
\frac{\sqrt{M}}{P\varepsilon\Delta}
\right)$ and $T_{\rm obs,LB}=\Omega\!\left(
\frac{\sqrt{M}}{P\varepsilon\Delta}
\right)$.
Thus, the fundamental worst-case scaling is attained catalytically, up to poly-logarithmic factors. 
In the beige region where
$\varepsilon\lesssim M^{-1/2}$ and $
{\sqrt{M}}/{\varepsilon}
\lesssim P
\lesssim
{1}/{\varepsilon^{2}}$,
the bounds become
$T_{\rm obs,UB}^{\rm cat}=\widetilde{\mathcal{O}}\left(\frac{1}\Delta\right)$ and
$T_{\rm obs,LB}
=
\Omega\!\left(
\frac{\sqrt{M}}{P\varepsilon\Delta}
\right)$.
The catalytic upper bound has entered its inverse-gap plateau, whereas the unrestricted lower bound continues to decrease with $P$.
It remains open whether the catalytic $1/\Delta$ plateau is necessary or the decreasing lower-bound branch can be attained by an unrestricted dynamics-assisted protocol.
In the blue region where $\varepsilon\gtrsim M^{-1/2}$ and $ P\lesssim {1}/{\varepsilon^{2}}$,
the bounds are
$T_{\rm obs,UB}^{\rm cat}=\widetilde{\mathcal{O}}\!\left(
\frac{\sqrt{M}}{P\varepsilon\Delta}
\right)$ and $
T_{\rm obs,LB}
=
\Omega\!\left(
\frac{1}{P\varepsilon^{2}\Delta}
\right)$.
Here the lower bound is independent of $M$, leaving open whether the $\sqrt{M}$ dependence of the present catalytic upper bound can be
removed or a stronger $M$-dependent lower bound is required.

\subsubsection{Full ground-state tomography}
In the green region where $P\lesssim {d}/{\eta}$,
the bounds are
$T_{\rm tom}^{\rm cat}
=\widetilde{\mathcal{O}}\!\left(
\frac{d}{P\eta\Delta}
\right)$ and 
$
T_{\rm tom}
=
\Omega\!\left(
\frac{d}{P\eta\Delta}
\right)$.
The fundamental worst-case tomography scaling is therefore attained catalytically.
In the beige region where
${d}/{\eta}
\lesssim P
\lesssim
{d}/{\eta^{2}}$, the bounds become
$T_{\rm tom}^{\rm cat}
=
\widetilde{\mathcal{O}}\left(\frac{1}\Delta\right)$ and $T_{\rm tom}
=
\Omega\!\left(
\frac{d}{P\eta\Delta}
\right)$.
As in multi-observable estimation, whether the catalytic $1/\Delta$ plateau is necessary remains unresolved.

\section*{Acknowledgements}
We thank the fruitful comments from Hiroyuki Harada, Leonard Logaric, and Yuki Koizumi. 
K.W. is supported by JSPS KAKENHI Grant Number JP24KJ1963 and JST ASPIRE Grant Number JPMJAP2316.
N.Y. is supported by JST Grant Number JPMJPF2221, JST CREST Grant Number JPMJCR23I4, IBM Quantum, Google Quantum AI, JST ASPIRE Grant Number JPMJAP2316, JST ERATO Grant Number JPMJER2302, JST [Moonshot R\&D] [Grant Number JPMJMS256J], and Institute of AI and Beyond of the University of Tokyo.

\clearpage

% \appendix

\begin{center}
	\Large
	\textbf{Supplementary Materials for\\
    ``Copy-scarce learning of ground states''
    }\\[0.5em]
    \large Kaito Wada, Nobuyuki Yoshioka
\end{center}

\addtocontents{toc}{\protect\setcounter{tocdepth}{2}}
\tableofcontents

\beginsupplement

\section{Notation and preliminary}

The notation $\widetilde{\mathcal{O}}$ suppresses poly-logarithmic factors only in the parameters specified in each statement.
For a positive integer $m$, we write $[m]:=\{1,2,\ldots,m\}$.
For a finite-dimensional vector $x$, we use the following norms
\begin{equation}
    \|x\|_1 := \sum_j |x_j|,
    \qquad
    \|x\|_2 := \left(\sum_j |x_j|^2\right)^{1/2},
    \qquad
    \|x\|_\infty := \max_j |x_j|.
\end{equation}
For a finite-dimensional operator $A$, we denote its operator norm and trace norm by
\begin{equation}
    \|A\| := \sup_{\|x\|_2=1}\|Ax\|_2,
    \qquad
    \|A\|_1 := \operatorname{tr}\sqrt{A^\dagger A},
\end{equation}
respectively.
For $x,y\in\mathbb{R}$, we define the circular distance
\begin{equation}
    d_{\mathbb{T}}(x,y):= \min_{q\in\mathbb{Z}} |x-y-2\pi q|.
\end{equation}
It satisfies the reflection and translational symmetries
\begin{equation}
    d_{\mathbb{T}}(-x,-y)=d_{\mathbb{T}}(x,y),
    \qquad
    d_{\mathbb{T}}(x+c,y+c)=d_{\mathbb{T}}(x,y)
\end{equation}
for any $x,y,c\in\mathbb{R}$.

\begin{dfn}
    [Block encoding~\cite{Low2019hamiltonian,10.1145/3313276.3316366}]\label{dfn:block_encoding}
    For positive values $\alpha,\varepsilon$ and a non-negative integer $a$, we say that an $(n+a)$-qubit unitary $U$ is an ($\alpha,a,\varepsilon$)-block-encoding of an $n$-qubit operator $A$, if 
    \begin{equation}
    \|A-\alpha (\bra{0}^{\otimes a}\otimes \bm{1})U (\ket{0}^{\otimes a}\otimes \bm{1})\|\leq \varepsilon.
    \end{equation}
    For simplicity, we shorten the perfect (i.e., $\alpha=1$ and $\varepsilon=0$) block encoding of $A$ as $a$-block-encoding of $A$. 
    When only the precision is relevant, we refer to a $(1,a,\varepsilon)$-block-encoding as an $\varepsilon$-precise block encoding.
\end{dfn}
For a unitary target operator, the block-encoding error also bounds the Euclidean error of the full output state.
Specifically, if $V$ is a $(1,a,\varepsilon)$-block-encoding
of an $n$-qubit unitary $U$, then, for every normalized $n$-qubit state $|\psi\rangle$,
\begin{equation}
    \left\|
        V\bigl(|0\rangle^{\otimes a}\otimes|\psi\rangle\bigr)
        -
        |0\rangle^{\otimes a}\otimes U|\psi\rangle
    \right\|
    \le \sqrt{2\varepsilon}.
\end{equation}
Indeed, writing $B:=(\langle 0|^{\otimes a}\otimes I)V(|0\rangle^{\otimes a}\otimes I)$, the squared norm on the left-hand side is
\begin{equation}
    2-2\operatorname{Re}
    \langle\psi|U^\dagger B|\psi\rangle
    \le 2\|B-U\|
    \le 2\varepsilon.
\end{equation}

\section{Lower bounds on elapsed Hamiltonian time}

In this section, we derive lower bounds on the elapsed Hamiltonian time within the $P$-parallel access model (see \green{Methods} in the main text)
% Section~\ref{sec_sm:p_parallel_access_model}
for the two target tasks: ground-state multiple observables estimation and ground-state tomography.
To this end, we first introduce a hard Hamiltonian family and list its properties.
Then, using this family, we derive 
the static--dynamic distinguishability inequality
% a lower bound 
in ground-state distribution estimation in Theorem~\ref{thm:ground-state-distribution}. 
This inequality can be reduced to lower bounds for those two target tasks, namely Theorems~\ref{thm:simultaneous-expectation-lower-bound} and~\ref{thm:ground-state-tomography-lower-bound}.

\subsection{Hard Hamiltonian family and its properties}
The hard-to-distinguish family of $d$-dimensional Hamiltonians is explicitly provided in the following lemma.
The following encoding of hidden bits through small biases in paired probabilities is inspired by Lemma~11 in Ref.~\cite{van2021quantum}.

\begin{lem}[Hard Hamiltonian family]
\label{lem:hard-hamiltonian}
Let $d\geq 4$ and let $m:=\left\lfloor({d-1})/{2}\right\rfloor$. 
% be defined by Eq.~\eqref{eq:hard-family-index}.
Fix an integer $k\geq 4$.
For each $b=(b_1,\ldots,b_m)\in\{0,1\}^m$, define a probability distribution $p_b$
on $[2m]$ by
\begin{equation}
    p_b(r):=\frac{k+2b_r}{2mk},
    \qquad
    p_b(r+m):=\frac{k-2b_r}{2mk},
    \qquad
    r\in[m],
    \label{eq:hard-distribution}
\end{equation}
and define the corresponding amplitude-encoded state by
$\ket{p_b}:=\sum_{j=1}^{2m}\sqrt{p_b(j)}\ket{j}$.
Define the $d$-dimensional Hamiltonian
\begin{equation}
    H_b:=2\Delta\left[
        \bm{1}
        -\frac{1}{2}
        \left(
            \ket{0}\!\bra{p_b}
            +
            \ket{p_b}\!\bra{0}
        \right)
    \right]
    \label{eq:hard-hamiltonian}
\end{equation}
for a reference quantum state $\ket{0}$ that is orthogonal to $\{|j\rangle\}_{j=1}^{2m}$.
Then the following statements hold.
\begin{enumerate}
    \item
    The Hamiltonian $H_b$ has the unique ground state
    \begin{equation}
        \ket{g_b}
        =
        \frac{\ket{0}+\ket{p_b}}{\sqrt{2}}=\frac{1}{\sqrt{2}}\ket{0}+\frac{1}{\sqrt{2}}\sum_{j=1}^{2m}\sqrt{p_b(j)}\ket{j},
        \label{eq:hard-ground-state}
    \end{equation}
    with ground-state energy $\Delta$, spectral gap $\Delta$, and
    $\left\|H_b\right\|\leq 3\Delta$.

    \item
    The computational-basis probability distribution $q_b(j)$ of $\ket{g_b}$ satisfies
    \begin{equation}
        q_b(0)=\frac{1}{2},
        \qquad
        q_b(j)=\frac{1}{2}p_b(j)
        \quad (j\in[2m]),
        \qquad
        q_b(j)=0
        \quad (j>2m).
        \label{eq:hard-ground-state-distribution}
    \end{equation}

    \item
    For $r\in[m]$, let $b^{(r)}:=b\oplus e_r$, where $e_r\in\{0,1\}^m$ has value one in its $r$-th coordinate and zero elsewhere.
    Then
    \begin{equation}
        \braket{g_b|g_{b^{(r)}}}
        =
        1-
        \frac{
            2-\sqrt{1+2/k}-\sqrt{1-2/k}
        }{4m},~~~1-\braket{g_b|g_{b^{(r)}}}
        \leq
        \frac{1}{mk^2}.
        \label{eq:adjacent-ground-state-overlap}
    \end{equation}

    \item
    The vector $\ket{\widetilde{\delta p_{b,r}}}:=\ket{p_b}-\ket{p_{b^{(r)}}}$
    is supported on ${\rm span}\{\ket{r},\ket{r+m}\}$ and satisfies
    \begin{equation}
        \left\|\ket{\widetilde{\delta p_{b,r}}}\right\|
        \leq
        \frac{2}{k\sqrt{m}}.
        \label{eq:adjacent-amplitude-difference-bound}
    \end{equation}
    Moreover,
    \begin{equation}
        H_b-H_{b^{(r)}}
        =
        -\Delta
        \left(
            \ket{0}\!\bra{\widetilde{\delta p_{b,r}}}
            +
            \ket{\widetilde{\delta p_{b,r}}}\!\bra{0}
        \right).
        \label{eq:adjacent-hamiltonian-difference}
    \end{equation}
\end{enumerate}
\end{lem}

\begin{proof}
For every $r\in[m]$, Eq.~\eqref{eq:hard-distribution} gives $p_b(r)+p_b(r+m)=\frac{1}{m}$.
Hence $p_b$ is normalized and $\ket{p_b}$ is a unit vector.
On ${\rm span}\{\ket{0},\ket{p_b}\}$, the Hamiltonian in
Eq.~\eqref{eq:hard-hamiltonian} has matrix representation
\begin{equation}
    \Delta
    \begin{pmatrix}
        2 & -1\\
        -1 & 2
    \end{pmatrix},
\end{equation}
whose eigenvalues are $\Delta$ and $3\Delta$.
On the orthogonal complement, $H_b$ acts as $2\Delta \bm{1}$.
This proves the first claim, and Eq.~\eqref{eq:hard-ground-state-distribution} follows directly
from Eq.~\eqref{eq:hard-ground-state}.
For adjacent strings $b$ and $b^{(r)}$, all amplitudes except those on $\ket{r}$ and
$\ket{r+m}$ agree in $\ket{p_b}$ and $\ket{p_{b^{(r)}}}$. 
Direct evaluation gives
\begin{equation}
    \braket{p_b|p_{b^{(r)}}}
    =
    1-\frac{1}{m}
    +
    \frac{\sqrt{1+2/k}+\sqrt{1-2/k}}{2m}.
    \label{eq:adjacent-amplitude-overlap}
\end{equation}
Combining Eq.~\eqref{eq:adjacent-amplitude-overlap} with
$\braket{g_b|g_{b^{(r)}}}=(1+\braket{p_b|p_{b^{(r)}}})/2$
gives Eq.~\eqref{eq:adjacent-ground-state-overlap}.
The vector in Eq.~\eqref{eq:adjacent-amplitude-difference-bound} is supported only on
$\ket{r}$ and $\ket{r+m}$. For $k\geq 4$,
$|\sqrt{1\pm 2/k}-1|\leq {2}/{k}$, which implies Eq.~\eqref{eq:adjacent-amplitude-difference-bound}.
Finally, since all amplitudes are real and then
$\braket{p_b|p_{b^{(r)}}}=\braket{p_{b^{(r)}}|p_b}$,
\begin{align}
    1-\braket{g_b|g_{b^{(r)}}}
    &=
    \frac{1}{4}
    \left\|\ket{p_b}-\ket{p_{b^{(r)}}}\right\|^2
    \leq
    \frac{1}{mk^2}.
\end{align}
\end{proof}

We also use the following property of the hard Hamiltonian family $\{H_b\}$ to prove Theorem~\ref{thm:ground-state-distribution}.

\begin{lem}[Local Hamiltonian difference]
\label{lem:local-hamiltonian-difference}
Consider the Hamiltonian family $\{H_b:b\in \{0,1\}^m\}$ defined in Lemma~\ref{lem:hard-hamiltonian}.
For $r\in[m]$, let $\Pi_r$ denote the projector onto the two-dimensional subspace spanned by $\ket{r}$ and $\ket{r+m}$.
Let $\ket{\Psi}$ and $\ket{\Psi'}$ be normalized states,
and let $C$ be an operator acting on registers other than the system register.
Then
\begin{align}
    &\left|
        \bra{\Psi}
        \left[
            C\otimes\left(H_b-H_{b^{(r)}}\right)
        \right]
        \ket{\Psi'}
    \right|
    \leq
    \frac{2\Delta\left\|C\right\|}{k\sqrt{m}}
    \left(
        \left\|\bm{1}_{\rm \overline{\rm S}}\otimes \Pi_r\ket{\Psi}\right\|
        +
        \left\| \bm{1}_{\rm \overline{\rm S}}\otimes \Pi_r\ket{\Psi'}\right\|
    \right),
    \label{eq:local-hamiltonian-difference-bound}
\end{align}
where $\bm{1}_{\overline{\rm S}}$ denotes the identity on registers other than the system register.
\end{lem}

\begin{proof}
By Eq.~\eqref{eq:adjacent-hamiltonian-difference}, the difference $H_b-H_{b^{(r)}}$ is the sum of two rank-one terms, and
$\ket{\widetilde{\delta p_{b,r}}}=\Pi_r\ket{\widetilde{\delta p_{b,r}}}$.
Applying the Cauchy--Schwarz inequality to the two terms and using
Eq.~\eqref{eq:adjacent-amplitude-difference-bound} gives
Eq.~\eqref{eq:local-hamiltonian-difference-bound}.
\end{proof}

\subsection{Static--dynamic distinguishability inequality}

In order to prove Theorem~\ref{thm:ground-state-distribution},
we here describe a general quantum algorithm in the $P$-parallel access model.
Consider any integer $P\geq 1$.
The algorithm initially receives $P$ system registers in the state $\ket{g_b}^{\otimes P}$
with a hidden oracle input $b$.
To include arbitrary quantum algorithms without loss of generality, we also allow the algorithm to introduce any finite number of ancilla qubits initialized in a fixed state $\ket{0}_{\mathrm a}$~\cite{ambainis2000quantum,hoyer2005lower,de2019quantum}.
Any mid-circuit measurement can be postponed until the end of the circuit: its outcome can be stored coherently in an ancilla register, and every subsequent classically controlled operation can be replaced by a unitary operation controlled by that register. 
It is therefore sufficient to %consider a single final 
assume that the 
measurement is only performed at the end of the circuit.

For $a\in[P]$, let
\begin{equation}
    H_b^{(a)}
    :=
    I_d^{\otimes(a-1)}
    \otimes H_b
    \otimes I_d^{\otimes(P-a)}
    \label{eq:copy-hamiltonian}
\end{equation}
act on the $a$-th system register.
The generator of the $\ell$-th parallel time evolution segment can be written as
\begin{equation}
    H'_{b,\ell}
    :=
    \sum_{a=1}^{P}
    s_{\ell,a}C_{\ell,a}\otimes H_b^{(a)},
    \qquad
    \ell\in[L],
    \label{eq:parallel-controlled-generator}
\end{equation}
where $s_{\ell,a}=\pm 1$, and each $C_{\ell,a}$ is either zero, the one-qubit identity, or a one-qubit control projector
$\ket{1}\!\bra{1}$ acting on an ancilla qubit. 
Identity operators on all remaining registers are omitted for simplicity. 
To ensure the parallel implementation, $C_{\ell,a}$ and $C_{\ell,a'}$ act on distinct ancilla qubits for all $a\neq a'$.
This form allows the $P$ Hamiltonian evolutions to be used or controlled separably within each parallel segment.

Let $V_0,\ldots,V_L$ be arbitrary unitary operators independent of $b$ and acting jointly on all system and ancilla registers. 
The intermediate states in the entire circuit can be written as
\begin{align}
    \ket{\Phi_b^{(0)}}
    &:=
    V_0
    \left(
        \ket{g_b}^{\otimes P}\ket{0}_{\mathrm a}
    \right),\qquad
    \ket{\Phi_b^{(\ell)}}:=
    V_\ell
    e^{-it_\ell H'_{b,\ell}}
    \ket{\Phi_b^{(\ell-1)}},
    \qquad
    \ell\in[L].
    \label{eq:general-parallel-algorithm}
\end{align}
Here $t_\ell\in\mathbb{R}$, so Eq.~\eqref{eq:general-parallel-algorithm} includes time reversal.
The unitaries $V_\ell$ include arbitrary fixed processing and entangling operations among the $P$ system registers and all ancilla registers. 
Any feed-forward operation is included through its coherent ancilla implementation, and every controlled or uncontrolled parallel Hamiltonian evolution permitted by our access model is included through
Eq.~\eqref{eq:parallel-controlled-generator}.
In addition, we can take an arbitrarily finite value $L$, meaning that any finite adaptive control of the evolution time is also included in this model.
As a result, Eq.~\eqref{eq:general-parallel-algorithm} represents all adaptive quantum algorithms allowed in the $P$-parallel access model.

To derive the inequality, we here introduce a progress function in the same way as the quantum adversary methods~\cite{ambainis2000quantum,hoyer2005lower,de2019quantum}.
This function tracks the indistinguishability of the hidden input (i.e., the bit-string $b$ in our case) during the algorithm.
Let $B$ be uniformly distributed over $\{0,1\}^m$, and let $R$ be uniformly distributed over $[m]$, independently of $B$.
We then define the progress function by
\begin{align}
    W_\ell
    &:=
    \frac{1}{2^m m}
    \sum_{b\in\{0,1\}^m}
    \sum_{r=1}^{m}
    \operatorname{Re}
    \braket{\Phi_b^{(\ell)}|\Phi_{b^{(r)}}^{(\ell)}}
    =
    \mathbb{E}_{B,R}
    \left[
        \operatorname{Re}
        \braket{\Phi_B^{(\ell)}|\Phi_{B^{(R)}}^{(\ell)}}
    \right],
    \qquad
    \ell=0,\ldots,L.
    \label{eq:progress-function}
\end{align}

\begin{lem}[Progress accumulation]
\label{lem:progress-accumulation}
For an arbitrary algorithm of the form Eq.~\eqref{eq:general-parallel-algorithm}, the progress function satisfies
\begin{equation}
    \left|W_L-W_0\right|
    \leq
    \frac{4P\Delta}{km}\,T,
    \label{eq:progress-accumulation}
\end{equation}
where $T$ is the elapsed controlled Hamiltonian-evolution time defined in \green{the main text Eq.~(3)}.
% sequential total evolution time defined by Eq.~\eqref{eq:sequential-total-evolution-time}.
\end{lem}

\begin{proof}
Telescoping expansion gives
\begin{equation}
    W_L-W_0
    =
    \sum_{\ell=1}^{L}
    \left(W_\ell-W_{\ell-1}\right).
\end{equation}
For fixed adjacent inputs $b$ and $b^{(r)}$, differentiating the overlap along the $\ell$-th segment gives
\begin{align}
    &\braket{\Phi_b^{(\ell)}|\Phi_{b^{(r)}}^{(\ell)}}
    -
    \braket{\Phi_b^{(\ell-1)}|\Phi_{b^{(r)}}^{(\ell-1)}}
    =
    i \frac{t_\ell}{|t_\ell|}
    \int_{0}^{|t_\ell|}
    \bra{X_{b,\ell}(\tau)}
    \left(
        H'_{b,\ell}-H'_{b^{(r)},\ell}
    \right)
    \ket{X_{b^{(r)},\ell}(\tau)}
    \,{\rm d}\tau,
    \label{eq:overlap-change-integral}
\end{align}
where
\begin{equation}
    \ket{X_{b,\ell}(\tau)}
    :=
    e^{-i\tau \frac{t_\ell}{|t_\ell|} H'_{b,\ell}}
    \ket{\Phi_b^{(\ell-1)}}.
    \label{eq:intermediate-evolution-state}
\end{equation}
The generator difference is
\begin{equation}
    H'_{b,\ell}-H'_{b^{(r)},\ell}
    =
    \sum_{a=1}^{P}
    s_{\ell,a}C_{\ell,a}\otimes
    \left(H_b-H_{b^{(r)}}\right)^{(a)}.
    \label{eq:parallel-generator-difference}
\end{equation}
Applying Lemma~\ref{lem:local-hamiltonian-difference} to each system register, we obtain, for every $\tau$,
\begin{align}
    &\mathbb{E}_{B,R}
    \left[
        \left|
            \bra{X_{B,\ell}(\tau)}
            \left(
                H'_{B,\ell}-H'_{B^{(R)},\ell}
            \right)
            \ket{X_{B^{(R)},\ell}(\tau)}
        \right|
    \right]
    \notag\\
    &\quad\leq \sum_{a=1}^{P}
     \mathbb{E}_{B,R}
    \left[
        \left|
            \bra{X_{B,\ell}(\tau)}
                C_{\ell,a}\otimes
    \left(H_B-H_{B^{(R)}}\right)^{(a)}
            \ket{X_{B^{(R)},\ell}(\tau)}
        \right|
    \right]
    \notag\\
    &\quad\leq
    \sum_{a=1}^{P}
    \frac{2\Delta}{k\sqrt{m}}
    \mathbb{E}_{B,R}
    \left[
        \left\|\Pi_R^{(a)}\ket{X_{B,\ell}(\tau)}\right\|
        +
        \left\|\Pi_R^{(a)}\ket{X_{B^{(R)},\ell}(\tau)}\right\|
    \right]
    \notag\\
    &\quad\leq
    \frac{4P\Delta}{km}.
    \label{eq:instantaneous-progress-bound}
\end{align}
In the final inequality, we used the following evaluation: noting that $\sum_{r=1}^m \Pi_r^{(a)}\leq \bm{1}^{(a)}$,
\begin{align}
    \left(\mathbb{E}_{B,R}
    \left[
        \left\|\Pi_R^{(a)}\ket{X_{B,\ell}(\tau)}\right\|
    \right]\right)^2&\leq \mathbb{E}_{B,R}
    \left[
        \left\|\Pi_R^{(a)}\ket{X_{B,\ell}(\tau)}\right\|^2
    \right]\notag\\
    &=
    \frac{1}{m}
    \mathbb{E}_{B}
    \left[
        \sum_{r=1}^{m}
        \bra{X_{B,\ell}(\tau)}\Pi_r^{(a)}\ket{X_{B,\ell}(\tau)}
    \right]
    \leq
    \frac{1}{m}.
\end{align}
The same evaluation holds when replacing $(B,R)\mapsto (B^{(R)},R)$.
Taking absolute values in Eq.~\eqref{eq:overlap-change-integral}, averaging over $(B,R)$,
and using Eq.~\eqref{eq:instantaneous-progress-bound}, we conclude that
\begin{align}
    \left|W_L-W_0\right|
    \leq
    \sum_{\ell=1}^{L}
    \frac{4P\Delta}{km}|t_\ell|
    =
    \frac{4P\Delta}{km}\,T,
    \label{eq:proof-progress-accumulation}
\end{align}
which proves Eq.~\eqref{eq:progress-accumulation}.
\end{proof}

Now, we are ready to prove the following theorem.

\begin{thm}
% [Lower bound in ground-state distribution estimation]
[Static--dynamic distinguishability inequality]
\label{thm:ground-state-distribution}
There exist universal positive constants $C,C'$, and $\omega_0$ such that the following holds.
For any system dimension $d\geq 4$, 
spectral gap $\Delta>0$, and estimation error $\omega\in (0,\omega_0]$, let us consider the Hamiltonian family with $k=\lfloor 1/(32\omega)\rfloor$ in Lemma~\ref{lem:hard-hamiltonian}. 
Let
\begin{equation}
    q_b(j):=\left|\braket{j|g_b}\right|^2,
    \qquad
    j\in\{0,\ldots,d-1\},
    \label{eq:ground-state-distribution}
\end{equation}
denote the computational-basis measurement distribution in the unique ground state $\ket{g_b}$ of $H_b$.
Then, any quantum algorithm in the $P$-parallel access model that, for every $b$, outputs an
estimator $\widehat q$ satisfying
\begin{equation}
    \Pr\!\left[\left\|\widehat q-q_b\right\|_1\leq\omega\right]
    \geq \frac{2}{3}
    \label{eq:distribution-estimation-success}
\end{equation}
% must use the sequential total evolution time
% \begin{equation}
%     T
%     =
%     \Omega\!\left(
%         \frac{d}{P\Delta\omega}
%         \left[
%             1-C\frac{P\omega^2}{d}
%         \right]_+
%     \right),
%     \qquad
%     [x]_+:=\max\{x,0\}.
%     \label{eq:distribution-estimation-lower-bound}
% \end{equation}
must obey the following inequality
\begin{equation}
    \frac{P\omega^2}{d}+\frac{P \omega \Delta T}{d}\geq C
    \label{eq:supple-sd-ineq}
\end{equation}
for its elapsed controlled Hamiltonian-evolution time $T$.
Equivalently, any such quantum algorithms must use Hamiltonian-evolution time of
\begin{equation}
    T
    =
    \Omega\!\left(
        \frac{d}{P\Delta\omega}
        \left[
            1-C'\frac{P\omega^2}{d}
        \right]_+
    \right),
    \qquad
    [x]_+:=\max\{x,0\}.
    \label{eq:distribution-estimation-lower-bound}
\end{equation}
\end{thm}

\begin{rem}[Role of the initial copies]
\label{rem:role-of-initial-copies}
The factor $[1-C'P\omega^2/d]_+$ in Eq.~\eqref{eq:distribution-estimation-lower-bound} accounts for the information already contained in the initial $P$ ground-state copies.
We can easily see that this dependence is unavoidable by assessing the empirical estimator given large number of copies, as follows. Measuring the $P$ initial copies in the computational basis produces $P$ independent samples from $q_b$. 
If $\widehat q_{\rm emp}$ denotes the empirical distribution, then $\mathbb{E}[\|\widehat q_{\rm emp}-q_b\|_1]\leq\sqrt{(d-1)/P}$, and Markov's inequality gives
$\Pr[\|\widehat q_{\rm emp}-q_b\|_1>\omega]\leq
\sqrt{(d-1)/P}/\omega$. 
Consequently, if $P\geq 9(d-1)/\omega^2$, the empirical estimator achieves $\ell_1$ error at most $\omega$ with probability at least $2/3$ and uses zero Hamiltonian evolution time.
\end{rem}

\begin{proof}[Proof of Theorem~\ref{thm:ground-state-distribution}]

This proof proceeds in four steps.
First, we combine the $\ell_1$-distribution estimator with a deterministic classical decoder and thereby obtain a decoding algorithm that recovers a constant fraction of the hidden bits $b\in \{0,1\}^m$ with constant probability.
Second, we prove a threshold of the final progress $W_L$ for all decoding quantum algorithms satisfying this recovery guarantee.
Third, we evaluate the initial progress $W_0$ from the overlap of the $P$ initial ground-state copies.
Finally, Lemma~\ref{lem:progress-accumulation} bounds the total change of the progress and yields the claimed evolution-time bound.

We first construct the deterministic decoder.
In particular, we show that any quantum algorithm yielding the $\ell_1$ estimator for $q_b$ induces a decoder that recovers at least a $15/16$ fraction of the hidden bits $b\in \{0,1\}^m$ with probability at least $2/3$.
By choosing $\omega_0\leq 1/128$, we have
\begin{equation}
    k\geq 4,
    \qquad
    2k\omega\leq\frac{1}{16},
    \qquad
    k\geq\frac{1}{64\omega}.
    \label{eq:k-properties}
\end{equation}
Let $\widehat q$ be the output of the algorithm. Define
\begin{equation}
    \widetilde p(j)
    :=
    2\widehat q(j),
    \qquad
    j\in[2m],
\end{equation}
and decode the hidden bits according to
\begin{equation}
    \widehat b_r
    :=
    \begin{cases}
        1,
        &
        \widetilde p(r)-\widetilde p(r+m)\geq 1/(mk),\\
        0,
        &
        \widetilde p(r)-\widetilde p(r+m)<1/(mk),
    \end{cases}
    \qquad
    r\in[m].
    \label{eq:hidden-bit-decoder}
\end{equation}
Here, Eq.~\eqref{eq:hard-distribution} implies
\begin{equation}
    p_b(r)-p_b(r+m)
    =
    \begin{cases}
        0,
        & b_r=0,\\[1mm]
        2/(mk),
        & b_r=1.
    \end{cases}
\end{equation}
The threshold in Eq.~\eqref{eq:hidden-bit-decoder} is the midpoint of these two values. Therefore, if $\widehat b_r\neq b_r$, then
\begin{equation}
    \left|
        \widetilde p(r)-p_b(r)
    \right|
    +
    \left|
        \widetilde p(r+m)-p_b(r+m)
    \right|
    \geq
    \frac{1}{mk}.
    \label{eq:error-cost-per-wrong-bit}
\end{equation}
Summing Eq.~\eqref{eq:error-cost-per-wrong-bit} over the incorrectly decoded bits gives
\begin{align}
    \frac{1}{m}|\{r\in [m]:\hat{b}_r\neq b_r\}|
    &\leq
    k
    \sum_{j=1}^{2m}
    \left|
        \widetilde p(j)-p_b(j)
    \right|\leq 2k\|\hat{q}-q_b\|_1
\end{align}
Thus, on the successful event $\left\|\widehat q-q_b\right\|_1\leq\omega$, $|\{r\in [m]:\hat{b}_r\neq b_r\}|\leq m/16$ holds.
This means that for every input $b$, any distribution-estimation algorithm, followed by the deterministic post-processing in Eq.~\eqref{eq:hidden-bit-decoder}, produces an output bit string
$\hat b$ satisfying
\begin{equation}
    \Pr\!\left[
        |\{r\in [m]:\hat{b}_r\neq b_r\}|\leq\frac{m}{16}
    \right]
    \geq
    \frac{2}{3}.
\end{equation}
This post-processing is entirely classical and does not use the additional Hamiltonian evolution.

We next consider an arbitrary decoding algorithm in our access model.
Let $\{F_a:a\in\{0,1\}^m\}$ be its arbitrary
final POVM, where the outcome $a$ is the decoded bit string. 
The POVM is independent of the unknown input $b$ and need not arise from an $\ell_1$ estimator.
The estimator followed by Eq.~\eqref{eq:hidden-bit-decoder} is one special case of this general decoding measurement.
Thus, a lower bound established below for any decoding algorithm applies in particular to the original $\ell_1$-distribution-estimation algorithm.

The only property of the arbitrary decoding POVM used below is
\begin{equation}
    \sum_{a:\,|\{r\in [m]:a_r\neq b_r\}|\leq m/16}
    \operatorname{tr}\!\left[
        F_a\rho_b^{(L)}
    \right]
    \geq
    \frac{2}{3}
    \qquad
    \text{for every }b\in\{0,1\}^m,
    \label{eq:general-decoding-povm-condition}
\end{equation}
where $\rho_b^{(L)}:=\ket{\Phi_b^{(L)}}\!\bra{\Phi_b^{(L)}}$.
For $r\in[m]$ and $z\in\{0,1\}$, define the bit-marginal POVM elements
\begin{equation}
    M_r^{(z)}
    :=
    \sum_{a:\,a_r=z}F_a.
    \label{eq:bit-marginal-povm}
\end{equation}
For every fixed $b$, Eqs.~\eqref{eq:general-decoding-povm-condition} and
\eqref{eq:bit-marginal-povm} give
\begin{align}
    \frac{1}{m}
    \sum_{r=1}^{m}
    \operatorname{tr}\!\left[
        M_r^{(b_r)}\rho_b^{(L)}
    \right]
    &=
    \sum_{a\in\{0,1\}^m}
    \left(
        1-\frac{|\{r\in [m]:a_r\neq b_r\}|}{m}
    \right)
    \operatorname{tr}\!\left[
        F_a\rho_b^{(L)}
    \right]
    \notag\\
    &\geq
    \frac{15}{16}
    \sum_{a:\,|\{r\in [m]:a_r\neq b_r\}|\leq m/16}
    \operatorname{tr}\!\left[
        F_a\rho_b^{(L)}
    \right]
    \geq
    \frac{5}{8}.
    \label{eq:average-bit-recovery-probability}
\end{align}
Because $(b^{(r)})_r=1-b_r$, one has
$M_r^{((b^{(r)})_r)}=I-M_r^{(b_r)}$, and thus, for every adjacent pair,
\begin{align}
    \operatorname{tr}\!\left[
        M_r^{(b_r)}\rho_b^{(L)}
    \right]
    +
    \operatorname{tr}\!\left[
        M_r^{((b^{(r)})_r)}\rho_{b^{(r)}}^{(L)}
    \right]
    &=
    1+
    \operatorname{tr}\!\left[
        M_r^{(b_r)}
        \left(
            \rho_b^{(L)}-\rho_{b^{(r)}}^{(L)}
        \right)
    \right]
    \notag\\
    &\leq
    1+
    \frac{1}{2}
    \left\|
        \rho_b^{(L)}-\rho_{b^{(r)}}^{(L)}
    \right\|_1.
    \label{eq:bit-recovery-trace-distance-bound}
\end{align}
Averaging Eq.~\eqref{eq:bit-recovery-trace-distance-bound} over $(B,R)$, and using
Eq.~\eqref{eq:average-bit-recovery-probability}, we obtain
\begin{equation}
    \mathbb{E}_{B,R}
    \left[
        \frac{1}{2}
        \left\|
            \rho_B^{(L)}-\rho_{B^{(R)}}^{(L)}
        \right\|_1
    \right]
    \geq
    \frac{1}{4}.
    \label{eq:average-terminal-trace-distance}
\end{equation}
For pure states,
\begin{equation}
    \left|
        \braket{
            \Phi_b^{(L)}
            |
            \Phi_{b^{(r)}}^{(L)}
        }
    \right|
    =
    \sqrt{
        1-
        \frac{1}{4}
        \left\|
            \rho_b^{(L)}-\rho_{b^{(r)}}^{(L)}
        \right\|_1^2
    }.
\end{equation}
Using $\operatorname{Re}z\leq|z|$, the concavity of
$x\mapsto\sqrt{1-x^2}$ on $[0,1]$, and
Eq.~\eqref{eq:average-terminal-trace-distance}, we conclude that
\begin{align}
    W_L
    \leq
    \mathbb{E}_{B,R}
    \left[
        \sqrt{
            1-
            \frac{1}{4}
            \left\|
                \rho_B^{(L)}-\rho_{B^{(R)}}^{(L)}
            \right\|_1^2
        }
    \right]
    \leq
    \sqrt{
        1-
        \left(
            \mathbb{E}_{B,R}
            \left[
                \frac{1}{2}
                \left\|
                    \rho_B^{(L)}-\rho_{B^{(R)}}^{(L)}
                \right\|_1
            \right]
        \right)^2
    }
    \leq
    \frac{\sqrt{15}}{4}.
    \label{eq:terminal-progress-bound}
\end{align}

We next evaluate the initial progress. The initial processing $V_0$ is independent of $b$, and the initial state contains $P$ copies of $\ket{g_b}$. Equation~\eqref{eq:adjacent-ground-state-overlap}
therefore gives
\begin{equation}
    W_0
    =
    \left(
        1-
        \frac{
            2-\sqrt{1+2/k}-\sqrt{1-2/k}
        }{4m}
    \right)^P.
    \label{eq:initial-progress-exact}
\end{equation}
Using Eq.~\eqref{eq:adjacent-ground-state-overlap}, we obtain
\begin{equation}
    W_0
    \geq
    \left(
        1-\frac{1}{mk^2}
    \right)^P
    \geq
    1-\frac{P}{mk^2}.
    \label{eq:initial-progress-lower-bound}
\end{equation}
Combining Eqs.~\eqref{eq:progress-accumulation},
\eqref{eq:terminal-progress-bound}, and
\eqref{eq:initial-progress-lower-bound} yields
\begin{equation}
    T
    \geq
    \frac{km}{4P\Delta}
    \left[
        1-\frac{\sqrt{15}}{4}-\frac{P}{mk^2}
    \right]_+.
    \label{eq:precise-intermediate-time-lower-bound}
\end{equation}
Finally, $m\geq d/4$ for $d\geq4$, and Eq.~\eqref{eq:k-properties} implies
\begin{equation}
    km
    =
    \Omega\!\left(
        \frac{d}{\omega}
    \right),
    \qquad
    \frac{P}{mk^2}
    =
    \mathcal{O}\!\left(
        \frac{P\omega^2}{d}
    \right).
    \label{eq:dimension-accuracy-conversion}
\end{equation}
Substituting Eq.~\eqref{eq:dimension-accuracy-conversion} into
Eq.~\eqref{eq:precise-intermediate-time-lower-bound} and absorbing universal constants proves
Eq.~\eqref{eq:distribution-estimation-lower-bound}.
The inequality~\eqref{eq:supple-sd-ineq} is also obtained from Eq.~\eqref{eq:precise-intermediate-time-lower-bound}.
\end{proof}

\subsection{Multiple observables estimation and tomography}

\begin{thm}[Lower bound in multiple observables estimation]
\label{thm:simultaneous-expectation-lower-bound}
There exist universal constants $C,\varepsilon_0>0$ such that the following holds.
For an arbitrary power of two $d\geq 4$, number of observables $M\geq 4$, spectral gap $\Delta>0$, and estimation error
$\varepsilon\in(0,\varepsilon_0]$, there exist an explicit family of $d$-dimensional Hamiltonians $\{H_b\}$ 
and $b$-independent $M$ mutually commuting observables $O_1,\ldots,O_M$ such that each $H_b$ has a unique ground state $\ket{g_b}$ and spectral gap $\Delta$, while $\left\|O_j\right\|\leq 1$ for any $j\in [M]$.
Then, any quantum algorithm in the $P$-parallel access model that for every $b$, outputs estimators $\widehat\mu_1,\ldots,\widehat\mu_M$ satisfying
\begin{equation}
    \Pr\!\left[
        \max_{j\in[M]}
        \left|
            \widehat\mu_j-
            \bra{g_b}O_j\ket{g_b}
        \right|
        \leq\varepsilon
    \right]
    \geq
    \frac{2}{3}
    \label{eq:expectation-estimation-success}
\end{equation}
must use the elapsed controlled Hamiltonian-evolution time
% sequential total evolution time
\begin{equation}
    T
    =
    \Omega\!\left(
        \frac{1}{P\Delta}
        \min\!\left\{
            \frac{1}{\varepsilon^2},
            \frac{\sqrt{M}}{\varepsilon},\frac{\sqrt{d}}{\varepsilon}
        \right\}
        \left[
            1-CP\varepsilon^2
        \right]_+
    \right).
    \label{eq:simultaneous-expectation-lower-bound}
\end{equation}
\end{thm}

From this theorem, the following corollary is immediately obtained.
\begin{cor}
    Given $P$ copies of the unique ground state of a $d$-dimensional Hamiltonian with spectral gap $\Delta$, together with $P$-parallel controlled access to its time evolution, 
    any quantum algorithm must, in the worst case, use the 
    % sequential total evolution time 
    elapsed controlled Hamiltonian-evolution time
    \begin{equation}
    T
    =
    \Omega\!\left(
        \frac{1}{P\Delta}
        \min\!\left\{
            \frac{1}{\varepsilon^2},
            \frac{\sqrt{M}}{\varepsilon},\frac{\sqrt{d}}{\varepsilon}
        \right\}
        \left[
            1-CP\varepsilon^2
        \right]_+
    \right).
    \end{equation}
    to estimate all ground-state expectation values of arbitrary given $M$ bounded-operator-norm observables within additive error $\varepsilon$ with high probability. Here, $C$ is a universal positive constant.
\end{cor}

\begin{proof}[Proof of Theorem~\ref{thm:simultaneous-expectation-lower-bound}]
We first focus on an $s$-dimensional hard Hamiltonian family in 
Theorem~\ref{thm:ground-state-distribution}, where $s\leq d$ is a power of two chosen below. 
We embed this family into a $d$-dimensional Hamiltonian family while preserving its unique ground state, spectral gap, and controlled time evolution cost.
We then construct mutually commuting observables such that estimating their ground-state expectation values solves the original $s$-dimensional $\ell_1$-distribution-estimation problem.
Consequently, any multiple observables estimation algorithm induces an algorithm in the class covered by Theorem~\ref{thm:ground-state-distribution}.

Let $\omega_0$ be the universal constant in Theorem~\ref{thm:ground-state-distribution}. We choose $s$ to be the largest power of two satisfying
\begin{equation}
    s
    \leq
    \min\left\{
        d,
        M,
        \frac{\omega_0^2}{\varepsilon^2}
    \right\}\leq \frac{\omega_0^2}{\varepsilon^2}
    \label{eq:effective-hard-dimension}
\end{equation}
By choosing the universal upper bound on $\varepsilon$
sufficiently small, we can assume $s\geq4$.
Consider Theorem~\ref{thm:ground-state-distribution} with the 
system dimension $s$ and estimation error given by $\omega:=\sqrt{s}\,\varepsilon$.
Let
$\{\widetilde H_b\}$ be the resulting $s$-dimensional hard
Hamiltonian family, let $\ket{\widetilde g_b}$ denote its
unique ground state, and let
\begin{equation}
    \widetilde q_b(y)
    :=
    \left|
        \braket{y|\widetilde g_b}
    \right|^2,
    \qquad
    y\in\{0,1\}^{\log_2 s},
\end{equation}
be its computational-basis distribution.

We next embed $\widetilde H_b$ into a $\log_2 d$-qubit
system. 
Define
\begin{equation}
    \nu
    :=
    \log_2\left(\frac{d}{s}\right),
    \qquad
    \Pi_0
    :=
    \ket{0}^{\otimes\nu}
    \!\bra{0}^{\otimes\nu}.
\end{equation}
The first register consists of the $\nu$ added qubits, and the subspace in which these qubits are all zero is identified with the $s$-dimensional hard subspace. Define
\begin{equation}
    H_b
    :=
    \Pi_0\otimes\widetilde H_b
    +
    \left(
        I^{\otimes \nu}-\Pi_0
    \right)
    \otimes
    2\Delta {I}^{\otimes \log_2 s}
    \label{eq:embedded-hard-hamiltonian}
\end{equation}
for the one-qubit identity $I$.
This Hamiltonian can be written as $H_b=\widetilde H_b\oplus2\Delta \bm{1}$ after a fixed permutation of the computational basis.
Therefore, $\ket{g_b}:=\ket{0}^{\otimes\nu}\otimes\ket{\widetilde g_b}$ is the unique ground state of $H_b$, and the spectral gap of $H_b$ is $\Delta$.

The corresponding time evolution is
\begin{equation}
    e^{-itH_b}
    =
    \Pi_0
    \otimes
    e^{-it\widetilde H_b}
    +
    \left(
        I^{\otimes \nu}-\Pi_0
    \right)
    \otimes
    e^{-2i\Delta t}I^{\otimes \log_2 s}.
    \label{eq:embedded-hamiltonian-evolution}
\end{equation}
This unitary, as well as its controlled version, can be
implemented using one controlled black-box application of
$e^{-it\widetilde H_b}$ together with gates and phase
operations that are independent of $b$. In particular, the
hard evolution is applied only when all the added qubits are
zero, whereas the known phase $e^{-2i\Delta t}$ is applied
on the orthogonal subspace. 
Thus, every controlled $P$-parallel $H_b$-evolution segment of time $t$ can be
realized using one controlled $P$-parallel
$\widetilde H_b$-evolution segment of the same time.
The sequential total evolution time is therefore unchanged when querying $e^{-itH_b}$, rather than $e^{-it\widetilde{H}_b}$.

Following a reduction technique used in the proof of Lemma~13 in Ref.~\cite{van2021quantum}, 
we now define the hard observables with the standard single-qubit Hadamard gate
\begin{equation}
    \mathsf{H}
    :=
    \frac{1}{\sqrt{2}}
    \begin{pmatrix}
        1 & 1\\
        1 & -1
    \end{pmatrix},
    \qquad
    \mathsf{H}_s
    :=
    \mathsf{H}^{\otimes \log_2 s}.
\end{equation}
For each $z\in\{0,1\}^{\log_2 s}$, define the observable
\begin{equation}
    O_z
    :=
    \Pi_0
    \otimes
    \sum_{y\in\{0,1\}^{\log_2 s}}
    \sqrt{s}\,
    \bra{z}\mathsf{H}_s\ket{y}
    \ket{y}\!\bra{y}.
    \label{eq:hard-hadamard-observable}
\end{equation}
The observables $\{O_z\}$ are mutually commuting,
diagonal in the computational-basis, and satisfy $\left\|O_z\right\|=1$.
If $M>s$, the remaining $M-s$ observables may be chosen
arbitrarily among mutually commuting computational-basis
diagonal operators with operator norm at most one. 
For the embedded ground state, the expectation value of
$O_z$ is
\begin{align}
    \mu_b(z)
    &:=
    \bra{g_b}O_z\ket{g_b}
    =
    \sqrt{s}
    \sum_{y\in\{0,1\}^{\log_2 s}}
    \bra{z}\mathsf{H}_s\ket{y}
    \widetilde q_b(y).
    \label{eq:expectation-from-hadamard-gate}
\end{align}
For the column vectors $\tilde{q}_b$ and $\mu_b$, it is equivalent to $\mu_b=\sqrt{s}\mathsf{H}_s\tilde{q}_b$.

Suppose that there exists a $d$-dimensional multiple observables estimation algorithm with 
% sequential total evolution time 
the elapsed controlled Hamiltonian-evolution time $T$ in the $P$-parallel access model.
We construct from it an $s$-dimensional $l_1$-distribution-estimation
algorithm as follows.
Let $\widehat\mu$ denote the first $s$ expectation-value
estimators produced by this algorithm.
Since $\mathsf{H}_s^2=I^{\otimes \log_2 s}$, we can define a distribution estimator $\widehat{q}:=\frac{1}{\sqrt{s}}\mathsf{H}_s\widehat{\mu}$.
On the success event Eq.~\eqref{eq:expectation-estimation-success}, the estimator satisfies
\begin{align}
    \left\|
        \widehat q-\widetilde q_b
    \right\|_1
    &\leq
    \sqrt{s}
    \left\|
        \widehat q-\widetilde q_b
    \right\|_2
    =\left\|
        \widehat\mu-\mu_b
    \right\|_2
    \leq
    \sqrt{s}
    \left\|
        \widehat\mu-\mu_b
    \right\|_\infty
    \leq
    \sqrt{s}\,\varepsilon.
\end{align}
Therefore, applying Theorem~\ref{thm:ground-state-distribution} with the system dimension $s$ and the estimation error
$\omega=\sqrt{s}\,\varepsilon$ gives
\begin{align}
    T
    &=
    \Omega\left(
        \frac{s}{
            P\Delta\sqrt{s}\varepsilon
        }
        \left[
            1-
            C
            \frac{
                P(\sqrt{s}\varepsilon)^2
            }{s}
        \right]_+
    \right)
    =
    \Omega\left(
        \frac{\sqrt{s}}{
            P\Delta\varepsilon
        }
        \left[
            1-CP\varepsilon^2
        \right]_+
    \right).
    \label{eq:effective-dimension-observable-lower-bound}
\end{align}
Since $s$ is the largest power of two satisfying
Eq.~\eqref{eq:effective-hard-dimension},
\begin{equation}
    \sqrt{s}
    =
    \Theta\left(
        \min\left\{
            \sqrt{d},
            \sqrt{M},
            \frac{1}{\varepsilon}
        \right\}
    \right).
    \label{eq:effective-hard-dimension-scaling}
\end{equation}
Substituting Eq.~\eqref{eq:effective-hard-dimension-scaling} into Eq.~\eqref{eq:effective-dimension-observable-lower-bound}
proves the claimed lower bound.
\end{proof}

\begin{thm}[Lower bound in ground-state tomography]
\label{thm:ground-state-tomography-lower-bound}
There exist universal constants $C,\eta_0>0$ such that the following holds.
For any system dimension $d\geq 4$, spectral gap $\Delta>0$, and estimation error $\eta\in(0,\eta_0]$, let us consider the $\Delta$-gap Hamiltonian family with $k=\lfloor1/(32\eta)\rfloor $ in Lemma~\ref{lem:hard-hamiltonian}, and let $\ket{g_b}$ be the unique ground state of $H_b$.
Any quantum algorithm in the $P$-parallel access model that for every $b$, outputs a classical description of a density matrix $\widehat\rho$ satisfying
\begin{equation}
    \Pr\!\left[
        \frac{1}{2}\left\|
            \widehat\rho-
            \ket{g_b}\!\bra{g_b}
        \right\|_1
        \leq\eta
    \right]
    \geq
    \frac{2}{3}
    \label{eq:tomography-success}
\end{equation}
must use the 
% sequential total evolution time
elapsed controlled Hamiltonian-evolution time
\begin{equation}
    T
    =
    \Omega\!\left(
        \frac{d}{P\Delta\eta}
        \left[
            1-C\frac{P\eta^2}{d}
        \right]_+
    \right).
    \label{eq:ground-state-tomography-lower-bound}
\end{equation}
\end{thm}

This theorem immediately yields the following corollary.
\begin{cor}
    Given $P$ copies of the unique ground state of a $d$-dimensional Hamiltonian with spectral gap $\Delta$, together with $P$-parallel controlled access to its time evolution, any quantum algorithm must, in the worst case, use the 
    % sequential total evolution time 
    elapsed controlled Hamiltonian-evolution time
    \begin{equation}
    T
    =
    \Omega\!\left(
        \frac{d}{P\Delta\eta}
        \left[
            1-C\frac{P\eta^2}{d}
        \right]_+
    \right)
    \end{equation}
    to produce a classical description of the ground-state density matrix within trace-distance error $\eta$ with high probability.
    Here, $C$ is a universal positive constant.
\end{cor}

\begin{proof}[Proof of Theorem~\ref{thm:ground-state-tomography-lower-bound}]
Let
\begin{equation}
    \rho_b
    :=
    \ket{g_b}\!\bra{g_b},
    \qquad
    \mathcal{M}(X)
    :=
    \sum_{j=0}^{d-1}
    \bra{j}X\ket{j}\,
    \ket{j}\!\bra{j}
    \label{eq:computational-basis-measurement-channel}
\end{equation}
be the ground-state density matrix and the computational-basis measurement channel, respectively. 
From an estimator $\widehat\rho$, define
$\widehat q(j):=\bra{j}\widehat\rho\ket{j}$.
On the success event in Eq.~\eqref{eq:tomography-success}, trace-norm contractivity under the quantum channel $\mathcal{M}$ gives
\begin{align}
    \left\|\widehat q-q_b\right\|_1
    =
    \left\|
        \mathcal{M}(\widehat\rho)-
        \mathcal{M}(\rho_b)
    \right\|_1
    \leq
    \left\|\widehat\rho-\rho_b\right\|_1
    \leq
    \eta.
    \label{eq:tomography-to-distribution-reduction}
\end{align}
Thus any tomography algorithm satisfying Eq.~\eqref{eq:tomography-success} also solves the
$l_1$-distribution-estimation problem in Theorem~\ref{thm:ground-state-distribution} with
$\omega=\eta$.
Equation~\eqref{eq:ground-state-tomography-lower-bound} follows immediately from
Eq.~\eqref{eq:distribution-estimation-lower-bound}.
\end{proof}

\begin{rem}[Catalytic protocols]
\label{rem:catalytic-lower-bounds}
Theorems~\ref{thm:simultaneous-expectation-lower-bound} and
\ref{thm:ground-state-tomography-lower-bound} impose no requirement on the final state of the $P$ input ground-state registers. Their lower bounds therefore remain valid for catalytic protocols, which form a subclass of the algorithms considered above.
\end{rem}

\section{Parallel multi-signal readout framework}\label{sec:parallel_multi_signal_readout}
\subsection{Single-signal example}

\begin{figure}[tb]
 \centering
 \includegraphics[scale=0.9]{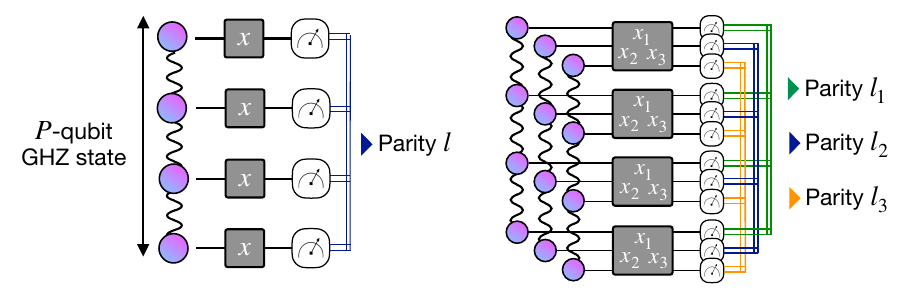}
 \caption{Separable parity measurement. (Left) Single-signal case. The black-box is the one-qubit phase shifter $S_{\tau x}={\rm diag}(1,e^{i\tau x})$. We perform $S_{\tau x}^{\otimes P}$ and then measure each qubit by the Pauli $X$ basis.
 After classically gathering the outcomes, we determine the parity of the number of getting $\ket{-}$ states.
 We may perform a single-qubit local gate before the Pauli-X measurement.
 (Right) Multi-signal case (three signals). The black-box is the joint phase shifter $S_{\tau x_1}\otimes S_{\tau x_2}\otimes S_{\tau x_3}$. Ideally, the measured parities $(l_1,l_2,l_3)$ are independent.
 This parity measurement can be separably performed in a parallel way after sharing GHZ states.}
 \label{fig:local_signal_readout_framework}
\end{figure}

We here describe a detailed procedure for the parallel multi-signal readout framework.
As a canonical example, we first consider the estimation of a signal $x\in [-\pi,\pi)$ encoded in a 1-qubit phase shifter $S_{\tau x}={\rm diag}(1,e^{i\tau x})$. 
The implementation cost (e.g., the sensing time) of this phase shifter is typically given by $|\tau|$.
By applying $S_{\tau x}$ to each qubit in a GHZ state of $P$ qubits, the GHZ state becomes
\begin{equation}
    |{\rm GHZ}_P(\tau x)\rangle= \frac{\ket{0}^{\otimes P}+e^{iP\tau x}\ket{1}^{\otimes P}}{\sqrt{2}},
\end{equation}
and it is well known that this state has a quantum-enhanced sensitivity to the signal $x$~\cite{PhysRevLett.96.010401}.
Specifically, the quantum Fisher information of $|{\rm GHZ}_P(\tau x)\rangle$ on $x$ quadratically scales in the total cost $P|\tau|$.
Indeed, the high sensitivity can be attained by performing separable measurements on each qubit of $|{\rm GHZ}_P(\tau x)\rangle$ followed by classical post-processing~\cite{PhysRevA.54.R4649,PhysRevA.102.042613}; we explain this below.

We here introduce the following primitive process, which we call (the single-signal version of) \textit{the separable parity measurement}, illustrated in Fig.~\ref{fig:local_signal_readout_framework}: 
\begin{itemize}
    \item[(i)] Preparing a GHZ state $|{\rm GHZ}_P\rangle$ and transferring each qubit to each of $P$ local quantum devices,
    \item[(ii)] In each local device, applying the phase shifter $S_{\tau x}$ and measuring the single qubit,
    \item[(iii)] Gathering the local measurement data by classical communication.
\end{itemize}
The measurement and data processing is specified as follows.
% We introduce the parallel readout procedure for a single signal $x$.
For the state $|{\rm GHZ}_P(\tau x)\rangle $ obtained in (ii), we perform the separable Pauli-$X$ measurement on each qubit and then store the parity $l$ ($l=1$ if the parity is even) of the total number of hitting $|-\rangle$.
The parity $l\in \{0,1\}$ follows the probability distribution of ${\rm Pr}[l=1]=(1/2)(1+\cos(P\tau x))$~\cite{PhysRevA.102.042613}.
If we apply $e^{i\pi Z/4}$ only to the first qubit of the phase-shifted GHZ state before the Pauli-$X$ measurements, then the probability of getting parity $l=1$ is changed to ${\rm Pr}[l=1]=(1/2)(1+\sin(P\tau x))$.
This parity measurement can be separably performed in a parallel way after sharing the GHZ state.

After performing a series of the separable parity measurements, we use the measurement data to estimate the target signal $x$.
Here, we need to resolve periodicity with respect to the signal $x$ since the probability distributions of the parity measurement are periodic.
To this end, we employ the following multi-step estimation developed by Refs.~\cite{higgins2009demonstrating,PhysRevA.92.062315,PhysRevA.102.042613}.
The entire procedure has $K=\mathcal{O}(\log(1/\varepsilon))$ stages for a given estimation error $\varepsilon$.
For $k=1,2,...,K$, let $l^{(+,k)}$ and $l^{(-,k)}$ be the measured parities from $|{\rm GHZ}_P(\tau_k x)\rangle$ with and without the first qubit rotation, respectively, where we set 
\begin{equation}
    \tau_k=2^{k-1}/P.
    \label{eq:kstage_tau}
\end{equation}
For each stage $k$, we perform $v_k$ independent measurements for $l^{(+,k)}$ and $l^{(-,k)}$ by consuming $2v_k$ GHZ states in total.
The sequence of $v_k$ is specified later.
Then, we define $\bar{l}^{(+,k)}$ and $\bar{l}^{(-,k)}$ as the sample average of those $v_k$ measurement outcomes.
After independent $K$ stages, we obtain an independently generated data set 
\begin{equation}\label{eq:data_set_single_signal}
    \{\bar{l}^{(+,k)},\bar{l}^{(-,k)}\}_{k=1}^K.
\end{equation}
There is a simple and efficient classical post-processing to estimate the signal $x$ from the data set.
The concrete procedure producing an  estimate $\hat{x} $ is given by Algorithm~\ref{alg_classical_part}.
It is proved by Ref.~\cite{PhysRevA.102.042613} that the resulting estimator $\hat{x}$ has the following mean-squared-error bound
    \begin{equation}
        \mathbb{E}\left[d_{\mathbb{T}}(\hat{x},x)^2\right]\leq \left(\frac{2\pi}{3}\right)^2\left(\frac{1}{4^K}+\sum_{k=1}^K \frac{e^{-3v_k/16}}{4^{k-4}}\right),
    \end{equation}
where the expectation is taken over all possible $\hat{x}$.
Therefore, by taking 
\begin{equation}
    v_k=\lceil (16/3)\ln(6)(K-k)\rceil+1,    
\end{equation}
we ensure that 
\begin{align}\label{eq:d_T_mse_bound}
    \mathbb{E}\left[d_{\mathbb{T}}(\hat{x},x)^2\right]
    &\leq \left(\frac{2\pi}{3}\right)^2\left(\frac{1}{4^K}+\sum_{k=1}^K \frac{6^{k-K}}{4^{k-4}}\right)= \mathcal{O}(4^{-K}).
\end{align}
Suppose that preparing a single $|{\rm GHZ}_P(\tau_k x)\rangle $ requires an implementation cost (e.g., the total sensing time) of $P\tau_k=2^{k-1}$. 
Then, the total cost in the entire estimation is given by 
\begin{equation}\label{eq:rpe_query_relation}
    \tilde{N}=\sum_{k=1}^K v_k 2^{k}=\mathcal{O}(2^K).
\end{equation}
This means that the parallel readout achieves the Heisenberg scaling $\mathbb{E}\left[d_{\mathbb{T}}(\hat{x},x)^2\right]=\mathcal{O}(1/\tilde{N}^2)$.

\begin{algorithm}[H]
    \caption{Classical data processing for one signal (Refs.~\cite{PhysRevA.102.042613,oshio2025near})}
    \label{alg_classical_part}
    \begin{algorithmic}[1]
    \Require Data set $\{\bar{l}^{(+,k)},\bar{l}^{(-,k)}\}_{k=1}^K$ where $\bar{l}^{(\pm,k)}\in [0,1]$.
    \Ensure Real value $\widehat{x} \in [-\pi, \pi) $
    \For{$k = 1, 2,..., K$}
        \State $\widehat{\varphi}'_{k, 0} \gets 2^{-(k-1)}\cdot {\rm atan2}(2\bar{l}^{(+, k)} - 1, 2\bar{l}^{(-, k)}  - 1) \in [0, 2\pi/2^{(k-1)})$.
        \If {$k = 1$}
            \State $\widehat{\varphi}'_1 \gets \widehat{\varphi}'_{1, 0}$
        \Else
            \State $\eta \gets \lfloor {\widehat{\varphi}'_{k-1}2^{k-2}/\pi}\rfloor$
            \If{$\widehat{\varphi}'_{k-1} - ( \widehat{\varphi}'_{k, 0} + (\eta - 1)\pi / 2^{k-2} ) \le \pi/2^{k-1}$}
                \State $\widehat{\varphi}'_k \gets \widehat{\varphi}'_{k, 0} + (\eta - 1)\pi / 2^{k-2}$
            \ElsIf{$( \widehat{\varphi}'_{k, 0} + (\eta + 1)\pi / 2^{k-2} ) - \widehat{\varphi}'_{k-1} < \pi/2^{k-1}$}
                \State $\widehat{\varphi}'_k \gets \widehat{\varphi}'_{k, 0} + (\eta + 1)\pi / 2^{k-2}$
            \Else
                \State $\widehat{\varphi}'_k \gets \widehat{\varphi}'_{k, 0} + \eta\pi / 2^{k-2}$
            \EndIf
        \EndIf
    \EndFor
    \State $\widehat{x} \gets \widehat{\varphi}'_K - 2\pi \lfloor {(\widehat{\varphi}'_K + \pi)}/{(2\pi)} \rfloor$    
    \end{algorithmic}
\end{algorithm}

\subsection{Procedure and guarantee for multi-signal cases}
The above parallel readout procedure is directly extended to multi-signal cases.
We first focus on three signals $\bm{x}=(x_1,x_2,x_3)$, but the following can be directly extended to more signals. 
Suppose that $P$ local quantum devices share 3 GHZ states $|{\rm GHZ}_P\rangle$ in the same way as (i), illustrated in the right of Fig.~\ref{fig:local_signal_readout_framework}.
In each local device, we apply the joint phase shifter $S_{\tau x_1}\otimes S_{\tau x_2}\otimes S_{\tau x_3}$ and separably perform the parity measurement on the 3 qubits.
By using this process as a primitive, 
we can generate an independent data set of Eq.~\eqref{eq:data_set_single_signal} for each signal $x_i$ ($i=1,2,3$), similarly to the single-signal case.
Note that if we can perfectly implement the joint phase shifter $S_{\tau x_1}\otimes S_{\tau x_2}\otimes S_{\tau x_3}$, the data set for $x_i$ and that for $x_j$ ($i\neq j$) are independent.
Therefore, by coordinate-wisely executing Algorithm~\ref{alg_classical_part} for the data set, we can estimate the signal $\bm{x}=(x_1,x_2,x_3)$ with the coordinate-wise mean-squared-error $\mathcal{O}(4^{-K})=\mathcal{O}(\varepsilon^2)$.

Now, we formalize the parallel readout framework for $M$ signals $\bm{x}=(x_1,...,x_M)$.
There are $P$ local quantum devices sharing $M$ GHZ states, and we assume that the following joint phase shifter for $M$ signals of interest can be applied to the local $M$ qubits in each device:
\begin{equation}\label{main_eq:target_phase_shifter_x}
    S_{\tau x_1}\otimes \cdots \otimes S_{\tau x_M}=\sum_{\bm{b}=(b_1,...,b_M)\in\{0,1\}^M} e^{i\tau\bm{b}\cdot \bm{x}}|\bm{b}\rangle\langle \bm{b}|.
\end{equation}
Then, we perform the separable parity measurements on the local $M$ qubits and gather the local measurement data.
By using this primitive process, we can generate the following multi-signal data set
\begin{equation}\label{eq:data_set_multiple_signal}
    {D}_{K}:=\{\bar{\bm{l}}^{(+,k)},\bar{\bm{l}}^{(-,k)}\}_{k=1}^K,~~~\bar{\bm{l}}^{(\pm,k)}=(\bar{l}^{(\pm,k)}_1,...,\bar{l}^{(\pm,k)}_M),
\end{equation}
where the subscript denotes the index of the coordinate of $\bm{x}$.
By applying Algorithm~\ref{alg_classical_part} to ${D}_K$ coordinate-wisely, we obtain an estimate $\hat{\bm{x}}=(\hat{x}_1,...,\hat{x}_M)$ such that $\mathbb{E}\left[(\hat{x}_j-x_j)^2\right]=\mathcal{O}(\varepsilon^2)$ holds for any coordinate $j=1,2,...,M$.
Furthermore, by independently generating $\Upsilon_{\rm med}=\mathcal{O}(\log(M/\delta))$ copies of $\bm{\hat{x}}$ and taking the coordinate-wise median of them as $\hat{\bm{x}}^{(\rm med)}$,
we can ensure that $\hat{\bm{x}}^{\rm (med)}$ satisfies Eq.~\eqref{eq_lem:ideal_data_set_guarantee}.
In the following lemma, we summarize the output guarantee in the parallel multi-signal readout framework.

\begin{lem}[Parallel multi-signal readout]\label{lem:ideal_data_set_guarantee}
    For error $\varepsilon$, let $K=\mathcal{O}(\log(1/\varepsilon))$.
    Suppose that we can generate $\Upsilon_{\rm med}=\mathcal{O}(\log(M/\delta))$ independent data sets ${D}_K$ for the target signal $\bm{x}=(x_1,...,x_M)\in [-2,2]^M$ from the ideal separable parity measurements.
    Then, we can efficiently produce an estimate $\hat{\bm{x}}^{(\rm med)}\in [-\pi,\pi)^M$ for the signal $\bm{x}$ such that
    \begin{equation}\label{eq_lem:ideal_data_set_guarantee}
        {\rm Pr}[\|\hat{\bm{x}}^{(\rm med)}-\bm{x}\|_{\infty}\leq \varepsilon]\geq 1-\delta
    \end{equation}
    holds for a failure probability $\delta$.
    
\end{lem}
\begin{proof}
    We first note that $|x-y|\leq \frac{\pi+2}{\pi-2}d_{\mathbb{T}}(x,y)$ holds for any $x\in [-\pi,\pi)$ and any $y\in [-2,2]$.
    By appropriately choosing the constant factor of $K=\mathcal{O}(\log(1/\varepsilon))$, we can ensure $\mathbb{E}\left[(\hat{x}_j-x_j)^2\right]\leq \varepsilon^2/3$ for any coordinate $j$.
    Thus, the Chebyshev's inequality yields 
    ${\rm Pr}[|\hat{x}_j-x_j|> \varepsilon]\leq {1}/{3}$  for any $j$.
Now, we use the so-called median-trick (e.g., Refs.~\cite{huang2020predicting,van2023quantum}). 
That is, the median $\hat{k}^{(\rm med)}$ of $\Upsilon_{\rm med}$ independent random variable $\hat{k}^{(i)}$ such that ${\rm Pr}[|\hat{k}^{(i)}-k^*|>\varepsilon]\leq p$ for $0<p<1/2$ satisfies 
\begin{equation}
    {\rm Pr}\left[|\hat{k}^{(\rm med)}-k^*|>\varepsilon\right]\leq e^{-2\Upsilon_{\rm med}(1/2-p)^2}.
\end{equation}
Thus, by taking $\Upsilon_{\rm med}=18\ln(M/\delta)$ and using the union bound, 
we have
\begin{align}
    {\rm Pr}\left[\|\hat{\bm{x}}^{(\rm med)}-\bm{x}\|_{\infty}>\varepsilon\right]&=
    {\rm Pr}\left[\exists j:|\hat{{x}}_j^{(\rm med)}-{x}_j|>\varepsilon\right]\leq \sum_{j=1}^M{\rm Pr}\left[|\hat{{x}}_j^{(\rm med)}-{x}_j|>\varepsilon\right]\notag\\
    &\leq Me^{-2\Upsilon_{\rm med}(1/2-1/3)^2}=\delta.
\end{align}
\end{proof}

Lemma~\ref{lem:ideal_data_set_guarantee} can be modified to produce a nearly unbiased estimate.
This nearly unbiased version is used to construct an efficient ground-state tomography protocol.
To this end, we introduce a signal-shift operation: in the $k$-th stage, we apply the 1-qubit phase shifter $S_{P\tau_k \theta}$ to one qubit of the phase-shifted GHZ state before the measurement.
The resulting state before the measurement is given by 
\begin{equation}
    S_{P\tau_k \theta}\otimes I^{\otimes P-1}|{\rm GHZ}_P(\tau_k x)\rangle =\frac{\ket{0}^{\otimes P}+e^{iP\tau_k( x+\theta)}\ket{1}^{\otimes P}}{\sqrt{2}},
\end{equation}
which has the shifted signal $x+\theta$, instead of the original unknown $x$.
This signal-shift operation can be directly extended to multi-signal cases by introducing $\otimes_{j=1}^M S_{P\tau_k \theta_j}$ for $\bm{\theta}=(\theta_1,...,\theta_M)$.
By performing this operation, we can generate a data set $D_{K}$ for the target signal $\bm{x}+\bm{\theta}$ instead of $\bm{x}$. 
For simplicity, we write such a signal-shifted data set as $D_{K,\bm{\theta}}$.

Our procedure to obtain a nearly unbiased estimate $\hat{\bm{x}}^{(\rm cmed)}=(\hat{{x}}_1^{(\rm cmed)},...,\hat{{x}}_M^{(\rm cmed)})$ is as follows.
By using the separable parity measurement with the ideal joint phase shifter, we first independently generate data sets 
\begin{equation}
        \{{D}_{K,\bm{\theta}^{(r)}}\}_{r=1}^{\Upsilon_{\rm med}},~~~\bm{\theta}^{(r)}\sim {\rm Unif}([-\pi,\pi)^M),
\end{equation}
where $\bm{\theta}^{(r)}$ follows the uniform probability distribution over $[-\pi,\pi)^M$ and is independent.
For each coordinate $j=1,...,M$, we independently repeat the following procedure.
For simplicity, we write the $j$-th coordinate restriction of $D_{K,\bm{\theta}^{(r)}}$ as $D_{K}^{(r,j)}$.
For each $r=1,2,...,\Upsilon_{\rm med}$, we independently choose $\xi_j^{(r)}\in \{0,1\}$ with equal probability, and set 
\begin{equation}
            {z}_j^{(r)}=\tilde{z}_j^{(r)}-2\pi \lfloor {(\tilde{z}_j^{(r)}+\pi)}/{2\pi}\rfloor,~~~{\tilde{z}}^{(r)}_j:=(-1)^{\xi_j^{(r)}}(\mathcal{A}_1\circ \mathcal{R}^{\xi_j^{(r)}})(D_{K}^{(r,j)})-{\theta}^{(r)}_j,
\end{equation}
where $\mathcal{A}_1$ denotes the Algorithm~\ref{alg_classical_part}, $\mathcal{R}^{0}$ is the identity map, and $\mathcal{R}$ takes a data set ${D}_K$ with one signal and returns a new data set
\begin{equation}
    \mathcal{R}\left(\{\bar{{l}}^{(+,k)},\bar{{l}}^{(-,k)}\}_{k=1}^K\right)=\{{1}-\bar{{l}}^{(+,k)},\bar{{l}}^{(-,k)}\}_{k=1}^K.
\end{equation}
We then find the shortest arc in the unit circle that contains at least $(\Upsilon_{\rm med}+1)/2$ elements in $\{z_j^{(r)}\}_{r=1}^{\Upsilon_{\rm med}}$. 
If there are multiple shortest arcs, we select one uniform randomly.
Finally, we set $\hat{x}_j^{(\rm cmed)}$ as the middle point of the shortest arc in the domain $[-\pi,\pi)$.

% That is, for the arc length $|[a,b]_{\mathbb{T}}|$, defining,  
% \begin{equation}
%             [a^*_j,b^*_j]_{\mathbb T}:=\underset{[a,b]_\mathbb{T}\in \Lambda_j}{\rm argmin}~|[a,b]_\mathbb{T}|,
% \end{equation}
% \begin{equation}
%             \Lambda_{j}:=\left\{[a,b]_{\mathbb T}:a,b\in \{z_{j}^{(r)}\}_{j=1}^{\Upsilon_{\rm med}},~\left|[a,b]_{\mathbb T}\cap \{z_{j}^{(r)}\}_{j=1}^{\Upsilon_{\rm med}}\right|\geq \frac{\Upsilon_{\rm med}+1}{2}\right\}
%         \end{equation}
% we set the middle of the shortest arc $[a^*_j,b^*_j]_{\mathbb T}$ as $\hat{{x}}_j^{(\rm rmed)}\in [-\pi,\pi)$.

\begin{lem}[Parallel multi-signal readout with near unbiasedness]\label{lem:ideal_data_set_guarantee_unbiased}
    For error $\varepsilon \in {(0,1]}$, let us take $K=\mathcal{O}(\log(1/\varepsilon))$ and an odd number $\Upsilon_{\rm med}=\mathcal{O}(\log(M/\delta))$.
    Suppose that we can independently generate data sets 
    \begin{equation}
        \{{D}_{K,\bm{\theta}^{(r)}}\}_{r=1}^{\Upsilon_{\rm med}},~~~\bm{\theta}^{(r)}\sim {\rm Unif}([-\pi,\pi)^M)
    \end{equation}
    for the target signal $\bm{x}=(x_1,...,x_M)\in [-2,2]^M$ from the ideal separable parity measurements.
    Then, we can efficiently produce a random vector $\hat{\bm{x}}^{(\rm cmed)}\in [-\pi,\pi)^M$ for the signal $\bm{x}$ such that
    \begin{equation}\label{eq:lem2_guarantee}
        {\rm Pr}[\|\hat{\bm{x}}^{(\rm cmed)}-{\bm{x}}\|_{\infty}\leq \varepsilon]\geq 1-\delta,~~~\|\mathbb{E}[\hat{\bm{x}}^{(\rm cmed)}]-{\bm{x}}\|_{\infty}\leq 2\pi \frac{\delta}{M}
    \end{equation}
    for a failure probability $\delta$. 
    Furthermore, there exists a random vector $\hat{\bm{y}}\in [-3,3]^M$ with independent coordinates that is $\delta$-close to $\hat{\bm{x}}^{(\rm cmed)}$ in total variation distance and satisfies 
    \begin{equation}
        \|\hat{\bm{y}}-{\bm{x}}\|_{\infty}\leq \varepsilon~(\mbox{with probability one}),~~~\|\mathbb{E}[\hat{\bm{y}}]-{\bm{x}}\|_{\infty}\leq 4\pi\frac{\delta}{M}.
    \end{equation}
\end{lem}

\begin{proof}
    We first remark that conditioned on $\theta_{j}^{(r)}$, the data sets $\{D_K^{(r,j)}\}_{r,j}$ are independent. Since $\{\xi_j^{(r)}\}_{r,j}$ are also independent, $\{z_j^{(r)}\}_{r,j}$ and thus $\{\hat{x}_j^{(\rm cmed)}\}_j$ are independent, respectively.
    Let us define a $2\pi$-periodic function $\chi(t)\in (-\pi,\pi)$ such that $\chi(-t)=-\chi(t)$, $\chi(t)=t$ ($|t|<\pi$), and $\chi(\pm\pi)=0$.
    Also, we define $\Theta(x):=x-2\pi \lfloor (x+\pi)/(2\pi) \rfloor \in [-\pi,\pi)$, and this function differs from $\chi(t)$ only at $(2p+1)\pi$, $p\in\mathbb{Z}$.

    We show the first inequality in Eq.~\eqref{eq:lem2_guarantee}. 
    Let us fix ${\theta}_j^{(r)}$ and $\xi_j^{(r)}$.
    If $D_K^{(r,j)}$ is distributed by the ideal joint distribution for the target signal $\Theta(x_j+\theta_j^{(r)}) \in [-\pi,\pi)$, then $\mathcal{R}(D_K^{(r,j)})$ follows the same distribution with the target signal $-\Theta(x_j+\theta_j^{(r)})$.
    Thus, the probability for the event
    \begin{equation}
        d_{\mathbb{T}}(\mathcal{A}_1\circ \mathcal{R}^{\xi_j^{(r)}}(D_{K}^{(r,j)}),(-1)^{\xi_j^{(r)}}\Theta(x_j+\theta_j^{(r)}))>\varepsilon/2
    \end{equation}
    can be upper bounded by $1/3$, as well as the proof of Lemma~\ref{lem:ideal_data_set_guarantee}.
    Since 
    \begin{align}
        &d_{\mathbb{T}}(\mathcal{A}_1\circ \mathcal{R}^{\xi_j^{(r)}}(D_{K}^{(r,j)}),(-1)^{\xi_j^{(r)}}\Theta(x_j+\theta_j^{(r)}))\notag\\
        &=d_{\mathbb{T}}((-1)^{\xi_j^{(r)}}\mathcal{A}_1\circ \mathcal{R}^{\xi_j^{(r)}}(D_{K}^{(r,j)}),\Theta(x_j+\theta_j^{(r)}))\notag\\
        &=d_{\mathbb{T}}((-1)^{\xi_j^{(r)}}\mathcal{A}_1\circ \mathcal{R}^{\xi_j^{(r)}}(D_{K}^{(r,j)}),x_j+\theta_j^{(r)})\notag\\
        &=d_{\mathbb{T}}(z_j^{(r)},x_j)
    \end{align}
    holds, the probability for $d_{\mathbb{T}}(z_j^{(r)},x_j)>\varepsilon/2$ is at most $1/3$.
    We used the reflection symmetry $d_{\mathbb{T}}(a,b)=d_{\mathbb{T}}(-a,-b)$ in the first equality, $d_{\mathbb{T}}(a,b)=d_{\mathbb{T}}(a,\Theta(b))=d_{\mathbb{T}}(\Theta(a),b)$ in the second equality, and the rotation symmetry $d_{\mathbb{T}}(a,b)=d_{\mathbb{T}}(a+c,b+c)$ in the final equality.
    Let us consider a good event where at least $h=(\Upsilon_{\rm med}+1)/2$ elements in $\{z_j^{(r)}\}_{r=1}^{\Upsilon_{\rm med}}$ satisfy $d_{\mathbb{T}}(z_j^{(r)},x_j)\leq \varepsilon/2$.
    Then, the length of the shortest arc must be $\leq \varepsilon$ from the definition.
    Moreover, at least one $z_{j}^{(r_0)}\in \{z_j^{(r)}\}_{r=1}^{\Upsilon_{\rm med}}$ satisfying $d_{\mathbb{T}}(z_j^{(r_0)},x_j)\leq \varepsilon/2$ is on the shortest arc.
    Recalling that $\hat{x}_j^{(\rm cmed)}$ is the middle point of the shortest arc, 
    \begin{equation}
        d_{\mathbb{T}}(\hat{x}_j^{(\rm cmed)},x_j)\leq d_{\mathbb{T}}(\hat{x}_j^{(\rm cmed)},z_j^{(r_0)})+d_{\mathbb{T}}({x}_j,z_j^{(r_0)})\leq \varepsilon
    \end{equation}
    holds.    
    From the Hoeffding's inequality, the probability of the good event is at least $1-\exp(-2\Upsilon_{\rm med}(1/2-1/3)^2)$. 
    We note that due to the assumption of $x_j\in [-2,2]$ and $\varepsilon\le 1<\pi-2$, in the good event, 
    \begin{equation}
        d_{\mathbb{T}}(\hat{x}_j^{(\rm cmed)},x_j)\le \varepsilon~~\Rightarrow~~|\hat{x}_j^{(\rm cmed)}-x_j|\le \varepsilon
    \end{equation}
    holds.
    Thus, taking $\Upsilon_{\rm med}=18\ln(M/\delta)$ and using the union bound, we obtain the first inequality in Eq.~\eqref{eq:lem2_guarantee}.
    
    Next, we show the second part of Eq.~\eqref{eq:lem2_guarantee}.
    Suppose that we can prove that the probability distribution of $\chi(z_j^{(r)}-x_j)\in (-\pi,\pi)$ is symmetric for any $j$.
    Let $\hat{x}_j^{(\rm cmed)}$ be the final estimate when a set $\{z_j^{(r)}\}$ is realized.
    With the same probability, we get another set $\{\check{z}_j^{(r)}\}$ whose elements lie on the unit circle and are the reflections of the points $\{z_j^{(r)}\}$ across the line through $x_j$ and the origin of the circle.
    The final estimate ${\check{x}}_j^{(\rm cmed)}$ corresponding to $\{\check{z}_j^{(r)}\}$ clearly satisfies 
    \begin{equation}
        \chi({\check{x}}_j^{(\rm cmed)}-x_j)=-\chi({\hat{x}}_j^{(\rm cmed)}-x_j).
    \end{equation}
    This implies that the probability distribution of $\chi({\hat{x}}_j^{(\rm cmed)}-x_j)$ is also symmetric.
    We note that the probability of realizing $z_j^{(r)}-x_j=\pm \pi$ is zero because the random variable $\theta_j^{(r)}$ is uniformly distributed on the unit circle (the probability measure gives zero for only one point on the circle).
    Then, we have $\mathbb{E}[\chi(\hat{x}_j^{(\rm cmed)}-x_j)]=0$. Furthermore, 
    \begin{align}
        \left\|\mathbb{E}[\hat{\bm{x}}^{(\rm cmed)}-\bm{x}]\right\|_{\infty}&=\max_j\left|\mathbb{E}[\hat{x}_j^{(\rm cmed)}-x_j-\chi(\hat{x}_j^{(\rm cmed)}-x_j)]\right|\notag\\
        &=\max_j\left|\mathbb{E}[\{\hat{x}_j^{(\rm cmed)}-x_j-\chi(\hat{x}_j^{(\rm cmed)}-x_j)\}\bm{1}_{d_{\mathbb{T}}(\hat{x}_j^{(\rm cmed)},x_j)> \varepsilon}]\right|\notag\\
        &\leq 2\pi \max_j{\rm Pr}[d_{\mathbb{T}}(\hat{x}_j^{(\rm cmed)},x_j)> \varepsilon]\leq 2\pi\frac{\delta}{M},
    \end{align}
    where $\bm{1}_A$ denotes the indicator function for the event $A$.
    In the second inequality, we used the fact that if $d_{\mathbb{T}}(\hat{x}_j^{(\rm cmed)},x_j)\le \varepsilon<1$, then $\hat{x}_j^{(\rm cmed)}-x_j=\chi(\hat{x}_j^{(\rm cmed)}-x_j)$.

    We below prove the symmetry in the distribution of $\chi(z_j^{(r)}-x_j)\in (-\pi,\pi)$.
    The characteristic function of the random variable is 
    \begin{align}
        \mathbb{E}[e^{-it\chi(z_j^{(r)}-x_j)}]=\frac{1}{2}I_0(t)+\frac{1}{2}I_1(t),
    \end{align}
    where (we omit the index $r,j$ for simplicity)
    \begin{equation}
        I_0(t)=\int {\rm d}\mu({\theta})\int{\rm d}\mathcal{P}_{\Theta({x}+{\theta})}(D_K)~\exp\left({-it}\chi(\mathcal{A}_1(D_K)-\Theta({x}+{\theta})\right),
    \end{equation}
    \begin{equation}
        I_1(t)=\int {\rm d}\mu({\theta})\int{\rm d}\mathcal{P}_{\Theta({x}+{\theta})}(D_K)~\exp\left({-it}\chi(-(\mathcal{A}_1\circ\mathcal{R})(D_K)-\Theta({x}+{\theta}))\right).
    \end{equation}
    Here, $\mu$ denotes the uniform probability measure on $[-\pi,\pi)$, $\mathcal{P}_{\Theta(x+\theta)}$ denotes the ideal joint probability measure of the data set $D_K^{(r,j)}$ for the target signal $\Theta(x+\theta)\in [-\pi,\pi)$.
    We have
    \begin{align}
        I_1(t)&=\int {\rm d}\mu({\theta})\int{\rm d}\mathcal{P}_{-\Theta({x}+{\theta})}(D_K)~\exp\left({-it}\chi(-\mathcal{A}_1(D_K)-\Theta({x}+{\theta}))\right)\notag\\
        &=\int {\rm d}\mu({\phi})\int{\rm d}\mathcal{P}_{\phi}(D_K)~\exp\left({-it}\chi(-\mathcal{A}_1(D_K)+\phi)\right)\notag\\
        &=I_{0}(-t).
    \end{align}
    Thus, $\mathbb{E}[e^{-it\chi(z_j^{(r)}-x_j)}]=\mathbb{E}[e^{it\chi(z_j^{(r)}-x_j)}]$ holds for any $t$, which implies the desired symmetry.

    Finally, we prove the existence of the random vector $\hat{\bm{y}}=(\hat{y}_1,...,\hat{y}_M)$ by construction.
    For each coordinate $\hat{{x}}^{(\rm cmed)}_j$, we define $\hat{y}_j$ as
    \begin{equation}
        \hat{y}_j=\begin{cases}
            \hat{x}_j^{(\rm cmed)}&\mbox{if}~|\hat{x}_j^{(\rm cmed)}-x_j|\leq \varepsilon\\
            x_j+\varepsilon&\mbox{if}~\hat{x}_j^{(\rm cmed)}>x_j+\varepsilon\\
            x_j-\varepsilon&\mbox{if}~\hat{x}_j^{(\rm cmed)}<x_j-\varepsilon
        \end{cases}.
    \end{equation}
    This immediately yields $\|\hat{\bm{y}}-{\bm{x}}\|_{\infty}\leq \varepsilon$ with probability one.
    We note that $\hat{\bm{y}}$ has independent coordinates. 
    We now prove the closeness in total variation distance. From its definition, we obtain
    \begin{align}
        &\max_B\left|{\rm Pr}[\hat{\bm{x}}^{(\rm cmed)}\in B]-{\rm Pr}[\hat{\bm{y}}\in B]\right|   \label{eq:tvd_eval_in_lem2}\\
        &= \max_B\left|\mathbb{E}\left[(\bm{1}_{\{\hat{\bm{x}}^{(\rm cmed)}\in B\}}-\bm{1}_{\{\hat{\bm{y}}\in B\}})\bm{1}_{\{\|\hat{\bm{x}}^{(\rm cmed)}-\bm{x}\|_{\infty}>\varepsilon\}}\right]\right|\notag\\
        &\leq \max_B\mathbb{E}\left[\bm{1}_{\{\|\hat{\bm{x}}^{(\rm cmed)}-\bm{x}\|_{\infty}>\varepsilon\}}\right]\leq \delta,
    \end{align}
    where the maximization is taken over all possible event.
    Also, we can show that
    \begin{align}
        \left|\mathbb{E}[\hat{y}_j-x_j]\right|&\leq \left|\mathbb{E}[\hat{y}_j-\hat{x}^{(\rm cmed)}_j]\right|+\left|\mathbb{E}[x_j-\hat{x}^{(\rm cmed)}_j]\right|\notag\\
        &= \left|\mathbb{E}[(\hat{y}_j-\hat{x}^{(\rm cmed)}_j)\bm{1}_{\{|\hat{{x}}_j^{(\rm cmed)}-{x}_j|>\varepsilon\}}]\right|+\left|\mathbb{E}[x_j-\hat{x}^{(\rm cmed)}_j]\right|\notag\\
        &\leq 2\pi \mathbb{E}[\bm{1}_{\{|\hat{{x}}_j^{(\rm cmed)}-{x}_j|>\varepsilon\}}]+\left|\mathbb{E}[x_j-\hat{x}^{(\rm cmed)}_j]\right|\notag\\
        &\leq 2\pi\frac{\delta}{M}+2\pi\frac{\delta}{M} = 4\pi\frac{\delta}{M},
    \end{align}
    where in the third line, we used $|\hat{y}_j-\hat{x}_j^{(\rm cmed)}|< 2\pi$ holds with probability one.
\end{proof}

\section{Parallel catalytic readout protocols}
\subsection{Multiple observables estimation}
We instantiate the parallel multi-signal readout framework in
Section~\ref{sec:parallel_multi_signal_readout} for the ground-state expectation values 
\begin{equation}
x_j=\langle\psi_0|O_j|\psi_0\rangle\equiv \langle O_j\rangle ,\qquad j=1,\ldots,M,    
\end{equation}
where \(|\psi_0\rangle\) is the unique ground state of the problem Hamiltonian
\(H\). 
The basic setup for quantum devices is the same as the parallel readout framework. 
That is, there are $P$ local quantum devices that can share multiple $P$-qubit GHZ states (see Fig.~\ref{fig:local_signal_readout_framework}).
In addition, we consider the following situation: each local device has access to the controlled Hamiltonian evolution for the problem Hamiltonian $H$ and a single copy of the ground state $\ket{\psi_0}$ which will be reused throughout the protocol.
Recall from Lemma~\ref{lem:ideal_data_set_guarantee} that, in order to obtain a precise estimate of $\bm{x}=(x_1,...,x_M)$, it suffices to generate $\Upsilon_{\rm med}$ classical data sets satisfying (1) the same format as ${D}_K$ in Eq.~\eqref{eq:data_set_multiple_signal} and (2) the underlying probability distribution is sufficiently close to the ideal one.
In the following, we describe how to generate such data sets under that setup.

The central idea is to use ground states as catalysts for approximately implementing the joint phase shifters Eq.~\eqref{main_eq:target_phase_shifter_x} for the signals $x_j=\langle O_j\rangle$. 
\begin{lem}[Catalytic joint phase shifter]\label{lem:catalytic_joint_phase_shifter}
    Let $L_1,...,L_R$ be $n$-qubit observables with the operator norm at most one, and let $\delta_{\rm phase}$ be an error parameter.
    We assume access to a controlled time evolution operator of a $\Delta$-gap Hamiltonian $H$ with the unique ground state $|\psi_0\rangle$.
    Suppose that $0<\sigma\leq R$ and that either $\sigma=R$, or
    \begin{equation}
        {\sigma^2}\geq \frac{\|S\|}{2}\ln\left(\frac{16}{\delta_{\rm phase}^2}\frac{{\rm tr}[S]}{\|S\|}\right)
    \end{equation}
    holds for $S:= \sum_{j=1}^R L_j^2$.
    Then, we can implement an $(R+n+n_{\rm anc})$-qubit unitary $V_{\rm phase}$ such that its $P$-parallel application $V_{\rm phase}^{\otimes P}$ maps $R$ GHZ states
    \begin{equation}
        |{\rm GHZ}_P\rangle^{\otimes R}\otimes (|\psi_0\rangle\otimes |0^{n_{\rm anc}}\rangle)^{\otimes P}
    \end{equation}
    to a quantum state that is $\delta_{\rm phase}$-close to 
    \begin{equation}
    e^{-i\langle \psi_0|\sum_{j}L_j|\psi_0\rangle P\tau/2}\left(\bigotimes_{j=1}^R |{\rm GHZ}_P(\tau \langle \psi_0| L_j|\psi_0\rangle ) \rangle \right)\otimes (|\psi_0\rangle\otimes |0^{n_{\rm anc}}\rangle)^{\otimes P}
    \end{equation}
    in the Euclidean distance.
    The implementation of one $V_{\rm phase}$ requires 
    \begin{itemize}
        \item $\mathcal{\tilde{O}}\left(\Delta^{-1}({\sigma |\tau|}+1)\right)$ controlled Hamiltonian evolution time for $H$,
        \item $\mathcal{\tilde{O}}\left({R(|\tau|}+\sigma^{-1})\right)$ queries to a controlled $(1,n_{\rm obs},0)$-block-encoding of $L_j$ for $j=1,...,R$, and
        \item $n_{\rm anc}=n_{\rm obs}+\mathcal{{O}}(\log(R))+\mathcal{{O}}\left(\log\left(\frac{\|H\|P\sigma |\tau|}{\Delta \delta_{\rm phase}}\right)\right)$ ancilla qubits.
    \end{itemize}
    Here, $\tilde{\mathcal{O}}$ hides poly-logarithmic factors in $P,R,\sigma,|\tau|,$ and $\delta^{-1}_{\rm phase}$.
\end{lem}
\noindent
We will prove this lemma in Section~\ref{sec:algorithmic_construction} 
based on a joint multidimensional phase-kickback construction.
The unitary gate $V_{\rm phase}$ for $R=M$ and $\{L_j\}_{j=1}^R=\{O_j\}_{j=1}^M$ approximately implements the required joint phase shifter for the signals $x_j=\langle O_j\rangle$.
In this approximate implementation, the ground state is used as a catalyst, that is, it gets only a small disturbance and is almost preserved in its own register by the action of $V_{\rm phase}$.
This enables us to sequentially generate the phase-shifted GHZ states by reusing the single ground-state copy; these GHZ states yield desirable data sets via the separable parity measurements explained below in detail.

The data generation process is illustrated in Fig.~\ref{fig:proposed_measurement_protocol}. 
There are sequential $k=1,2,..., K$ stages and in each stage, we generate $\Upsilon_{\rm med}$ average parities $\bm{\bar{l}}^{(+,k)}$ and $\bm{\bar{l}}^{(-,k)}$ by performing $2v_k\Upsilon_{\rm med}$ separable parity measurements in total. 
Let us describe one parity measurement in the $k$-th stage.
We first distribute $M$ copies of the $P$-qubit GHZ states $|{\rm GHZ}_P\rangle $, one for each observable $O_j$.
The $p$-th qubit of every GHZ state is stored in the $p$-th local device. 
If the $P$ local devices share sufficiently many GHZ states in advance, this distribution step can be skipped. 
Next, each local device applies $V_{\rm phase}^{[k]}$ to its local $M$ qubits from those GHZ states together with its ground-state register and ancilla register.
Here, $V^{[k]}_{\rm phase}$ is the unitary gate constructed in Lemma~\ref{lem:catalytic_joint_phase_shifter} with 
% $R=M$, ${L_j}={O_j}$, and $\tau_k = 2^{k-1}/P$.
\begin{equation}
    R=M,~~{L_j}={O_j}~(j=1,2,...,M),~~\tau_k = 2^{k-1}/P.
\end{equation}
These operations are performed in parallel over the $P$ local devices. 
Importantly, we do not prepare a fresh ground state before applying $V^{[k]}_{\rm phase}$; the same ground-state register is reused throughout the protocol.
Finally, we perform the single-qubit measurements on the GHZ-state registers and gather the local measurement outcomes by classical communication.
This yields a single parity vector ${\bm{l}}^{(+,k)}$ or ${\bm{l}}^{(-,k)}$. 
Repeating this primitive $2v_k \Upsilon_{\rm med}$ times and taking the sample average of size-$v_k$ patch gives $\Upsilon_{\rm med}$ samples of $\bm{\bar{l}}^{(+,k)}$ and $\bm{\bar{l}}^{(-,k)}$.
After all stages, the generation process outputs $\Upsilon_{\rm med}$ samples of a classical data array with the same format as the ideal data sets ${D}_K$ required in Lemma~\ref{lem:ideal_data_set_guarantee}.

\begin{figure}[bt]
 \centering
 \includegraphics[scale=0.66]{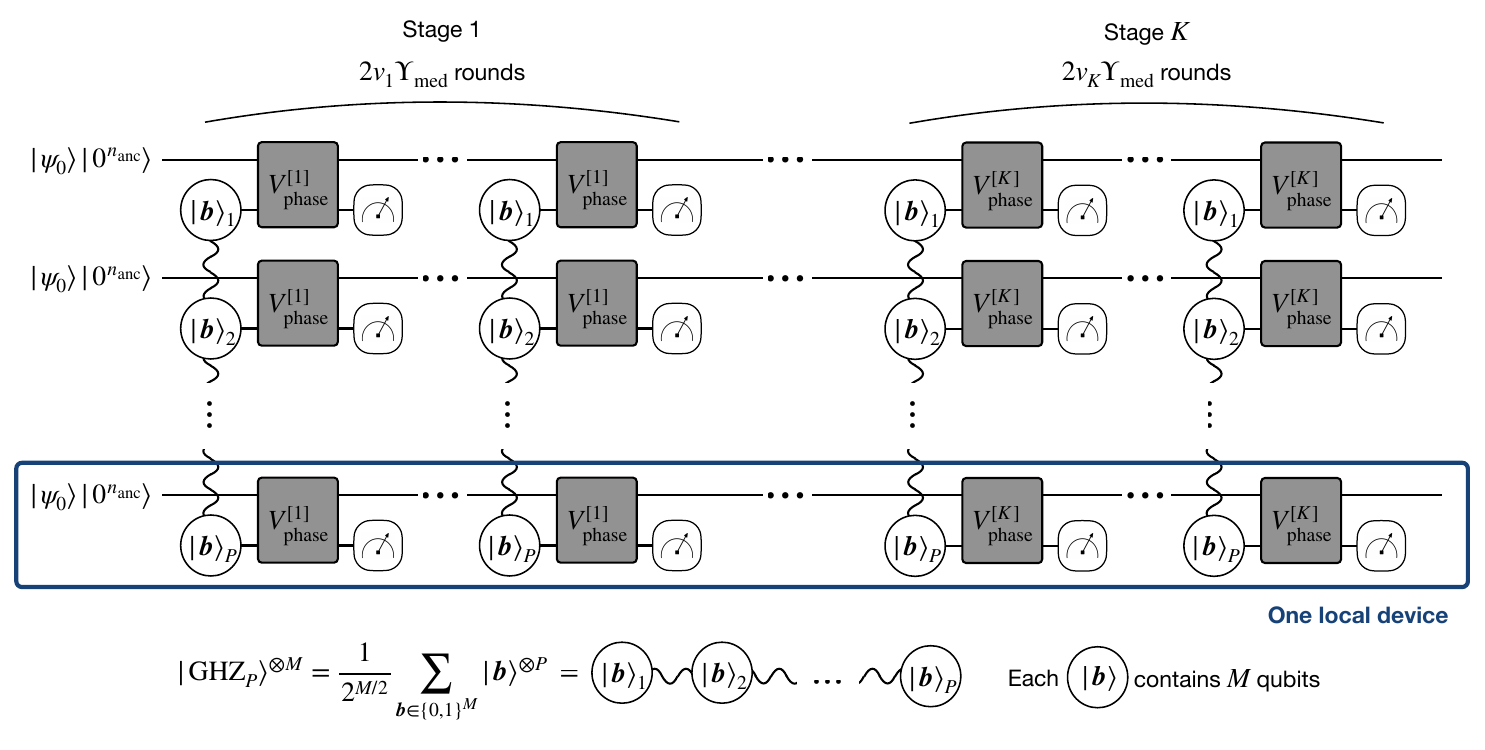}
 \caption{Data generation process.
 We perform single-qubit measurements on the local $M$ qubits from shared GHZ states at the meter diagram.
 }
 \label{fig:proposed_measurement_protocol}
\end{figure}

Combining this data generation process with Lemma~\ref{lem:ideal_data_set_guarantee}, we prove the following theorem.
\begin{thm}
    [Parallel catalytic readout from ground states]\label{thm:main_theorem_for_multiple_observable_est}
    Given $P$ copies of the unique ground state $\ket{\psi_0}$ of a $\Delta$-gap Hamiltonian and observables $O_1,...,O_M$ with the norm at most one, there is a measurement protocol that with probability at least $1-\delta$, produces element-wise estimates of all $\langle \psi_0|O_j|\psi_0\rangle$ with additive error $\varepsilon$ and returns a state $\delta$-close to $|\psi_0\rangle^{\otimes P}$ in trace distance.
    This protocol uses 
    \begin{itemize}
        \item $\mathcal{\tilde{O}}\left(\frac{1}{\Delta}\left(\frac{\sigma}{ P\varepsilon}+1\right)\right)$ sequential time for controlled Hamiltonian evolution of $H$, and
        \item $\mathcal{\tilde{O}}\left(\frac{M}{P\varepsilon}+\frac{M}{\sigma}\right)$ sequential queries to a controlled block-encoding of $O_j$ ($j=1,2,...,M$)
    \end{itemize}
    with $P$-way parallelism as illustrated in Fig.~\ref{fig:proposed_measurement_protocol}, where 
    \begin{equation}
        \sigma=\min\left\{M,\mathcal{O}\left({\|S_M\|^{1/2}\log^{1/2}\left[\frac{1}{\delta}\frac{{\rm tr}[S_M]}{\|S_M\|}\log \left(\frac{M}{\delta}\right)\log \left(\frac{1}{\varepsilon}\right)\right]}\right)\right\},~~~S_M:=\sum_{j=1}^M O_j^2.
        \label{eq:thmS4_sigma}
    \end{equation}
    Furthermore, after sharing $\mathcal{\tilde{O}}(M)$ $P$-qubit GHZ states, this measurement protocol can be performed on separable $P$ quantum devices, each of which has one ground state copy.
    Here, $\tilde{\mathcal{O}}$ hides poly-logarithmic factors in $P,M,\sigma,\varepsilon^{-1},$ and $\delta^{-1}$.
\end{thm}

We note that $\|S_M\|$ and ${\rm tr}[S_M]/\|S_M\|$ become
\begin{equation} 
\|S_M\|= M,~~~{\rm tr}[S_M]/\|S_M\|={\rm dim}(H)
\end{equation}
in the worst case, e.g., a set of Pauli strings, where ${\rm dim}(H)$ denotes the dimension of Hamiltonian. 
Hence, for typical regimes on $\varepsilon,\delta$, the factor $\sigma$ becomes
\begin{equation}
    \sigma=\mathcal{O}\left({\sqrt{M}\log^{1/2}\left[\frac{{\rm dim}(H)}{\delta}\log \left(\frac{M}{\delta}\right)\log \left(\frac{1}{\varepsilon}\right)\right]}\right),
\end{equation}
and the worst-case 
elapsed Hamiltonian time
% sequential total evolution time 
is 
\begin{equation}
    \mathcal{\tilde{O}}\left(\frac{1}{\Delta}\left(\frac{\sqrt{M}}{P\varepsilon}\log^{1/2}(\dim(H))+1\right)\right).
\end{equation}

\begin{proof}[Proof of Theorem~\ref{thm:main_theorem_for_multiple_observable_est}]
The ideal data set including the final median step in Lemma~\ref{lem:ideal_data_set_guarantee} is generated by the separable parity measurements on the following quantum state
\begin{equation}\label{eq:target_probe_state}
    |{\rm probe}_K\rangle:=\bigotimes_{k=1}^{K} \bigotimes_{j=1}^{R^{[k]}} |{\rm GHZ}_P(\tau_k x^{[k]}_j)\rangle,~~~R^{[k]}=2v_kM\Upsilon_{\rm med},
\end{equation}
where $\bm{x}^{[k]}$ is the vector obtained by concatenating $2v_k \Upsilon_{\rm med}$ copies of $\bm{x}$.
Let $\mathcal{P}$ denotes the joint probability distribution of all parity measurements on $|{\rm probe}_K\rangle$.
To analyze the actual probability distribution $\mathcal{P}'$, we defer all the measurements in Fig.~\ref{fig:proposed_measurement_protocol} to the end and consider the final state $|\Psi_{\rm final}\rangle$ before the deferred measurement.
Then, we have
 \begin{equation}\label{eq:lemS4_EDist_bound}
        \left\||\Psi_{\rm final}\rangle -\left(\ket{0^{n_{\rm anc}}}\ket{\psi_0}\right)^{\otimes P}|{\rm probe}_K\rangle\right\|\leq 2\delta_{\rm phase}\Upsilon_{\rm med}\sum_{k=1}^{K} v_k = \delta^2/4,
\end{equation}
where we take $\delta_{\rm phase}=\delta^2/(8\Upsilon_{\rm med}\sum_k v_k)$.
Since the Euclidean distance upper bounds the trace distance and thus the total variation distance of any measurements, we obtain $(1/2)\|\mathcal{P}-\mathcal{P'}\|_1\leq \delta^2/4$.
Thus, from Lemma~\ref{lem:ideal_data_set_guarantee}, we have an estimate $\hat{x}_j$ such that
\begin{equation}\label{eq:output_confidence_obsest}
    {\rm Pr}[\max_{j}|\hat{x}_j-\langle O_j\rangle |\leq \varepsilon]\geq 1-(\delta/4+\delta^2/4)\geq 1-\delta/2.
\end{equation}

Let $\bm{z}$ be the collection of raw bit strings in the entire parity measurement.
The $\Upsilon_{\rm med}$ data sets are generated by this raw data $\bm{z}$ by classical post-processing to compute the parities.
The entire measurement can be seen as the following quantum channel 
\begin{equation}
    |\Psi_{\rm final}\rangle \mapsto \sum_{\bm{z}}\mathcal{P}'(\bm{z})|\bm{z}\rangle \langle \bm{z}|\otimes \tau_{\bm{z}},
\end{equation}
\begin{equation}
    \left(\ket{0^{n_{\rm anc}}}\ket{\psi_0}\right)^{\otimes P}|{\rm probe}_K\rangle \mapsto \sum_{\bm{z}}\mathcal{P}(\bm{z})|\bm{z}\rangle \langle \bm{z}|\otimes \psi_{0,\rm anc}^{\otimes P},
\end{equation}
where $\tau_{\bm{z}}$ is a post-measurement state, $\psi_{0,\rm anc}=\ket{0^{n_{\rm anc}}}\bra{0^{n_{\rm anc}}}\otimes \ket{\psi_0}\bra{\psi_0}$, and the first register denotes the classical register storing the raw measurement outcomes.
Let $D_{\rm tr}(A,B)$ be the trace distance between $A$ and $B$.
We now show
\begin{equation}\label{eq:return_highprob_ineq}
    \Pr_{{\bm{z}\sim \mathcal{P}'}}
    \left[D_{\rm tr}({\rm tr}_{\rm anc}[\tau_{\bm{z}}],|\psi_0\rangle\langle\psi_0|^{\otimes P})>\delta\right]\leq  \delta/2.
\end{equation}
From Markov's inequality, it is sufficient to show that
\begin{equation}\label{eq:return_highprob_ineq_markov}
    \sum_{\bm{z}}\mathcal{P}'(\bm{z})D_{\rm tr}({\rm tr}_{\rm anc}[\tau_{\bm{z}}],|\psi_0\rangle\langle\psi_0|^{\otimes P})\leq \delta^2/2.
\end{equation}
The trace-distance contractivity and the triangle inequality yield
\begin{equation}
\begin{aligned}
\mathcal{P}'(\bm{z})D_{\rm tr}({\rm tr}_{\rm anc}[\tau_{\bm{z}}],|\psi_0\rangle\langle\psi_0|^{\otimes P})
&\leq \mathcal{P}'(\bm{z})
D_{\mathrm{tr}}
\left(
\tau_{\bm{z}},
\psi_{0,\mathrm{anc}}^{\otimes P}
\right)
\\
&=
\frac12
\left\|
\mathcal{P}'(\bm{z})\tau_{\bm{z}}
-
\mathcal{P}'(\bm{z})\psi_{0,\mathrm{anc}}^{\otimes P}
\right\|_1
\\
&\leq
\frac12
\left\|
\mathcal{P}'(\bm{z})\tau_{\bm{z}}
-
\mathcal{P}(\bm{z})\psi_{0,\mathrm{anc}}^{\otimes P}
\right\|_1
+
\frac12
\left|
\mathcal{P}'(\bm{z})-\mathcal{P}(\bm{z})
\right|.
\end{aligned}
\end{equation}
From the block-diagonal structure,
\begin{equation}
    \frac12
\sum_{\bm{z}}\left\|
\mathcal{P}'(\bm{z})\tau_{\bm{z}}
-
\mathcal{P}(\bm{z})\psi_{0,\mathrm{anc}}^{\otimes P}
\right\|_1=D_{\rm tr}\left(\sum_{\bm{z}}\mathcal{P}'(\bm{z})|\bm{z}\rangle \langle \bm{z}|\otimes \tau_{\bm{z}},\sum_{\bm{z}}\mathcal{P}(\bm{z})|\bm{z}\rangle \langle \bm{z}|\otimes \psi_{0,\rm anc}^{\otimes P}\right)
\end{equation}
holds, and since the Euclidean distance upper bounds the trace distance, 
\begin{equation}
    \frac12
\sum_{\bm{z}}\left\|
\mathcal{P}'(\bm{z})\tau_{\bm{z}}
-
\mathcal{P}(\bm{z})\psi_{0,\mathrm{anc}}^{\otimes P}
\right\|_1\leq \left\||\Psi_{\rm final}\rangle -\left(\ket{0^{n_{\rm anc}}}\ket{\psi_0}\right)^{\otimes P}|{\rm probe}_K\rangle\right\|\leq \delta^2/4.
\end{equation}
Furthermore, the total variation distance between $\mathcal{P}'$ and $\mathcal{P}$ is also bounded by $\delta^2/4$.
Hence, we complete the proof of Eq.~\eqref{eq:return_highprob_ineq_markov}, thus Eq.~\eqref{eq:return_highprob_ineq}.
Combining Eqs.~\eqref{eq:return_highprob_ineq} and~\eqref{eq:output_confidence_obsest}, and using the union bound, we have
\begin{equation}
    \Pr_{\bm{z}\sim \mathcal{P}'}\left[\left\{\max_{j}|\hat{x}_j-\langle O_j\rangle |\leq \varepsilon\right\}\&\left\{D_{\rm tr}\left({\rm tr}_{\rm anc}[\tau_{\bm{z}}],|\psi_0\rangle\langle\psi_0|^{\otimes P}\right)\leq \delta\right\}\right]\geq 1-\delta.
\end{equation}

We finally calculate the implementation cost.
In each local device, the sequential total evolution time for the controlled Hamiltonian evolution for $H$ is given by
\begin{align}
    \sum_{k=1}^K 2v_k\Upsilon_{\rm med} \cdot \mathcal{\tilde{O}}(\Delta^{-1}(\sigma|\tau_k|+1))&=
    \mathcal{\tilde{O}}\left(\frac{\Upsilon_{\rm med}}{\Delta}\left[\frac{\sigma}{P}\sum_{k=1}^K 2v_k 2^{k-1}+\sum_{k=1}^K 2{v_k}\right]\right)\notag\\
    &=\mathcal{\tilde{O}}\left(\frac{\Upsilon_{\rm med} }{\Delta }\left[\frac{\sigma}{P}2^K+K^2\right]\right)=\mathcal{\tilde{O}}\left(\frac{1}{\Delta}\left(\frac{\sigma}{ P\varepsilon}+1\right)\right)
\end{align}
where in the second equality we used Eq.~\eqref{eq:rpe_query_relation} and $\Upsilon_{\rm med}=\mathcal{O}(\log(M/\delta))$.
Similarly, the sequential total queries to a controlled block-encoding of $O_j$ is $\mathcal{\tilde{O}}(\frac{M}{P\varepsilon}+\frac{M}{\sigma})$ for $j=1,2,...,M$.
\end{proof}

\begin{lem}
    [Nearly unbiased parallel catalytic readout]\label{thm:main_theorem_for_multiple_observable_est_unbiased}
    Under the same setting as Theorem~\ref{thm:main_theorem_for_multiple_observable_est}, the measurement protocol can be modified to produce a random vector $\hat{\bm{x}}=(\hat{{x}}_1,...,\hat{{x}}_M)\in[-\pi,\pi)^M$ while preserving all the guarantees including the ground-state return and resource bounds of Theorem~\ref{thm:main_theorem_for_multiple_observable_est}.
    In addition, the estimator $\hat{\bm{x}}$ is nearly unbiased:
    \begin{equation}\label{thm:main_theorem_for_multiple_observable_est_unbiased_eq1}
        \max_{j\in [M]}|\mathbb{E}[\hat{x}_j]-\langle \psi_0|O_j|\psi_0\rangle|\leq \frac{\pi\delta^2}{2}+\frac{\pi\delta}{2M}.
    \end{equation}
    Furthermore, there exists a random vector $\hat{\bm{y}}=(\hat{{y}}_1,...,\hat{{y}}_M)\in [-3,3]^M$ with independent coordinates that is $(\delta+\delta^2)/4$-close to $\hat{\bm{x}}$ in total variation distance and satisfies 
    \begin{equation}\label{thm:main_theorem_for_multiple_observable_est_unbiased_eq2}
        \max_{j\in [M]}|\hat{{y}}_j-\langle \psi_0|O_j|\psi_0\rangle|\leq \varepsilon~(\mbox{with probability one}),~~~\max_{j\in [M]}|\mathbb{E}[\hat{{y}}_j]-\langle \psi_0|O_j|\psi_0\rangle|\leq \frac{\pi\delta}{M}.
    \end{equation}
\end{lem}
\begin{proof}
    The proof strategy is almost the same as in Theorem~\ref{thm:main_theorem_for_multiple_observable_est}, but we use Lemma~\ref{lem:ideal_data_set_guarantee_unbiased}, rather than Lemma~\ref{lem:ideal_data_set_guarantee}.
    The required data set can be generated by the circuit Fig.~\ref{fig:proposed_measurement_protocol} with single-qubit rotation gates to encode the random phases $\{\bm{\theta}^{(r)}\}$.
    These are unitary gates acting finally on the probe registers, and thus, the Euclidean-distance bound in Eq.~\eqref{eq:lemS4_EDist_bound} and the subsequent proof remain unchanged.
    Since the output $\bm{\hat{x}}$ after performing the classical post-processing in Lemma~\ref{lem:ideal_data_set_guarantee_unbiased} follows a probability distribution that is $\delta^2/4$-close to the ideal distribution in total variation distance, we can directly obtain Eqs.~\eqref{thm:main_theorem_for_multiple_observable_est_unbiased_eq1} and~\eqref{thm:main_theorem_for_multiple_observable_est_unbiased_eq2}.
\end{proof}

The joint multidimensional phase kickback of Lemma~\ref{lem:catalytic_almost_linear_block_hamiltonian} underpins the catalytic joint phase shifter in Lemma~\ref{lem:catalytic_joint_phase_shifter} and hence the above catalytic readout protocols.
We conclude this subsection by combining the same phase-kickback construction with quantum gradient estimation (QGE) to obtain an alternative catalytic readout protocol for the single-copy setting $P=1$. 
We first briefly review the QGE algorithm~\cite{jordangradest,gilyen2019optimizing}.

In the QGE algorithm, we aim for estimating the gradient $\bm{g}:=({g}_1,...,g_M)$ ($\|\bm{g}\|_{\infty}\leq 1/3$) of an approximately linear function $h(\bm{x}):\mathbb{R}^M\to \mathbb{R}$ at $\bm{x}=\bm{0}$ by querying an oracle $O_h$ of $h(\bm{x})$.
To evaluate the function coherently, we label the computational basis as a grid point around $\bm{x}=\bm{0}$.
Specifically, we discretize a one-dimensional line into $2^p$ points and write $G_p$ as the set of all these points:
\begin{equation}
    G_p:=\left\{\frac{B}{2^p}-\frac{1}{2}+\frac{1}{2^{p+1}}:B\in \{0,1,...,2^{p}-1\}\right\}.
\end{equation}
Later $p$ determines the estimation precision.
The QGE algorithm prepares the following state by using the oracle $O_h$ only $\mathcal{O}(2^p)$ times:
\begin{align}
    \frac{1}{\sqrt{2^{pM}}}\sum_{\boldsymbol{x}\in G_p^M} e^{2\pi i2^p h(\boldsymbol{x})}|\boldsymbol{x}\rangle\approx \left(\frac{1}{\sqrt{2^{p}}}\sum_{{x}_1} e^{2\pi i2^p{g}_1\cdot {x}_1}|{x}_1\rangle \right)\otimes \left(\frac{1}{\sqrt{2^{p}}}\sum_{{x}_2} e^{2\pi i2^p{g}_2\cdot {x}_2}|{x}_2\rangle\right) \otimes \cdots.
\end{align}
Similar to quantum phase estimation, performing the computational basis measurement after $M$ tensor products of the inverse quantum Fourier transform ${\rm QFT}_{G_p}^\dagger$
\begin{equation}
    {\rm QFT}_{G_p}:\ket{y}\mapsto \frac{1}{\sqrt{2^p}}\sum_{z\in G_p}e^{2\pi i 2^p yz}\ket{z},
\end{equation}
we obtain an estimate of gradient $g_j$ within $3/2^p$ additive error with failure probability $\leq 1/4+\xi$ for any $j\in [M]$, if the approximation error remains at most $\xi$ in Euclidean distance~\cite{gilyen2019optimizing,van2023quantum}.

\begin{thm}[Catalytic readout via gradient estimation]
\label{lem:Catalytic_gradient_readout}
Let $\ket{\psi_0}$ be the unique ground state of a Hamiltonian $H$ with spectral gap $\Delta$, and let
$O_1,\ldots,O_M$ be observables satisfying
$\|O_j\|\leq 1$.
Define
\begin{equation}
    \mu_j
    :=
    \bra{\psi_0}O_j\ket{\psi_0},
    \qquad
    \bm \mu
    :=
    (\mu_1,\ldots,\mu_M),
\end{equation}
Given one copy of $\ket{\psi_0}$, there exists a catalytic gradient-estimation protocol that with probability at least $1-\delta$, produces a random vector
$\widehat{\bm g}
   \in    [-3/2,3/2]^M$ such that $\left\|
        \widehat{\bm g}
        -
        \bm \mu
    \right\|_\infty
    \leq
    \varepsilon$
and returns a state $\delta$-close to $\ket{\psi_0}$ in trace distance.
This protocol uses
\begin{itemize}
    \item $\widetilde{\mathcal O}\left[
        \frac{1}{\Delta}
        \left(
            \frac{\sigma_{\rm grad}}{\varepsilon}
            +
            1
        \right)
    \right]$ sequential time for controlled Hamiltonian evolution of $H$, and 
    \item $\widetilde{\mathcal O}\left(
        \frac{M}{\varepsilon}
        +
        \frac{M}{\sigma_{\rm grad}}
    \right)$ sequential queries to a controlled block-encoding of each $O_j$,
\end{itemize}
where
\begin{equation}
    \sigma_{\rm grad}
    =
    \min\left\{
        M,\,
        \mathcal O\left(
            \|S_M\|^{1/2}
            \log^{1/2}\left[
                \frac{1}{\delta}
                \frac{
                    \operatorname{tr}[S_M]
                }{
                    \|S_M\|
                }\log\left(\frac{M}{\delta}\right)
            \right]
        \right)
    \right\},~~~S_M
    :=
    \sum_{j=1}^{M}O_j^2.
    \label{eq:gradient_sigma_scaling}
\end{equation}
Here, $\widetilde{\mathcal O}$ hides poly-logarithmic factors in $M$, $\sigma_{\rm grad}$, $\varepsilon^{-1}$, and
$\delta^{-1}$.
\end{thm}

\begin{proof}
Let $\Upsilon_{\rm med}$ be an odd integer satisfying
$\Upsilon_{\rm med}=\mathcal O(\log({M}/{\delta}))$.
We take
\begin{equation}
    \xi
    :=
    \frac{\delta^2}{16},
    \qquad
    \eta
    :=
    \frac{\xi}{\Upsilon_{\rm med}},
    \qquad
    \delta_{\rm obs}
    =
    \mathcal{O}(\eta^2).
    \label{eq:gradient_internal_errors}
\end{equation}
We first describe the ideal gradient-estimation experiment. Let 
\begin{equation}
    p
    :=
    \left\lceil
        \log_2\left(
            \frac{9}{\varepsilon}
        \right)
    \right\rceil,
    \qquad
    \bm\gamma
    :=
    \frac{\bm\mu}{3}.
    \label{eq:gradient_resolution}
\end{equation}
For $\bm B=(B_1,\ldots,B_M)\in\{0,\ldots,2^p-1\}^M$, we write $x_j(\bm B)
    :=
    \frac{B_j}{2^p}
    -
    \frac12
    +
    \frac{1}{2^{p+1}}$.
Consider the ideal phase state
\begin{equation}
    \ket{\Phi_{\bm\gamma}}
    :=
    \frac{1}{\sqrt{2^{pM}}}
    \sum_{\bm B\in\{0,\ldots,2^p-1\}^M}
    e^{
        2\pi i2^p 
            \bm x(\bm B)\cdot 
            \bm\gamma
    }
    \ket{\bm B}\equiv \frac{1}{\sqrt{2^{pM}}}
    \sum_{\bm x\in G_p^M}
    e^{
        2\pi i2^p 
            \bm x\cdot 
            \bm\gamma
    }
    \ket{\bm x}.
    \label{eq:ideal_gradient_phase_state}
\end{equation}
Applying the inverse quantum Fourier transform to each coordinate and measuring gives an estimate
$\widehat{\bm\gamma}$ satisfying the standard
phase-estimation bound $\Pr[\left|\widehat{\gamma}_j-\gamma_j\right|>{3}/{2^p}]\leq1/4$ for any $j\in [M]$.
Repeating this ideal experiment independently $\Upsilon_{\rm med}$ times and taking the coordinate-wise median gives
\begin{equation}
    \Pr\left[
        \left\|
            \widehat{\bm\gamma}^{(\mathrm{med})}
            -
            \bm\gamma
        \right\|_\infty
        >
        \frac{3}{2^p}
    \right]
    \leq
    \frac{\delta}{4},
    \label{eq:gradient_ideal_median_bound}
\end{equation}
by the Hoeffding bound and a union bound over the $M$ coordinates.

We next reproduce this ideal experiment approximately and catalytically. Apply Lemma~\ref{lem:catalytic_almost_linear_block_hamiltonian} with
\begin{equation}
    R=M,
    \qquad
    L_j=O_j,
    \qquad
    a=p,
    \qquad
    T
    =
    -\frac{2\pi}{3}
    2^p\sigma_{\rm grad}.
    \label{eq:gradient_phase_time}
\end{equation}
The lemma gives a subset
$F\subseteq\{0,\ldots,2^p-1\}^M$, containing at least a $1-\delta_{\rm obs}$ fraction of the grid points, such that
\begin{equation}
    \left|
        f(\bm B)
        -
        \frac{1}{\sigma_{\rm grad}}
        \sum_{j=1}^{M}
        \mu_jx_j(\bm B)
    \right|
    \leq
    \varepsilon_{\rm obs}
    \qquad
    (\bm B\in F).
    \label{eq:gradient_f_approximation}
\end{equation}
Moreover, for the filtered Hamiltonian
\begin{equation}
    A_f(\bm B)
    :=
    \ket{\psi_0}\!\bra{\psi_0}
    W_{\rm red,H}(\bm B)
    \ket{\psi_0}\!\bra{\psi_0}
    +
    \Pi_\Delta
    \mathcal E_H
    \left(
        W_{\rm red,H}(\bm B)
    \right)
    \Pi_\Delta,
\end{equation}
we have $A_f(\bm B)\ket{\psi_0}=f(\bm B)\ket{\psi_0}$.
Thus, the exact time-$T$ evolution acts on the ground-state sector as $\ket{\bm B}\ket{\psi_0}\mapsto e^{-iTf(\bm B)}\ket{\bm B}\ket{\psi_0}$.
By Eqs.~\eqref{eq:gradient_resolution},
\eqref{eq:gradient_phase_time}, and
\eqref{eq:gradient_f_approximation}, for every
$\bm B\in F$, $|-Tf(\bm B)-2\pi 2^p\bm x(\bm B)\cdot\bm\gamma|\leq |T|\varepsilon_{\rm obs}$.
Therefore, if $\ket{\Phi_f}$ denotes the phase state produced by the exact filtered-Hamiltonian evolution, then
\begin{equation}
    \left\|
        \ket{\Phi_f}
        -
        \ket{\Phi_{\bm\gamma}}
    \right\|^2
    \leq
    |T|^2\varepsilon_{\rm obs}^2
    +
    4\delta_{\rm obs}.
    \label{eq:gradient_exact_phase_state_error}
\end{equation}
Indeed, on $F$ we use
$\lvert e^{iu}-e^{iv}\rvert\leq |u-v|$, while each grid point outside $F$ contributes at most 4 to the squared Euclidean distance.
Let us take
\begin{equation}
    \varepsilon_{\rm obs}
    =
    \mathcal{O}\left(\frac{\eta}{|T|}\right),
    \qquad
    \varepsilon_{\rm HS}
    =
    \mathcal{O}\left(\eta^2\right),
    \qquad
    \varepsilon_{\rm OFT}
    =
    \mathcal{O}\left(\frac{\eta^2}{|T|}\right).
    \label{eq:gradient_implementation_precisions}
\end{equation}
The block-encoding error in
Lemma~\ref{lem:catalytic_almost_linear_block_hamiltonian} additionally contributes $\mathcal O (\sqrt{\varepsilon_{\rm HS}+|T|\varepsilon_{\rm OFT}})$ to the Euclidean-distance error of the output state. 
Together with Eqs.~\eqref{eq:gradient_internal_errors} and
\eqref{eq:gradient_exact_phase_state_error}, applying the unitary Eq.~\eqref{eq:catalytic_almost_linear_Vphase} in Lemma~\ref{lem:catalytic_almost_linear_block_hamiltonian} to $\frac{1}{\sqrt{2^{pM}}}\sum_{\bm{B}}\ket{\bm{B}} \ket{\psi_0} \ket{0}^{\otimes n_{\rm anc}}$,
one can prepare a quantum state that is $\eta$-close to $\ket{\Phi_{\bm\gamma}}\ket{\psi_0}\ket{0}^{\otimes n_{\rm anc}}$.
We sequentially repeat this preparation
$\Upsilon_{\rm med}$ times, using fresh grid registers but reusing the same ground-state register. All Fourier measurements may be deferred to the end. From telescoping expansion, the final state before the deferred Fourier measurements
is at most $\Upsilon_{\rm med}\eta=\xi$ away in Euclidean distance from
\begin{equation}
    \ket{\Phi_{\bm\gamma}}^{\otimes\Upsilon_{\rm med}}
    \ket{\psi_0}
    \ket{0}^{\otimes n_{\rm anc}}.
    \label{eq:gradient_global_ideal_state}
\end{equation}
Consequently, the actual joint distribution of all Fourier measurement outcomes is $\xi$-close in total variation distance to the ideal distribution.
Let $\widehat{\bm\gamma}^{(\mathrm{med},\rm actual)}$ denote the coordinate-wise median obtained from the actual measurement
outcomes, and define $\widehat{\bm g}:=3\widehat{\bm\gamma}^{(\mathrm{med},\rm actual)}$.
Equation~\eqref{eq:gradient_ideal_median_bound} and the total variation bound imply
\begin{equation}
    \Pr\left[
        \left\|
            \widehat{\bm g}
            -
            \bm\mu
        \right\|_\infty
        >
        \varepsilon
    \right]
    \leq
    \frac{\delta}{4}
    +
    \xi= \frac{\delta}{4}+\frac{\delta^2}{16}
    \label{eq:gradient_estimation_failure}
\end{equation}

It remains to prove the catalytic return.
In the ideal state
Eq.~\eqref{eq:gradient_global_ideal_state}, the ground-state register is exactly $\ket{\psi_0}$.
Applying the same classical--quantum block-diagonal argument as in the proof of Theorem~\ref{thm:main_theorem_for_multiple_observable_est} to the $\xi$-close actual final state
gives
\begin{equation}
    \sum_z
    P'(z)
    D_{\rm tr}\left(
        \rho_z,
        \ket{\psi_0}\!\bra{\psi_0}
    \right)
    \leq
    2\xi,
\end{equation}
where $z$ denotes the complete Fourier-measurement outcome, ${P}'$ is the actual probability distribution, and $\rho_z$ is the corresponding returned state in the ground-state register.
Hence, Markov's inequality yields
\begin{equation}
    \Pr_{z\sim P'}\left[
        D_{\rm tr}\left(
            \rho_z,
            \ket{\psi_0}\!\bra{\psi_0}
        \right)
        >
        \delta
    \right]
    \leq
    \frac{2\xi}{\delta}
    =
    \frac{\delta}{8}.
    \label{eq:gradient_return_failure}
\end{equation}
Combining
Eqs.~\eqref{eq:gradient_estimation_failure} and
\eqref{eq:gradient_return_failure} proves that the
estimation and return guarantees hold simultaneously with probability at least $1-\delta$.

Finally, we evaluate the required quantum resources. 
By Eq.~\eqref{eq:gradient_resolution}, $|T|=\mathcal O\left({\sigma_{\rm grad}}/{\varepsilon}\right)$.
For one phase-state preparation, Lemma~\ref{lem:catalytic_almost_linear_block_hamiltonian} uses
\begin{equation}
    \mathcal O\left[
        \frac{1}{\Delta}
        \log\left(
            \frac{1}{\varepsilon_{\rm OFT}}
        \right)
        \left(
            |T|
            +
            \log\left(
                \frac{1}{\varepsilon_{\rm HS}}
            \right)
        \right)
    \right]
\end{equation}
sequential controlled-Hamiltonian evolution time.
Multiplying by
$\Upsilon_{\rm med}
=\mathcal O(\log(M/\delta))$ and substituting
Eq.~\eqref{eq:gradient_implementation_precisions} gives
\begin{equation}
    \widetilde{\mathcal O}\left[
        \frac{1}{\Delta}
        \left(
            \frac{\sigma_{\rm grad}}{\varepsilon}
            +
            1
        \right)
    \right].
\end{equation}
Similarly, the total number of sequential block-encoding queries is
\begin{align}
    \widetilde{\mathcal O}\left[
        \frac{M}{\sigma_{\rm grad}}
        \left(
            |T|+1
        \right)
    \right]
    =
    \widetilde{\mathcal O}\left[
        \frac{M}{\sigma_{\rm grad}}
        \left(
            \frac{\sigma_{\rm grad}}{\varepsilon}
            +
            1
        \right)
    \right]
    =
    \widetilde{\mathcal O}\left(
        \frac{M}{\varepsilon}
        +
        \frac{M}{\sigma_{\rm grad}}
    \right).
\end{align}
\end{proof}

By replacing the simple Fourier measurement with a nearly unbiased phase-estimation routine, the joint multidimensional phase kickback (Lemma~\ref{lem:catalytic_almost_linear_block_hamiltonian}) would also produce a nearly unbiased estimator without changing the leading resource scaling (see e.g., Corollary~32 in Ref.~\cite{van2023quantum}).

\subsection{Ground-state tomography}

\begin{thm}
    [Parallel catalytic tomography]
    Given $P$ copies of the unique ground state $\ket{\psi_0}$ of a $\Delta$-gap $d$-dimensional Hamiltonian, there is a measurement protocol that with probability at least $1-\delta_{\rm tom}$, produces a classical description of $|\psi_0\rangle\langle \psi_0|$ with trace-distance error $\eta$ and returns a state $\delta_{\rm tom}$-close to $\ket{\psi_0}^{\otimes P}$ in trace distance.
    This protocol uses 
    \begin{itemize}
        \item $\mathcal{\tilde{O}}\left(\frac{1}{\Delta}\left(\frac{d}{ P\eta}+1\right)\right)$ sequential time for controlled Hamiltonian evolution of $H$
    \end{itemize}
    with $P$-way parallelism as illustrated in Fig.~\ref{fig:proposed_measurement_protocol} (where $M$ is given by $M=2d^2-d$).
    Furthermore, after sharing $\mathcal{\tilde{O}}(d^2)$ $P$-qubit GHZ states, this measurement protocol can be performed on separable $P$ quantum devices, each of which has one ground state copy.
    Here, $\tilde{\mathcal{O}}$ hides poly-logarithmic factors in $P,d,\eta^{-1},$ and $\delta_{\rm tom}^{-1}$.
\end{thm}

\begin{proof}
    We prove this theorem by combining Lemma~\ref{thm:main_theorem_for_multiple_observable_est_unbiased} and Lemma~\ref{lem:lemma_for_tomography}.
    Let $\rho:=|\psi_0\rangle\langle\psi_0|$, and let $\varepsilon_5$ and $\delta_5$ denote the error and failure parameters used in Lemma~\ref{thm:main_theorem_for_multiple_observable_est_unbiased}, respectively. 
    Define $M_{\rm tom}:=2d^2-d$.
    % , and consider a typical regime $\ln (2M_{\rm tom}/\delta_5)\lesssim d^2$, where $\delta_5$ is not extremely small.
    Applying Lemma~\ref{thm:main_theorem_for_multiple_observable_est_unbiased} to the following set of observables
    \begin{equation}\label{eq:set_observables_tomography}
        \{O_j\}_{j=1}^{M_{\rm tom}}:=\left\{|i\rangle\langle i|:i\in[d]\right\}\cup \left\{\frac{|i\rangle\langle j|+|j\rangle\langle i|}{2},\frac{|j\rangle\langle i|-|i\rangle\langle j|}{2i}:i,j\in[d],i\neq j\right\},
    \end{equation}
    we obtain a random vector $\hat{\bm{x}}$.
    Here, all observables in this set have the operator norm at most one.
    Reshaping $\hat{\bm{x}}$ appropriately, we have a random matrix $X$ that is $(\delta_5+\delta_5^2)/4$-close in total variation distance to a random matrix $Y$ with independent entries satisfying
\begin{equation}
    |Y_{ij}-\rho_{ij}|
    \leq
    \sqrt{2}\,\varepsilon_5~
    \text{(with probability one)},~~~
    \left|
        \mathbb{E}[Y_{ij}]-\rho_{ij}
    \right|
    \leq
    \sqrt{2}\pi\frac{\delta_5}{M_{\mathrm{tom}}},
\end{equation}
for every $i,j$.    
Therefore, Lemma~\ref{lem:lemma_for_tomography} can be used to obtain an estimator $\hat{\rho}$ from $X$, and it ensures
\begin{align}
    \|\widehat{\rho}-\rho\|_1
    &\leq 
    c_1
        \sqrt{2}\,\varepsilon_5
    \left[
    \sqrt{d}
    +
    \sqrt{
        \ln\left(
            \frac{2M_{\mathrm{tom}}}{\delta_5}
        \right)}\right]
    +
    c_0'd
    \frac{\delta_5}{M_{\mathrm{tom}}}
    \label{eq:tomography_error_eta}
\end{align}
with probability at least $1-\delta_5/M_{\rm tom}-\delta_5(\delta_5+1)/4$.
Moreover, from the union bound, we can ensure both Eq.~\eqref{eq:tomography_error_eta} and the ground-state $\ket{\psi_0}^{\otimes P}$ return within trace-distance $\delta_5$ with a failure probability at most
\begin{equation}
    \delta_5
    +
    \frac{\delta_5}{M_{\mathrm{tom}}}
    +
    \frac{\delta_5+\delta_5^2}{4}.
    \label{eq:tomography_prob_error}
\end{equation}

Hereafter, we carefully take $\varepsilon_5$ and $\delta_5$ to bound the trace-distance and the failure probability, and then, calculate the required sequential evolution time.
For a sufficiently small universal constant $3/5\geq c>0$, we take
\begin{equation}
    \delta_5
    =
    c\min\left\{
        \delta_{\mathrm{tom}},
        \eta d
    \right\}\leq \delta_{\rm tom},
    \qquad
    \varepsilon_5
    =
    \frac{c\eta}{
        \sqrt{d}
        +
        \sqrt{
            \ln\left(
                2M_{\mathrm{tom}}/\delta_5
            \right)
        }
    }.
    \label{eq:tomography_parameter_choice}
\end{equation}
We first verify the tomography error. 
% By the assumed regime
% $\ln(2M_{\mathrm{tom}}/\delta_5)\lesssim d^2$, we have
% $\sqrt{d}+\sqrt{\ln\left(2M_{\mathrm{tom}}/\delta_5\right)}\lesssim d$.
Substituting Eq.~\eqref{eq:tomography_parameter_choice} into
Eq.~\eqref{eq:tomography_error_eta}, and using
$M_{\mathrm{tom}}=2d^2-d\geq d^2$, gives
% \begin{align}
%     \|\widehat{\rho}-\rho\|_1
%     &\leq
%     \sqrt{2}c_1\varepsilon_5
%     \left[
%         \sqrt{d}
%         +
%         \sqrt{
%             \ln\left(
%                 \frac{2M_{\mathrm{tom}}}{\delta_5}
%             \right)
%         }
%     \right]
%     \nonumber\\
%     &\quad
%     +
%     \frac{\delta_5}{M_{\mathrm{tom}}}
%     \left\{
%         c_1
%         \left[
%             \sqrt{d}
%             +
%             \sqrt{
%                 \ln\left(
%                     \frac{2M_{\mathrm{tom}}}{\delta_5}
%                 \right)
%             }
%         \right]
%         +
%         c_0'd
%     \right\}
%     \lesssim
%     c\eta
%     +
%     \frac{\delta_5}{d}
%     \leq
%     2c\eta.
%     \label{eq:tomography_error_bound}
% \end{align}
\begin{align}
    \|\widehat{\rho}-\rho\|_1
    &\leq
    \sqrt{2}c_1\varepsilon_5
    \left[
        \sqrt{d}
        +
        \sqrt{
            \ln\left(
                \frac{2M_{\mathrm{tom}}}{\delta_5}
            \right)
        }
    \right]
    +c_0'd
    \frac{\delta_5}{M_{\mathrm{tom}}}
    \le
    \sqrt{2} c_1c\eta
    +
    c_0'\frac{\delta_5}{d}
    \lesssim
    2c\eta.
    \label{eq:tomography_error_bound}
\end{align}
Hence, by choosing $c$ sufficiently
small, Eq.~\eqref{eq:tomography_error_bound} implies $(1/2)\|\widehat{\rho}-\rho\|_1\leq\eta$.
We next bound the overall failure probability.  
Since $d\geq2$ implies $M_{\mathrm{tom}}\geq6$, Eq.~\eqref{eq:tomography_prob_error} gives
\begin{align}
    \delta_5
    +
    \frac{\delta_5}{M_{\mathrm{tom}}}
    +
    \frac{\delta_5+\delta_5^2}{4}
    \leq
    \delta_5
    +
    \frac{\delta_5}{6}
    +
    \frac{\delta_5}{2}
    =
    \frac{5}{3}\delta_5
    \leq
    \delta_{\mathrm{tom}}.
    \label{eq:tomography_failure_bound}
\end{align}
Since $\delta_5\leq\delta_{\mathrm{tom}}$, the returned state is
also $\delta_{\mathrm{tom}}$-close to
$|\psi_0\rangle^{\otimes P}$ in trace distance.  Combining
Eqs.~\eqref{eq:tomography_error_bound} and
\eqref{eq:tomography_failure_bound}, the tomography and the
ground-state return are simultaneously successful with probability
at least $1-\delta_{\mathrm{tom}}$.

It remains to calculate the sequential total evolution time.  
For the set of observables in Eq.~\eqref{eq:set_observables_tomography}, we can evaluate
\begin{equation}
    \|S_{M_{\mathrm{tom}}}\|
    =
    d,
    \qquad
    \operatorname{tr}[S_{M_{\mathrm{tom}}}]
    =
    d^2,
    \qquad
    \frac{
        \operatorname{tr}[S_{M_{\mathrm{tom}}}]
    }{
        \|S_{M_{\mathrm{tom}}}\|
    }
    =
    d.
    \label{eq:tomography_SM_parameters}
\end{equation}
Substituting Eq.~\eqref{eq:tomography_SM_parameters} into
Eq.~\eqref{eq:thmS4_sigma}, with the parameters
$(\varepsilon,\delta)=(\varepsilon_5,\delta_5)$, yields
\begin{equation}
    \sigma
    =
    \min\left\{
        M_{\mathrm{tom}},
        \mathcal{O}\left(
            \sqrt{d}
            \log^{1/2}\left[
                \frac{d}{\delta_5}
                \log\left(
                    \frac{M_{\mathrm{tom}}}{\delta_5}
                \right)
                \log\left(
                    \frac{1}{\varepsilon_5}
                \right)
            \right]
        \right)
    \right\}.
    \label{eq:tomography_sigma_initial}
\end{equation}
The choice of $\delta_5$ in
Eq.~\eqref{eq:tomography_parameter_choice} gives
$\frac{d}{\delta_5}=\frac{1}{c}
    \max\left\{
        \frac{d}{\delta_{\mathrm{tom}}},
        \frac{1}{\eta}
    \right\}$.
Furthermore, since
$M_{\mathrm{tom}}/d=2d-1$ and $\max\left\{
        \frac{d}{\delta_{\mathrm{tom}}},
        \frac{1}{\eta}
    \right\}
    \geq d$, we have
\begin{align}
    \ln\frac{M_{\mathrm{tom}}}{\delta_5}
    &=
    \ln\left(\frac{2d-1}{c}
    \max\left\{
        \frac{d}{\delta_{\mathrm{tom}}},
        \frac{1}{\eta}
    \right\}\right)
    =
    \mathcal{O}\left(\ln\left[
        \max\left\{
            \frac{d}{\delta_{\mathrm{tom}}},
            \frac{1}{\eta}
        \right\}\right]
    \right).
    \label{eq:tomography_M_delta_relation}
\end{align}
The choice of $\varepsilon_5$ also gives
\begin{align}
    \ln\frac{1}{\varepsilon_5}
    &=
    \ln\left(\frac{1}{c\eta}
    \left[
        \sqrt{d}
        +
        \sqrt{
            \ln\left(
                \frac{2M_{\mathrm{tom}}}{\delta_5}
            \right)
        }
    \right]\right)
    =
    \mathcal{O}\left(\ln\left[
        \max\left\{
            \frac{d}{\delta_{\mathrm{tom}}},
            \frac{1}{\eta}
        \right\}\right]
    \right),
    \label{eq:tomography_epsilon_relation}
\end{align}
where we used
$d\leq\max\{d/\delta_{\mathrm{tom}},1/\eta\}$ and 
$\eta^{-1}\leq\max\{d/\delta_{\mathrm{tom}},1/\eta\}$.
Combining them, we obtain
\begin{align}
    &\ln\left[
        \frac{d}{\delta_5}
        \ln\left(
            \frac{M_{\mathrm{tom}}}{\delta_5}
        \right)
        \ln\left(
            \frac{1}{\varepsilon_5}
        \right)
    \right]
    =
    \mathcal{O}\left(
        \ln\left[
            \max\left\{
                \frac{d}{\delta_{\mathrm{tom}}},
                \frac{1}{\eta}
            \right\}
        \right]
    \right).
    \label{eq:tomography_outer_log}
\end{align}
Equation~\eqref{eq:tomography_sigma_initial} therefore implies
\begin{equation}
    \sigma
    =
    \mathcal{O}\left(
        \sqrt{
            d
            \ln\left[
                \max\left\{
                    \frac{d}{\delta_{\mathrm{tom}}},
                    \frac{1}{\eta}
                \right\}
            \right]
        }
    \right)
    \label{eq:tomography_sigma_final}
\end{equation}

Lemma~\ref{thm:main_theorem_for_multiple_observable_est_unbiased} preserves the sequential evolution-time bound of
Theorem~\ref{thm:main_theorem_for_multiple_observable_est}.
Thus, using
Eqs.~\eqref{eq:tomography_parameter_choice} and
\eqref{eq:tomography_sigma_final}, the required sequential total
evolution time is
\begin{align}
    \tilde{\mathcal{O}}\left[
        \frac{1}{\Delta}
        \left(
            \frac{\sigma}{P\varepsilon_5}
            +
            1
        \right)
    \right]
    &=
    \tilde{\mathcal{O}}\left[
        \frac{1}{\Delta}
        \left(
            \frac{
                \sqrt{
                    d
                    \ln\left[
                        \max\left\{
                            d/\delta_{\mathrm{tom}},
                            1/\eta
                        \right\}
                    \right]
                }
            }{
                P\eta
            }
            \left[
                \sqrt{d}
                +
                \sqrt{
                    \ln\left(
                        2M_{\mathrm{tom}}/\delta_5
                    \right)
                }
            \right]
            +
            1
        \right)
    \right]
    \nonumber\\
    &=
    \tilde{\mathcal{O}}\left[
        \frac{1}{\Delta}
        \left(
            \frac{d}{P\eta}
            +
            1
        \right)
    \right].
\end{align}
\end{proof}

\begin{lem}
    [Pure-state estimator from element-wise estimates]\label{lem:lemma_for_tomography}
    For a $d$-dimensional pure state $\rho=|\psi\rangle\langle \psi|$, 
    let us consider a $d\times d$ random matrix $Y$ with independent entries satisfying the following: for any $i,j$, $|{Y}_{ij}-\rho_{ij}|\leq \varepsilon$ holds with probability one and $|\mathbb{E}[Y_{ij}]-\rho_{ij}|\leq c_0\delta$ holds for a universal constant $c_0$.
    Given one sample of a random matrix $X$ that is $\xi$-close to $Y$ in total variation distance, we can construct a density matrix estimator $\hat{\rho}$ that satisfies 
    \begin{equation}
        \left\|\hat{\rho}-\rho\right\|_1\leq c_1 \varepsilon\left(\sqrt{d}+\sqrt{\ln(2/\delta)}\right)+c_0'd\delta
    \end{equation}
    with probability at least $1-\delta-\xi$ for universal positive constants $c_0',c_1$.
\end{lem}

\begin{proof}
This lemma mostly follows from Proposition~2.4 in Ref.~\cite{rudelson2010non} and Corollary~44 in Ref.~\cite{van2023quantum}.
We first define a classical post-processing function $\phi$.
This function $\phi$ takes a $d\times d$ matrix $A$ and then outputs 
\begin{equation}
    \phi(A):=\underset{\sigma\geq0,~{\rm tr[\sigma]}=1,~{\rm rank}(\sigma)=1}{{\rm argmin}}~\|A_{\rm H}-\sigma\|,
\end{equation}
where $A_{\rm H}:=(A+A^\dagger)/2$.
The $\rho$ itself is the rank-one density matrix, so $\phi(A)$ must be in the ball centered at $A_{\rm H}$ with the radius $\|A_{\rm H}-\rho\|$ measured by the operator norm.
\begin{equation}
    \|\phi(A)-\rho\|\leq \|\phi(A)-A_{\rm H}\|+\|\rho-A_{\rm H}\|\leq 2\|\rho-A_{\rm H}\|.
\end{equation}

Since $Y$ has independent entries bounded by
\begin{equation}
    |Y_{ij}-\mathbb{E}[Y_{ij}]|
    \le
    |Y_{ij}-\rho_{ij}|+|\rho_{ij}-\mathbb{E}[Y_{ij}]|
    \le 2\varepsilon,
\end{equation}
from Proposition~2.4 in Ref.~\cite{rudelson2010non}, we have
\begin{equation}
    {\rm Pr}[\|Y-\mathbb{E}[Y]\|>\sqrt{\ln(2)}(2\varepsilon)(2C\sqrt{d}+\tau)]\le 2e^{-c\tau^2}
\end{equation}
for any $\tau\geq 0$ and universal constants $C,c>0$.
Therefore, by setting $\tau=\sqrt{\ln(2/\delta)/c}$, with probability at least $1-\delta$, 
\begin{equation}
    \|Y_{\rm H}-\rho\|\leq \|Y-\rho\|\le \|Y-\mathbb{E}[Y]
    \|+\|\mathbb{E}[Y]-\rho\|\le
    \sqrt{\ln(2)}(2\varepsilon)(2C\sqrt{d}+\tau)
    +c_0d\delta
\end{equation}
holds. In the final inequality, we used the fact that the operator norm of $d$-dimensional matrices is bounded by the product of $d$ and the maximum value of absolute entries of $A$.
We now use the trace norm bound $\|A\|_1\leq {\rm rank}(A)\|A\|$.
Then, we conclude that $\phi(Y)$ satisfies 
\begin{equation}
    \|\phi(Y)-\rho\|_1\leq 2\|\phi(Y)-\rho\|\leq 4\|\rho-Y_{\rm H}\|=4\sqrt{\ln(2)}(2\varepsilon)(2C\sqrt{d}+\tau)
    +4c_0d\delta
\end{equation}
with probability at least $1-\delta$.
We actually have the estimate $\phi(X)$, instead of $\phi(Y)$, but the failure probability would increase by at most $\xi$ due to the closeness assumption in total variation distance.
\end{proof}

\subsection{Algorithmic construction}
\label{sec:algorithmic_construction}

In this subsection, we construct the joint multidimensional phase kickback that underpins our catalytic readout protocols.
This construction jointly encodes the expectation values of possibly
noncommuting observables in a control-dependent ground-state eigenphase,
rather than through separate single-observable phase-kickback operations.

We first review joint multi-observable encoding in Lemma~\ref{lem:multi_obs_encoding_improved} and
energy-basis block-diagonalization via the operator Fourier transformation
in Lemma~\ref{lem:block_diagonalization}, both based on previous work.
The key step is to combine these ingredients 
by applying energy filtering and block-Hamiltonian simulation 
coherently conditioned on the same registers that control the joint observable encoding.
This yields the joint multidimensional phase kickback of Lemma~\ref{lem:catalytic_almost_linear_block_hamiltonian}.
Its general formulation also supports the catalytic gradient-estimation protocol with single copy in Theorem~\ref{lem:Catalytic_gradient_readout}.
We then specialize it to binary controls and analyze the $P$-parallel
action on GHZ states, including the coherent error from exceptional
control strings, to prove the catalytic joint phase shifter
guarantee in Lemma~\ref{lem:catalytic_joint_phase_shifter}.

\begin{lem}[Joint multi-observable encoding]\label{lem:multi_obs_encoding_improved}
    Let $L_1,...,L_R$ be $n$-qubit observables with the operator norm $\|L_j\|\leq 1$.
    For a given $\delta_{\rm obs}>0$ and any integer $a\geq 1$, suppose that we can take $\sigma$ such that 
    % \begin{equation}
    %     \sqrt{\frac{1-2^{-2a}}{3/2}\|{V}\|\ln\left(\frac{8}{\delta_{\rm obs}}\frac{{\rm tr}[V]}{\|V\|}\right)} +\frac{1-2^{-a}}{3/2}\ln\left(\frac{8}{\delta_{\rm obs}}\frac{{\rm tr}[V]}{\|V\|}\right)\leq \sigma\leq R,
    % \end{equation}
    % \begin{equation}
    %     \sigma\le R,~~ \frac{\sigma^2}{(1-2^{-a})\sigma+(1-4^{-a})\|S\|}\geq \frac{2}{3}\ln\left(\frac{8}{\delta_{\rm obs}}\frac{{\rm tr}[S]}{\|S\|}\right)
    % \end{equation}
    \begin{equation}\label{eq:condition_joint_multiencoding}
        {0<\sigma\le R,~~ \sigma^2\geq \frac{2}{3}(1-4^{-a})\|S\|\ln\left(\frac{2}{\delta_{\rm obs}}\frac{{\rm tr}[S]}{\|S\|}\right)}
    \end{equation}
    hold for $S:= \sum_{j=1}^R L_j^2$.
    Then, we can implement a unitary gate 
    \begin{equation}
        \sum_{\boldsymbol{B}=(B_1,...,B_R)\in \{0,1,...,2^{a}-1\}^R} |\boldsymbol{B}\rangle \langle \boldsymbol{B}|\otimes W(\boldsymbol{B}).
    \end{equation}
    The unitary $W(\boldsymbol{B})$ is a perfect block-encoding of an $n$-qubit Hermitian $W_{\rm red,H}(\bm{B})$ satisfying the following.
    There exists a subset $F\subseteq \{0,1,...,2^{a}-1\}^R $ of the size $|F|\geq (1-\delta_{\rm obs})2^{aR}$ such that 
        for any $\boldsymbol{B}\in F$,  
        $W_{\rm red,H}(\boldsymbol{B})$ is $\varepsilon_{\rm obs}$-close to the following $n$-qubit Hermitian in the operator norm
        \begin{equation}
        \frac{1}{\sigma}\sum_{j=1}^R L_j\left(\frac{B_j}{2^a}-\frac{1}{2}+\frac{1}{2^{a+1}}\right).
    \end{equation}
    This implementation uses 
    $\mathcal{{O}}( R\sigma^{-1}\log(R\sigma^{-1}/\epsilon_{\rm obs}))$ queries to a controlled $(1,n_{\rm obs},0)$-block-encoding of $L_j$ ($j=1,2,...,R$), and $aR+n+n_{\rm obs}+\mathcal{O}(\log R)$ qubits.
\end{lem}

\begin{proof}
    This lemma mostly follows from Lemma~13 in Ref.~\cite{PRXQuantum.6.020308}.
    Let ${\bm{B}}=(B_1,...,B_R)\in \{0,1,...,2^a-1\}^R$.
    According to (Fig.~12 of) Ref.~\cite{PRXQuantum.6.020308}, we can implement a $|{\bm{B}}\rangle$-controlled $(1,n_{\rm obs}+\mathcal{O}(\log R),\epsilon)$-block-encoding of the following $n$-qubit Hermitian
    \begin{equation}
        \frac{2}{\pi 2^{\lceil \log_2 R\rceil}}\sum_{j=1}^R L_j\left(\frac{B_j}{2^a}-\frac{1}{2}+\frac{1}{2^{a+1}}\right).
    \end{equation}
    This implementation uses
    one application of each block-encoding for $L_j$ ($j=1,\ldots,R$).
    The resulting Hermitian block-encoding can be further combined with the uniform singular value amplification~\cite{low2017hamiltonian,10.1145/3313276.3316366}. This yields a unitary $W(\bm{B})$ that gives a $(4m\sqrt{\epsilon}+\epsilon')$-precise block-encoding of
    \begin{equation}\label{eq:amplified_BE_proof}
        \frac{1}{\sigma}\sum_{j=1}^R L_j\left(\frac{B_j}{2^a}-\frac{1}{2}+\frac{1}{2^{a+1}}\right),
    \end{equation}
    where $m=\mathcal{O}( R\sigma^{-1}\log(R\sigma^{-1}/\epsilon'))$, if the operator norm of Eq.~\eqref{eq:amplified_BE_proof} is smaller than 1/2.
    We define $F$ as a set of $\bm{B}$ such that it makes the operator 
norm of Eq.~\eqref{eq:amplified_BE_proof} is smaller than 1/2.
    
    % We here evaluate the size of the set $F$; the following evaluation is indeed tighter than that in Ref.~\cite{PRXQuantum.6.020308}.
    % Let $\hat{B}_j$ be a random variable uniformly distributed on $\{0,1,...,2^{a}-1\}$.
    % Then, $\hat{X}_j:= \hat{B}_j/2^a-1/2+1/2^{a+1}$ is symmetrically distributed in $[-1/2,1/2]$.
    % Therefore, from the Matrix Bernstein's inequality with intrinsic dimension~\cite{tropp2015introduction}, we have
    % \begin{equation}
    %     {\rm Pr}\left[\left\|\sum_{j=1}^R L_jX_j\right\|\geq \frac{\sigma}{2} \right]\leq 8\frac{{\rm tr}[S]}{\|S\|}\exp\left({-\frac{\sigma^2/8}{\mathbb{E}[X_1^2]\|S\|+(1/2-1/2^{a+1})\sigma/6}}\right)\leq \delta_{\rm obs}.
    % \end{equation}
    % The last inequality holds due to the assumption of $\sigma$.
    % This indicates that ${\rm Pr}[\hat{\bm{B}}\in F]\geq 1-\delta_{\rm obs}$ and thus $|F|\geq (1-\delta_{\rm obs})2^{aR}$ as the probability distribution of $(\hat{B}_1,...,\hat{B}_R)$ is uniform.

    We here evaluate the size of the set $F$; the following evaluation is indeed tighter than that in Ref.~\cite{PRXQuantum.6.020308}.
    Let $\hat{B}_j$ be a random variable uniformly distributed on $\{0,1,...,2^{a}-1\}$, and let
    $\hat{X}_j:= \hat{B}_j/2^a-1/2+1/2^{a+1}$.
    Write the binary expansion of $\widehat B_j$ as
\begin{equation}
    \widehat B_j
    =
    \sum_{\ell=1}^{a}
    2^{a-\ell}\widehat b_{j,\ell},
    \qquad
    \widehat b_{j,\ell}\in\{0,1\}.
\end{equation}
Since $\widehat B_j$ is uniform over all $a$-bit strings, the random variables $\varepsilon_{j,\ell}:=2\widehat b_{j,\ell}-1$ are all independent Rademacher random variables.
A direct calculation gives $\widehat X_j=\sum_{\ell=1}^{a}2^{-\ell-1}\varepsilon_{j,\ell}$.
Therefore,
\begin{equation}
    \sum_{j=1}^{R}L_j\widehat X_j
    =
    \sum_{j=1}^{R}
    \sum_{\ell=1}^{a}
    \varepsilon_{j,\ell}
    \left(2^{-\ell-1}L_j\right)
\end{equation}
is a matrix Rademacher series.
    From the intrinsic matrix-Rademacher bound in Lemma~\ref{lem_Intrinsic_matrix-Rademacher_bound}, we have
    \begin{equation}
        {\rm Pr}\left[\left\|\sum_{j=1}^R L_jX_j\right\|\geq \frac{\sigma}{2} \right]\leq 2\frac{{\rm tr}[S]}{\|S\|}\exp\left({-\frac{\sigma^2}{(2/3)(1-4^{-a})\|S\|}}\right)\leq \delta_{\rm obs}.
    \end{equation}
    The last inequality holds due to the assumption of $\sigma$.
    This indicates that ${\rm Pr}[\hat{\bm{B}}\in F]\geq 1-\delta_{\rm obs}$ and thus $|F|\geq (1-\delta_{\rm obs})2^{aR}$ as the probability distribution of $(\hat{B}_1,...,\hat{B}_R)$ is uniform.
\end{proof}

% Now, we introduce a block-encoding of the filtered operator $\hat{A}_f$ of an operator $A$ with a Gaussian function.
% We recall that for a $\Delta$-gaped Hamiltonian $H$, $\hat{A}_f$ is defined by
% \begin{align}
%     \hat{A}_f&:=\frac{1}{\sqrt{2\pi}}\int_{-\infty}^{\infty} {\rm d}t~f(t)e^{iHt}Ae^{-iHt}.
% \end{align}
% We consider the Gaussian function $f(t)$ with a width $1/v>0$
% \begin{equation}
%     f(t)=v e^{-v^2 t^2/2},~~~\frac{1}{\sqrt{2\pi}}\int_{-\infty}^{\infty} {\rm d}t~f(t)=1.
% \end{equation}
% The Fourier transformation $\hat{f}(w)$ of $f(t)$ is given by $\hat{f}(w)=e^{-w^2/(2v^2)}$.
% For any operator $A$ of the operator norm $\|A\|\leq 1$, the filtered operator $\hat{A}_f$ becomes almost block-diagonal in the energy basis as follows
% \begin{equation}
%     \left\|\hat{A}_f-|\psi_0\rangle\langle \psi_0|A|\psi_0\rangle\langle \psi_0|-A^\perp\right\|\leq e^{-\Delta^2/(2\sigma^2)},
% \end{equation}
% for some operator $A^\perp$ such that $A^\perp|\psi_0\rangle =\langle \psi_0|A^\perp =0$.

\begin{lem}
    [Block-diagonalization in the energy basis~\cite{chen2025catalytic}]\label{lem:block_diagonalization}
    Let $H$ be a Hamiltonian that has an unique ground state $|\psi_0\rangle$ with a spectral gap $\Delta$.
    For any Hermitian operator $A$ of the operator norm $\|A\|\leq 1$, we can implement a unitary gate for an $\varepsilon_{\rm OFT}$-precise block-encoding of the following block-diagonal operator in the energy basis
    \begin{equation}\label{eq:be_diagonal_a}
        |\psi_0\rangle\langle \psi_0| A|\psi_0\rangle\langle \psi_0|+\Pi_{\Delta}\mathcal{E}_H(A)\Pi_{\Delta},~~~\Pi_\Delta:=\bm{1}-|\psi_0\rangle\langle \psi_0|,
    \end{equation}
    where $\mathcal{E}_H$ is a mixed unitary channel depending on $H$.
    This implementation uses one query to a $(1,n_A,0)$-block-encoding of $A$, $\mathcal{O}(\Delta^{-1}{\log(1/\varepsilon_{\rm OFT})})$ controlled Hamiltonian evolution time for $H$, and $n+n_A+\mathcal{O}(\log(\|H\|/(\varepsilon_{\rm OFT}\Delta)))$ qubits.
    % Here, $\mathcal{\tilde{O}}(1)$ in the qubit complexity hides logarithmic factors in $\varepsilon_{\rm OFT}^{-1}$ and $\|H\|\Delta^{-1}$.
\end{lem}
\begin{proof}
    This lemma follows from Lemma V.1 and Lemma A.1 in Ref.~\cite{chen2025catalytic}.
    Let us write the spectral decomposition of $H$ as $H=\sum_{i\ge 0} E_i|\psi_i\rangle\langle \psi_i|$.
    According to Lemma V.1, we focus on the filtered operator $\hat{A}_f$ 
    \begin{align}
    \hat{A}_f:=\frac{1}{\sqrt{2\pi}}\int_{-\infty}^{\infty} {\rm d}t~f(t)e^{iHt}Ae^{-iHt}=\sum_{i,j}|{\psi_i}\rangle\langle{\psi_i}|A|{\psi_j}\rangle\langle{\psi_j}|\hat{f}(E_j-E_i),
    \end{align}
    with the Gaussian filter with a width $1/v>0$
    \begin{equation}
    f(t)=v e^{-v^2 t^2/2},~~~\frac{1}{\sqrt{2\pi}}\int_{-\infty}^{\infty} {\rm d}t~f(t)=1
    \end{equation}
    and its Fourier transformation $\hat{f}(w)=e^{-w^2/(2v^2)}$.
    Under this choice, the off-diagonal terms of $\hat{A}_f$ are well suppressed by $v^2$
    \begin{align}
        \left\|\Pi_\Delta\hat{A}_f |\psi_0\rangle \langle \psi_0|\right\|&=\left\|\sum_{i\neq 0}\hat{f}(E_0-E_i)\ket{\psi_i}\bra{\psi_i}A\ket{\psi_0}\bra{\psi_0}\right\|\notag\\
        &\leq \left\|\sum_{i\neq 0}e^{-(E_0-E_i)^2/(2v^2)}\ket{\psi_i}\bra{\psi_i}\right\|\leq e^{-\Delta^2/(2v^2)}.
    \end{align}
    This yields
    \begin{align}
        \left\|\hat{A}_f-|\psi_0\rangle \langle \psi_0|A|\psi_0\rangle \langle \psi_0|-\Pi_\Delta \hat{A}_f \Pi_\Delta\right\|&=\left\|\Pi_\Delta\hat{A}_f |\psi_0\rangle \langle \psi_0|+|\psi_0\rangle \langle \psi_0|\hat{A}_f \Pi_\Delta\right\|\notag\\
        &\leq 2\left\|\Pi_\Delta\hat{A}_f |\psi_0\rangle \langle \psi_0|\right\|\leq 2e^{-\Delta^2/(2v^2)}.
    \end{align}

    From Lemma~A.1 in Ref.~\cite{chen2025catalytic}, we can implement a unitary gate for an $\varepsilon_{\rm OFT}/2$-precise block-encoding of $\hat{A}_f$ using a single query to a block-encoding of $A$ and $\mathcal{O}(v^{-1}\sqrt{\log(2/\varepsilon_{\rm OFT})})$ controlled Hamiltonian evolution time for $H$.
    Therefore, by taking $v=\Delta/\sqrt{2\ln(4/\varepsilon_{\rm OFT})}$, this unitary gate results in an $\varepsilon_{\rm OFT}$-precise block-encoding of the operator \eqref{eq:be_diagonal_a}.
\end{proof}

\begin{lem}[Joint multidimensional phase kickback]
\label{lem:catalytic_almost_linear_block_hamiltonian}
Let $L_1,\ldots,L_R$ be $n$-qubit observables with
$\|L_j\|\leq 1$, and define $S
    :=
    \sum_{j=1}^{R}L_j^2$.
Let \(a\geq 1\) be an integer, and let
\(\delta_{\rm obs},\varepsilon_{\rm obs}>0\).
Suppose that \(0<\sigma\leq R\) and that either
\(\sigma=R\), or
\begin{equation}
    \sigma^2
    \geq
    \frac{2}{3}
    \left(
        1-4^{-a}
    \right)
    \|S\|
    \ln\left(
        \frac{2}{\delta_{\rm obs}}
        \frac{\operatorname{tr}[S]}{\|S\|}
    \right).
    \label{eq:catalytic_almost_linear_sigma_condition}
\end{equation}
We assume controlled access to the time evolution of a Hamiltonian
$H$ having the unique ground state $\ket{\psi_0}$ and spectral
gap $\Delta$.
Then, for any $T\in \mathbb{R}$ and $\varepsilon_{\rm OFT},\varepsilon_{\rm HS}>0$, one can implement a unitary gate
\begin{equation}
    \sum_{\bm B\in\{0,1,\ldots,2^a-1\}^{R}}
    \ket{\bm B}\!\bra{\bm B}
    \otimes
    V(\bm B),
    \label{eq:catalytic_almost_linear_Vphase}
\end{equation}
where the \((n_{\rm anc}+n)\)-qubit unitary \(V(\bm B)\) is an $(\varepsilon_{\rm HS}
        +
        |T|\varepsilon_{\rm OFT})$-precise block encoding of the time-\(T\) evolution operator of
the following Hamiltonian
\begin{equation}
    \ket{\psi_0}\!\bra{\psi_0}
    W_{\rm red,H}(\bm B)
    \ket{\psi_0}\!\bra{\psi_0}
    +
    \Pi_\Delta
    \mathcal E_H
    \left(
        W_{\rm red,H}(\bm B)
    \right)
    \Pi_\Delta,
    \qquad
    \Pi_\Delta
    :=
    \bm 1
    -
    \ket{\psi_0}\!\bra{\psi_0}.
    \label{eq:catalytic_almost_linear_filtered_operator}
\end{equation}
Here, \(\mathcal E_H\) is the mixed-unitary channel in
Lemma~\ref{lem:block_diagonalization}, and $W_{\rm red,H}(\bm{B})$ is the Hermitian operator encoded by $W(\bm{B})$ in Lemma~\ref{lem:multi_obs_encoding_improved} such that
\begin{equation}
    \left|
        f(\bm{B})
        -
        \frac{1}{\sigma}
        \sum_{j=1}^{R}
        \bra{\psi_0}L_j\ket{\psi_0}
        \left(
            \frac{B_j}{2^a}
            -
            \frac{1}{2}
            +
            \frac{1}{2^{a+1}}
        \right)
    \right|
    \leq
    \varepsilon_{\rm obs},~~~f(\bm{B}):=
    \bra{\psi_0}
    W_{\rm red,H}(\bm B)
    \ket{\psi_0}
    \label{eq:catalytic_almost_linear_scalar_approximation}
\end{equation}
holds for every element $\bm{B}$ in a certain subset $F\subseteq \{0,1,...,2^a-1\}^R$ with $|F|\geq (1-\delta_{\rm obs})2^{aR}$.
The implementation of
Eq.~\eqref{eq:catalytic_almost_linear_Vphase} requires
\begin{itemize}
    \item
    $
        \mathcal O\left[
            \frac{1}{\Delta}
            \log\left(
                \frac{1}{\varepsilon_{\rm OFT}}
            \right)
            \left(
                |T|
                +
                \log\left(
                    \frac{1}{\varepsilon_{\rm HS}}
                \right)
            \right)
        \right]
    $
    controlled Hamiltonian-evolution time for \(H\),

    \item
    $
        \mathcal O\left[
            \frac{R}{\sigma}
            \log\left(
                \frac{R}{
                    \sigma\varepsilon_{\rm obs}
                }
            \right)
            \left(
                |T|
                +
                \log\left(
                    \frac{1}{\varepsilon_{\rm HS}}
                \right)
            \right)
        \right]$
    queries to a controlled $(1,n_{\rm obs},0)$-block-encoding of each $L_j$,

    \item
    $aR+n+n_{\rm anc}$ qubits, where
    $
        n_{\rm anc}
        =
        n_{\rm obs}
        +
        \mathcal O(\log R)
        +
        \mathcal O\left(
            \log\left(
                \frac{\|H\|}{
                    \Delta\varepsilon_{\rm OFT}
                }
            \right)
        \right).
    $
\end{itemize}
\end{lem}

\begin{proof}
In Lemma~\ref{lem:multi_obs_encoding_improved}, we first note that if there is no $\sigma$ satisfying Eq.~\eqref{eq:condition_joint_multiencoding}, we can still take $\sigma=R$.
This is because for every $\bm B\in\{0,1,\ldots,2^a-1\}^{R}$,
\begin{align}
    \left\|
        \frac{1}{R}
        \sum_{j=1}^{R}
        L_j
        \left(
            \frac{B_j}{2^a}
            -
            \frac{1}{2}
            +
            \frac{1}{2^{a+1}}
        \right)
    \right\|
    &\leq
    \frac{1}{R}
    \sum_{j=1}^{R}
    \|L_j\|
    \left|
        \frac{B_j}{2^a}
        -
        \frac{1}{2}
        +
        \frac{1}{2^{a+1}}
    \right|
    <
    \frac{1}{2},
\end{align}
and hence, the uniform singular-value amplification in the proof of
Lemma~\ref{lem:multi_obs_encoding_improved} is valid for every $\bm B$, and one may take $F=\{0,1,\ldots,2^a-1\}^{R}$ in this case.
Thus, Lemma~\ref{lem:multi_obs_encoding_improved} gives a controlled unitary
$\sum_{\bm B}\ket{\bm B}\!\bra{\bm B}\otimes W(\bm B)$, where $W(\bm B)$ is a perfect block encoding of an $n$-qubit Hermitian operator $W_{\rm red,H}(\bm B)$, and there exists a subset $F$ of size
$|F|\geq(1-\delta_{\rm obs})2^{aR}$
such that
\begin{equation}
    \left\|
        W_{\rm red,H}(\bm B)
        -
        \frac{1}{\sigma}
        \sum_{j=1}^{R}
        L_j
        \left(
            \frac{B_j}{2^a}
            -
            \frac{1}{2}
            +
            \frac{1}{2^{a+1}}
        \right)
    \right\|
    \leq
    \varepsilon_{\rm obs}
    \qquad
    (\forall\bm B\in F).
    \label{eq:catalytic_almost_linear_operator_bound_proof}
\end{equation}
Taking the expectation value of
Eq.~\eqref{eq:catalytic_almost_linear_operator_bound_proof}
in $\ket{\psi_0}$ immediately gives
Eq.~\eqref{eq:catalytic_almost_linear_scalar_approximation}.
We next apply Lemma~\ref{lem:block_diagonalization} coherently conditioned on $\bm B$,
with $A=W_{\rm red,H}(\bm B)$. This gives an $\varepsilon_{\rm OFT}$-precise block encoding
of
\begin{equation}
    \ket{\psi_0}\!\bra{\psi_0}
    W_{\rm red,H}(\bm B)
    \ket{\psi_0}\!\bra{\psi_0}
    +
    \Pi_\Delta
    \mathcal E_H
    \left(
        W_{\rm red,H}(\bm B)
    \right)
    \Pi_\Delta.
    \label{eq:catalytic_almost_linear_Af}
\end{equation}
Applying block-Hamiltonian simulation (Refs.~\cite{10.1145/3313276.3316366,Low2019hamiltonian}) to this block encoding
yields the unitary in
Eq.~\eqref{eq:catalytic_almost_linear_Vphase}.

This block-Hamiltonian simulation uses
$\mathcal{\mathcal{O}}(|T|+\log(1/{\varepsilon_{\rm HS}}))
$
queries to the block encoding of Eq.~\eqref{eq:catalytic_almost_linear_Af} and its inverse.
By Lemma~\ref{lem:block_diagonalization}, each such query requires
$\mathcal{O}\left(\frac{1}{\Delta}\log\frac{1}{\varepsilon_{\rm OFT}}\right)$ controlled Hamiltonian-evolution time for $H$.
This gives the first claimed resource bound.
The block-encoding queries can be similarly evaluated.
The $\bm B$ register uses $aR$ qubits, Lemma~\ref{lem:multi_obs_encoding_improved} requires
$n_{\rm obs}+\mathcal{O}(\log R)$ ancilla qubits, and
Lemma~\ref{lem:block_diagonalization} requires an additional
$\mathcal O\left(\log\left(\frac{\|H\|}{\Delta\varepsilon_{\rm OFT}}\right)\right)$
qubits.
The constant number of ancilla qubits used by block-Hamiltonian simulation is absorbed into the $\mathcal O(\cdot)$ terms.
This proves the stated qubit complexity.
\end{proof}

% \begin{lem}[Catalytic joint phase shifter]
%     Let $L_1,...,L_R$ be $n$-qubit observables with the norm at most one, and let $\delta_{\rm shift}$ be an error parameter.
%     We assume access to a time evolution operator of a gapped Hamiltonian $H$ with a unique ground state $|\psi_0\rangle$.
%     Suppose that 
%     \begin{equation}
%         \sigma\le R,~~ \frac{\sigma^2}{(1/2)\sigma+(3/4)\|S\|}\geq \frac{2}{3}\ln\left(\frac{64}{\delta_{\rm shift}^2}\frac{{\rm tr}[S]}{\|S\|}\right)
%     \end{equation}
%     hold for $S:= \sum_{j=1}^R L_j^2$.
%     Then, we can implement an $(R+n+n_{\rm anc})$-qubit unitary $V_{\rm phase}$ such that its $P$-parallel application $V_{\rm phase}^{\otimes P}$ maps $R$ GHZ states
%     \begin{equation}
%         |{\rm GHZ}_P\rangle^{\otimes R}\otimes (|\psi_0\rangle\otimes |0^{n_{\rm anc}}\rangle)^{\otimes P}
%     \end{equation}
%     to a quantum state that is $\delta_{\rm shift}$-close to 
%     \begin{equation}
%     e^{-iRP\tau/2}\left(\bigotimes_{j=1}^R |{\rm GHZ}_P(\tau \langle \psi_0| L_j|\psi_0\rangle ) \rangle \right)\otimes (|\psi_0\rangle\otimes |0^{n_{\rm anc}}\rangle)^{\otimes P}
%     \end{equation}
%     in the Euclidean distance.
%     The implementation of one $V_{\rm phase}$ requires $\mathcal{\tilde{O}}(\Delta^{-1}{\sigma |\tau|})$ controlled Hamiltonian evolution time for $H$,
%     $\mathcal{\tilde{O}}(\Delta^{-1}{R|\tau|})$ queries to a controlled $(1,n_{\rm obs},0)$-block-encoding of $L_j$ for $j=1,2,...,R$, and 
%     $n_{\rm anc}=n_{\rm obs}+\mathcal{\tilde{O}}(1)$ ancilla qubits.
% \end{lem}

\begin{proof}[Proof of Lemma~\ref{lem:catalytic_joint_phase_shifter}]
    % From Lemma~\ref{lem:multi_obs_encoding_improved}, by choosing $a=1$, we can implement
    % \begin{equation}
    %     \sum_{\bm{b}=(b_1,...,b_R)\in \{0,1\}^R}|\bm{b}\rangle \langle \bm{b}|\otimes {W}(\bm{b}).
    % \end{equation}
    % Here, ${W}(\bm{b})$ is an (exact) block-encoding of the $n$-qubit Hermitian $W_{\rm red,H}(\bm{b})$.
    % By further using Lemma~\ref{lem:block_diagonalization}, we can implement $\sum_{\bm{b}}|\bm{b}\rangle \langle \bm{b}|\otimes \widetilde{W}(\bm{b})$
    % such that $\widetilde{W}(\bm{b})$ is an $\varepsilon_{\rm OFT}$-precise block-encoding of the following block-diagonal operator in the energy basis
    % \begin{equation}\label{eq:block_hamiltonian}
    %     |\psi_0\rangle\langle \psi_0| W_{\rm red,H}(\bm{b})|\psi_0\rangle\langle \psi_0|+\Pi_{\Delta}\mathcal{E}_H(W_{\rm red,H}(\bm{b}))\Pi_{\Delta}.
    % \end{equation}
    % Thus, by using the block-Hamiltonian simulation technique (Refs.~\cite{10.1145/3313276.3316366,Low2019hamiltonian}), we can implement 
    % \begin{equation}
    %     V_{\rm phase}:=\sum_{\bm{b}=(b_1,...,b_R)\in \{0,1\}^R}|\bm{b}\rangle \langle \bm{b}|\otimes V(\bm{b}), 
    % \end{equation}
    % where the $(n_{\rm anc}+n)$-qubit unitary $V(\bm{b})$ is an $(\varepsilon_{\rm HS}+|T|\varepsilon_{\rm OFT})$-precise block-encoding of the time-$T$ evolution operator of the Hamiltonian~\eqref{eq:block_hamiltonian}.
    % The evolution time $T$ is given by $T=-2\sigma \tau$.
    % This Hamiltonian simulation requires $\mathcal{O}(|T|+\log(1/\varepsilon_{\rm HS}))$ queries to the block-encoding of Eq.~\eqref{eq:block_hamiltonian}.

Apply Lemma~\ref{lem:catalytic_almost_linear_block_hamiltonian} with $a=1$ and $T=-2\sigma \tau$,
and denote the resulting unitary by
\begin{equation}
    V_{\rm phase}
    :=
    \sum_{\bm b\in\{0,1\}^{R}}
    \ket{\bm b}\!\bra{\bm b}
    \otimes
    V(\bm b).
    \label{eq:Vphase_from_almost_linear_lemma}
\end{equation}
Hereafter, we show that $V_{\rm phase}^{\otimes P}$ put a desired phase shift on the GHZ states.
To this end, we first note that
\begin{align}
|{\rm GHZ}_P\rangle^{\otimes R}&=\frac{1}{2^{R/2}}\sum_{(b_1,...,b_R)\in \{0,1\}^R}|b_1\rangle^{\otimes P}|b_2\rangle^{\otimes P}\cdots 
|b_R\rangle^{\otimes P} \notag\\
&= \frac{1}{2^{R/2}}\sum_{(b_1,...,b_R)\in \{0,1\}^R} |b_1b_2...b_R\rangle^{\otimes P}= \frac{1}{2^{R/2}}\sum_{\boldsymbol{b}=(b_1,...,b_R)\in \{0,1\}^R} |\boldsymbol{b}\rangle^{\otimes P},
\end{align}
and 
\begin{align}
&\bigotimes_{j=1}^R|{\rm GHZ}_P(x_j)\rangle\notag\\
&= 
\frac{1}{2^{R/2}}
\left(\sum_{b_1\in \{0,1\}} e^{iPb_1x_1}|b_1\rangle^{\otimes P}\right)\otimes 
\left(\sum_{b_2\in \{0,1\}} e^{iPb_2x_2}|b_2\rangle^{\otimes P}\right)\otimes \cdots \otimes  
\left(\sum_{b_R\in \{0,1\}} e^{iPb_Rx_R}|b_R\rangle^{\otimes P}\right)\notag\\
&= 
\frac{1}{2^{R/2}}
\sum_{(b_1,b_2,...,b_R)\in \{0,1\}^R} e^{iP\sum_{j=1}^R b_jx_j}|b_1\rangle^{\otimes P}
|b_2\rangle^{\otimes P}
\cdots 
|b_R\rangle^{\otimes P}
\notag\\
&= 
\frac{1}{2^{R/2}}
\sum_{(b_1,b_2,...,b_R)\in \{0,1\}^R} e^{iP\sum_{j=1}^R b_jx_j}|b_1b_2...b_R\rangle^{\otimes P}= 
\frac{1}{2^{R/2}}
\sum_{\boldsymbol{b}=(b_1,b_2,...,b_R)\in \{0,1\}^R} e^{iP\sum_{j=1}^R b_jx_j}|\boldsymbol{b}\rangle^{\otimes P}.
\end{align}
The action of $V_{\rm phase}^{\otimes P}$ is calculated as
\begin{align}\label{eq:v_phase_ghz_initial}
&V_{\rm phase}^{\otimes P}\left(|{\rm GHZ}_P\rangle^{\otimes R}\otimes (|\psi_0\rangle\otimes |0^{n_{\rm anc}}\rangle)^{\otimes P}\right)\notag\\
&= \frac{1}{2^{R/2}}\sum_{\boldsymbol{b}=(b_1,...,b_R)\in \{0,1\}^R}  (V_{\rm phase}|\boldsymbol{b}\rangle\otimes |\psi_0\rangle\otimes |0^{n_{\rm anc}}\rangle)^{\otimes P}\notag\\
&= \frac{1}{2^{R/2}}\sum_{\boldsymbol{b}=(b_1,...,b_R)\in \{0,1\}^R}  (|\boldsymbol{b}\rangle\otimes V(\bm{b})(|\psi_0\rangle\otimes |0^{n_{\rm anc}}\rangle))^{\otimes P}.
\end{align}
The Euclidean distance between this quantum state and the following target state
\begin{align}\label{eq:v_phase_ghz_end}
    &e^{-i\Phi P\tau}\left(\bigotimes_{j=1}^R |{\rm GHZ}_P(\tau \langle L_j\rangle ) \rangle \right)\otimes (|\psi_0\rangle\otimes |0^{n_{\rm anc}}\rangle)^{\otimes P}\notag\\
    &=\frac{e^{-i\Phi P\tau}}{2^{R/2}}\left(\sum_{\bm{b}\in \{0,1\}^R}e^{iP\tau \sum_{j=1}^R b_j\langle L_j\rangle}|\bm{b}\rangle^{\otimes P}\right)\otimes (|\psi_0\rangle\otimes |0^{n_{\rm anc}}\rangle)^{\otimes P}
\end{align}
is upper bounded by (up to the $1/2^{R/2}$ factor)
\begin{align}\label{eq:euclidean_distance_GHZ_shift}
    &\left\|\sum_{\boldsymbol{b}\in \{0,1\}^R}  (|\boldsymbol{b}\rangle V(\bm{b})(|\psi_0\rangle |0^{n_{\rm anc}}\rangle))^{\otimes P}-\sum_{\boldsymbol{b}\in \{0,1\}^R}  e^{-iPT \langle \psi_0|W_{\rm red,H}(\bm{b})|\psi_0\rangle}(|\boldsymbol{b}\rangle |\psi_0\rangle|0^{n_{\rm anc}}\rangle)^{\otimes P}\right\|\notag\\
    &~~~+\left\|  \sum_{\boldsymbol{b}\in \{0,1\}^R}  e^{-iPT \langle \psi_0|W_{\rm red,H}(\bm{b})|\psi_0\rangle}|\boldsymbol{b}\rangle^{\otimes P}-e^{-i\Phi P\tau}\sum_{\boldsymbol{b}\in \{0,1\}^R} e^{iP\tau \sum_{j=1}^R b_j \langle L_j\rangle }|\boldsymbol{b}\rangle^{\otimes P}\right\|,
\end{align}
where we used the triangle inequality.
The first term is evaluated as follows.
\begin{align}\label{eq:first_term_eval_ghz_distance}
    &\left\|\sum_{\boldsymbol{b}\in \{0,1\}^R}  (|\boldsymbol{b}\rangle V(\bm{b})(|\psi_0\rangle |0^{n_{\rm anc}}\rangle))^{\otimes P}-\sum_{\boldsymbol{b}\in \{0,1\}^R}  e^{-iPT \langle \psi_0|W_{\rm red,H}(\bm{b})|\psi_0\rangle}(|\boldsymbol{b}\rangle|\psi_0\rangle|0^{n_{\rm anc}}\rangle)^{\otimes P}\right\|^2\notag\\
    &=\left\|\sum_{\boldsymbol{b}\in \{0,1\}^R}  \ket{\bm{b}}^{\otimes P}\left\{(V(\bm{b})|\psi_0\rangle|0^{n_{\rm anc}}\rangle)^{\otimes P}-(e^{-iT \langle \psi_0|W_{\rm red,H}(\bm{b})|\psi_0\rangle}|\psi_0\rangle|0^{n_{\rm anc}}\rangle)^{\otimes P}\right\}\right\|^2\notag\\
    &=\sum_{\boldsymbol{b}\in \{0,1\}^R} \left\| (V(\bm{b})|\psi_0\rangle|0^{n_{\rm anc}}\rangle)^{\otimes P}-(e^{-iT \langle \psi_0|W_{\rm red,H}(\bm{b})|\psi_0\rangle}|\psi_0\rangle|0^{n_{\rm anc}}\rangle)^{\otimes P}\right\|^2\notag\\
    &\leq \sum_{\boldsymbol{b}\in \{0,1\}^R} 2\left|1-
    \langle \psi_0|\langle 0^{n_{\rm anc}}| V(\bm{b})^\dagger
    e^{-iT \langle \psi_0|W_{\rm red,H}(\bm{b})|\psi_0\rangle}|\psi_0\rangle|0^{n_{\rm anc}}\rangle ^{ P}
    \right|\notag\\
    &\leq \sum_{\boldsymbol{b}\in \{0,1\}^R} 2P\left|1-
    \langle \psi_0|\langle 0^{n_{\rm anc}}| V(\bm{b})^\dagger
    (e^{-iT \langle \psi_0|W_{\rm red,H}(\bm{b})|\psi_0\rangle}|\psi_0\rangle|0^{n_{\rm anc}}\rangle)
    \right|\notag\\
    % &\leq 2^R P^2 \max_{\bm{b}}\left\| V(\bm{b})|\psi_0\rangle|0^{n_{\rm anc}}\rangle-e^{-iT \langle \psi_0|W_{\rm red,H}\notag(\bm{b})|\psi_0\rangle}|\psi_0\rangle|0^{n_{\rm anc}}\rangle\right\|^2\notag\\
    &\leq 2^R\cdot 2P\cdot \left({\varepsilon_{\rm HS}+|T|\varepsilon_{\rm OFT}}\right),
\end{align}
where the second inequality comes from $|1-c^P|\leq P|1-c|$ for complex numbers $|c|\leq 1$, and the last one holds because the construction of $V(\bm{b})$ ensures for any $\bm{b}$
% \begin{equation}
%     \left\|V(\bm{b})(|\psi_0\rangle |0^{n_{\rm anc}}\rangle)-e^{-iT \langle \psi_0|W_{\rm red,H}(\bm{b})|\psi_0\rangle}|\psi_0\rangle |0^{n_{\rm anc}}\rangle\right\|^2\leq \mathcal{O}\left({\varepsilon_{\rm HS}+|T|\varepsilon_{\rm OFT}}\right).
% \end{equation}
\begin{align}
    \left|1-
    \langle \psi_0|\langle 0^{n_{\rm anc}}| V(\bm{b})^\dagger
    (e^{-iT \langle \psi_0|W_{\rm red,H}(\bm{b})|\psi_0\rangle}|\psi_0\rangle|0^{n_{\rm anc}}\rangle)
    \right|&\leq \left|
    \langle \psi_0|\left(\langle 0^{n_{\rm anc}}| V(\bm{b})^\dagger|0^{n_{\rm anc}}\rangle- U^\dagger(\bm{b})\right) U(\bm{b})
    |\psi_0\rangle)
    \right|\notag\\
    &\leq \|\langle 0^{n_{\rm anc}}| V(\bm{b})^\dagger|0^{n_{\rm anc}}\rangle-U^\dagger(\bm{b})\|
\end{align}
for the ideal time evolution unitary $U(\bm{b})$.
The second term in Eq.~\eqref{eq:euclidean_distance_GHZ_shift} is evaluated as follows.
We recall that there exists a subset $F\subseteq \{0,1\}^R$ of the size $|F|\geq (1-\delta_{\rm obs})2^{R}$ such that for any $\bm{b}\in F$, 
\begin{equation}
    \left|\langle \psi_0| W_{\rm red,H}(\bm{b})|\psi_0\rangle-\left(\frac{1}{2\sigma}\sum_{j=1}^R \langle L_j\rangle b_j -\frac{1}{4}\frac{\langle\sum_{j=1}^R L_j \rangle }{\sigma}\right)\right|\leq \varepsilon_{\rm obs}
\end{equation}
holds. Then, 
\begin{align}\label{eq:second_term_eval_ghz_distance}
    &\left\|  \sum_{\boldsymbol{b}\in \{0,1\}^R}  e^{-iPT \langle \psi_0|W_{\rm red,H}(\bm{b})|\psi_0\rangle}|\boldsymbol{b}\rangle^{\otimes P}-e^{-i\Phi P\tau}\sum_{\boldsymbol{b}\in \{0,1\}^R} e^{iP\tau \sum_{j=1}^R b_j \langle L_j\rangle }|\boldsymbol{b}\rangle^{\otimes P}\right\|^2\notag\\
    &= \sum_{\boldsymbol{b}\in \{0,1\}^R}\left|e^{-iPT \langle \psi_0|W_{\rm red,H}(\bm{b})|\psi_0\rangle}-e^{iP \tau(\sum_{j=1}^R b_j \langle L_j\rangle-\Phi) }\right|^2\notag\\
    &\leq 4\cdot |F^c|+\sum_{\boldsymbol{b}\in F}\left|e^{-iPT \langle \psi_0|W_{\rm red,H}(\bm{b})|\psi_0\rangle}-e^{iP\tau (\sum_{j=1}^R b_j \langle L_j\rangle-\Phi) }\right|^2\notag\\
    &\leq 4\cdot |F^c|+\sum_{\boldsymbol{b}\in F}(P\tau)^2\left|2\sigma\langle \psi_0|W_{\rm red,H}(\bm{b})|\psi_0\rangle -\left(\sum_{j=1}^R b_j \langle L_j\rangle-\Phi\right)\right|^2\notag\\
    &\leq 4\cdot 2^R\cdot \delta_{\rm obs}+2^R \cdot (2\sigma \tau P\varepsilon_{\rm obs})^2,
\end{align}
where in the second inequality we used $|e^{ia}-e^{ib}|\leq |a-b|$. In the final line, we used the definition of $\Phi=\langle \sum_j L_j\rangle /2$.

Combining Eqs.~\eqref{eq:first_term_eval_ghz_distance} and~\eqref{eq:second_term_eval_ghz_distance}, we conclude that 
\begin{align}
    &\left\|V_{\rm phase}^{\otimes P}\left(|{\rm GHZ}_P\rangle^{\otimes R} (|\psi_0\rangle |0^{n_{\rm anc}}\rangle)^{\otimes P}\right)-e^{-i\Phi P\tau}\left(\bigotimes_{j=1}^R |{\rm GHZ}_P(\tau \langle L_j\rangle ) \rangle \right) (|\psi_0\rangle |0^{n_{\rm anc}}\rangle)^{\otimes P}\right\|\notag\\
    &\leq \sqrt{P\cdot \mathcal{O}\left({{\varepsilon_{\rm HS}+|T|\varepsilon_{\rm OFT}}}\right)}+\sqrt{4\delta_{\rm obs}+(2\sigma \tau P\varepsilon_{\rm obs})^2}\notag\\
    &\leq \delta_{\rm phase}\label{eq:final_jointphaseshifterror}
\end{align}
holds by choosing $\delta_{\rm obs}=\frac{1}{8}\delta_{\rm phase}^2$, $\varepsilon_{\rm HS}=\mathcal{O}(P^{-1}\delta^2_{\rm phase})$, $\varepsilon_{\rm OFT}=\mathcal{O}(P^{-1}|T|^{-1}\delta^2_{\rm phase})=\mathcal{O}(P^{-1}\sigma^{-1}|\tau|^{-1}\delta^2_{\rm phase})$, and
    $\varepsilon_{\rm obs}=\mathcal{O}(P^{-1}\sigma^{-1}|\tau|^{-1}\delta_{\rm phase})$.
\end{proof}

\begin{lem}
    [Intrinsic matrix-Rademacher bound]\label{lem_Intrinsic_matrix-Rademacher_bound}
    Let $N$ be a positive integer, let
    $\{\widehat{r}_k\}_{k=1}^{N}$ be independent Rademacher
    random variables, and let $\{A_k\}_{k=1}^{N}$ be fixed
    $d$-dimensional Hermitian matrices. Define
    \begin{equation}
        S'
        :=
        \sum_{k=1}^{N} A_k^2.
    \end{equation}
    Then, for every $t>0$ satisfying
    \begin{equation}
        t^2
        \geq
        \|S'\|\ln 4,
    \end{equation}
    we have
    \begin{equation}
        \Pr\left[
            \left\|
                \sum_{k=1}^{N}
                \widehat{r}_k A_k
            \right\|
            \geq t
        \right]
        \leq
        2
        \frac{\operatorname{tr}[S']}{\|S'\|}
        \exp\left(
            -\frac{t^2}{2\|S'\|}
        \right).
    \end{equation}
\end{lem}

\begin{proof}
    Let us take
    $Y:=
        \sum_{k=1}^{N}
        \widehat{r}_k A_k$.
    Lemma~3.5.1 and~Lemma 4.6.3 in Ref.~\cite{tropp2015introduction} directly yield, for every
    $\theta\in\mathbb{R}$,
    \begin{equation}
        \mathbb{E}
        \operatorname{tr}
        \left[
            e^{\theta Y}
        \right]
        \leq
        \operatorname{tr}
        \left[
            \exp\left(
                \frac{\theta^2 S'}{2}
            \right)
        \right].
    \end{equation}
    Applying this inequality with $\theta$ and $-\theta$ gives
    \begin{equation}
        \mathbb{E}
        \operatorname{tr}
        \left[
            \cosh(\theta Y)-I
        \right]
        =\mathbb{E}
        \operatorname{tr}
        \left[
            \frac{e^{\theta Y}+e^{-\theta Y}}{2}-I
        \right]
        \leq
        \operatorname{tr}
        \left[
            \exp\left(
                \frac{\theta^2 S'}{2}
            \right)-I
        \right].
        \label{eq:rademacher_cosh_mgf}
    \end{equation}
    If $\|Y\|\geq t$ holds, then at least one eigenvalue of $Y$ has absolute value at least $t$. Hence, $\operatorname{tr}[\cosh(\theta Y)-I]\geq \cosh(\theta t)-1$ holds because all eigenvalues of
    $\cosh(\theta Y)-I$ are nonnegative.
    From Markov's inequality and Eq.~\eqref{eq:rademacher_cosh_mgf}, we obtain
    \begin{equation}
        \Pr[\|Y\|\geq t]\leq \Pr[\operatorname{tr}[\cosh(\theta Y)-I]\geq \cosh(\theta t)-1]
        \leq
        \frac{
            \operatorname{tr}
            \left[
                \exp\left(
                    \theta^2 S'/2
                \right)-I
            \right]
        }{
            \cosh(\theta t)-1
        }.
        \label{eq:rademacher_intrinsic_intermediate}
    \end{equation}

    For every $\alpha\geq0$, the convexity of
    $x\mapsto e^{\alpha x}-1$ on $[0,\|S'\|]$ implies
    \begin{equation}
        e^{\alpha S'}-I
        \preceq
        \frac{
            e^{\alpha\|S'\|}-1
        }{
            \|S'\|
        }
        S'.
    \end{equation}
    Therefore,
    \begin{equation}
        \operatorname{tr}
        \left[
            e^{\alpha S'}-I
        \right]
        \leq
        \frac{\operatorname{tr}[S']}{\|S'\|}
        \left(
            e^{\alpha\|S'\|}-1
        \right).
    \end{equation}
    Substituting this inequality into
    Eq.~\eqref{eq:rademacher_intrinsic_intermediate} and taking
    $\theta=t/\|S'\|$ yields
    \begin{equation}
        \Pr[\|Y\|\geq t]
        \leq
        \frac{\operatorname{tr}[S']}{\|S'\|}
        \frac{
            \exp\left(
                t^2/(2\|S'\|)
            \right)-1
        }{
            \cosh\left(
                t^2/\|S'\|
            \right)-1
        }.
    \end{equation}
    Finally, for every $u\geq\ln 4$,
    \begin{equation}
        \frac{e^{u/2}-1}{\cosh(u)-1}
        =
        \frac{
            2e^{-u/2}
        }{
            (1-e^{-u/2})(1+e^{-u/2})^2
        }
        \leq
        2e^{-u/2},
    \end{equation}
    where the last inequality follows from
    $e^{-u/2}\leq1/2$ and
    $(1-x)(1+x)^2\geq1$ for $x\in[0,1/2]$.
    Applying this inequality with $u=t^2/\|S'\|$ completes the
    proof.
\end{proof}

\section{Benchmarks for the critical one-copy preparation cost}

Ground-state preparation and readout should be compared in a common quantum
resource.
For gapped ground-state preparation, the total amount of
Hamiltonian-evolution time is a standard cost metric in adiabatic
methods~\cite{jansen2007bounds,albash2018adiabatic,lutz2026adiabatic}, filtering-based methods~\cite{ge2019faster,lin2020near,dong2022ground},
and dissipative methods~\cite{PhysRevResearch.6.033147,wzb3-dbg9}.
We therefore compare the copy-consuming and catalytic approaches using the 
% \textit{sequential total Hamiltonian-evolution time}
elapsed Hamiltonian-evolution time.
The following comparisons cover full computational-basis distribution learning, stabilizer-state fidelity prediction (a canonical application of Clifford classical shadows), selected-orbital fermionic reduced density matrices relevant to quantum chemistry and correlated fermion systems, and local-Pauli expectation values.

\subsection{Critical one-copy preparation cost}
\label{sec:critical_cost_definition}

Let $t_{\rm prep}$ denote the elapsed Hamiltonian-evolution time
required to prepare one ground-state copy on one device.
If a conventional copy-consuming approach needs $N_{\rm copy}$ copies and $P$ devices prepare
copies in parallel, its preparation time is
\begin{equation}
    \mathcal T_{\rm copy}
    =
    \left\lceil\frac{N_{\rm copy}}{P}\right\rceil t_{\rm prep}.
    \label{eq:copy_consuming_total_time}
\end{equation}
The catalytic approach prepares the initial stock once and then performs the
coherent readout,
\begin{equation}
    \mathcal T_{\rm cat}=t_{\rm prep}+T_{\rm obs}^{\rm cat},
    \label{eq:catalytic_total_time}
\end{equation}
where $T_{\rm obs}^{\rm cat}$ is the elapsed controlled
Hamiltonian-evolution time of the readout. 
The two approaches have equal total Hamiltonian time at
\begin{align}
    t_{\rm crit}
    :=
    \frac{T_{\rm obs}^{\rm cat}}
    {\left\lceil N_{\rm copy}/P\right\rceil-1},~~~
    \Delta t_{\rm crit}
    =
    \frac{\Delta T_{\rm obs}^{\rm cat}}
    {\left\lceil N_{\rm copy}/P\right\rceil-1}
    \simeq
    \frac{P\Delta T_{\rm obs}^{\rm cat}}{N_{\rm copy}}.
    \label{eq:tcrit_approx}
\end{align}
The approximation is accurate in the typical regime $N_{\rm copy}/P\gg1$.
Thus, $t_{\rm prep}>t_{\rm crit}$ selects the prepare-once, coherent-readout
approach.

\subsection{Numerical evaluation of the catalytic readout cost}
\label{sec:numerical_catalytic_readout_cost}

We numerically evaluate the catalytic readout cost
$T_{\rm obs}^{\rm cat}$.
The calculation follows the proof of
Theorem~\ref{thm:main_theorem_for_multiple_observable_est}, with several refinements on constants and poly-logarithmic factors described below.
We use the notation of
Lemmas~\ref{lem:catalytic_joint_phase_shifter} and
\ref{lem:catalytic_almost_linear_block_hamiltonian}.
We set $P=20$ throughout benchmarks.

\paragraph{Phase-readout schedule and common confidence amplification.}

At stage $k=1,\ldots,K$ of the multi-signal readout procedure, the parameter $\tau_k$ on phase imprinting is given by Eq.~\eqref{eq:kstage_tau}.
We use the numerically optimized number of independent measurements~\cite{PhysRevA.102.042613,dutkiewicz2025error,oshio2025near}
\begin{equation}
    v_k
    =
    \left\lceil
    \alpha_{\rm ph}(K-k)+\beta_{\rm ph}
    \right\rceil,
    \qquad
    \alpha_{\rm ph}=4.0835,
    \qquad
    \beta_{\rm ph}=11,
    \label{eq:benchmark_vk}
\end{equation}
and define
\begin{equation}
    Q_{\rm phase}(K)
    :=
    \sum_{k=1}^{K}2v_k2^{k-1}.
    \label{eq:benchmark_Qphase}
\end{equation}
For one ideal data set (Eq.~\eqref{eq:data_set_single_signal}) with a single signal, the empirical phase error is
well approximated by $5\pi/Q_{\rm phase}(K)$~\cite{dutkiewicz2025error,oshio2025near}.
Hence, we select the smallest $K$ satisfying
\begin{equation}
    \frac{5\pi}{Q_{\rm phase}(K)}
    \leq
    \varepsilon_{d_\mathbb{T}}\sqrt{p_\star},
    \qquad
    p_\star=0.1.
\end{equation}
For expectation values in $[-1,1]$, the success event
$d_\mathbb{T}(\widehat x_j,x_j)\leq\varepsilon$ coincides with the usual additive-error event.
In the four benchmarks below, we use
\begin{equation}
    \varepsilon_{d_\mathbb{T}}
    =
    \begin{cases}
      \eta_{\ell_1}/d,
      & \text{full computational-basis distribution learning},\\
      \varepsilon,
      & \text{fidelities},\\
      \varepsilon/\sqrt2,
      & \text{each real or imaginary RDM coordinate},\\
      \varepsilon,
      & \text{local Pauli expectation values}.
    \end{cases}
    \label{eq:task_to_phase_error}
\end{equation}
All methods use the same base failure probability $p_\star=0.1$ and the same median amplification formula.
Hoeffding's inequality gives
\begin{equation}
    \Upsilon_{\rm med}(\delta_{\rm med})
    =
    \operatorname{oddceil}
    \left[
      \frac{25}{8}
      \ln\left(\frac{M}{\delta_{\rm med}}\right)
    \right],
    \qquad
    \delta_{\rm med}=\frac{\delta_{\rm est}}{4},
    \qquad
    \delta_{\rm est}=0.05,
    \label{eq:common_median_count}
\end{equation}
where we fix the readout failure probability $0.05$ throughout benchmarks.

\paragraph{Stage-dependent error allocation.}

We slightly refine the proof of
Theorem~\ref{thm:main_theorem_for_multiple_observable_est}.  Let
$\delta_{\rm phase}^{[k]}$ be the Euclidean-distance error of the complete
$P$-parallel call
$\bigl(V_{\rm phase}^{[k]}\bigr)^{\otimes P}$.  We require
\begin{equation}
    2\Upsilon_{\rm med}
    \sum_{k=1}^{K}v_k\delta_{\rm phase}^{[k]}
    \leq
    \zeta_{\rm phase},
    \qquad
    \zeta_{\rm phase}=\frac{\delta_{\rm ret}}{4},
    \qquad
    \delta_{\rm ret}=0.05.
    \label{eq:phase_error_telescoping}
\end{equation}
By the same argument as in
Theorem~\ref{thm:main_theorem_for_multiple_observable_est}, this choice ensures that the returned stock is within trace distance $1/2$ except with failure probability at most $\delta_{\rm ret}$.
We use a simple ansatz
\begin{equation}
    \delta_{\rm phase}^{[k]}(p)
    =
    \zeta_{\rm phase}
    \frac{\tau_k^p}
    {2\Upsilon_{\rm med}\sum_{\ell=1}^{K}v_\ell\tau_\ell^p},
    \label{eq:stage_dependent_phase_error}
\end{equation}
where $p$ is optimized over a fixed finite grid.

\paragraph{Refined resource evaluation in
Lemma~\ref{lem:catalytic_joint_phase_shifter}.}

For a control string $\bm b$, write
$W_{\bm b}:=W_{\mathrm{red,H}}(\bm b)$ and define the ideal block-diagonal
Hamiltonian appearing in
Lemma~\ref{lem:catalytic_almost_linear_block_hamiltonian} by
\begin{equation}
    G_{\bm b}
    :=
    P_0W_{\bm b}P_0
    +\Pi_\Delta\mathcal E_H(W_{\bm b})\Pi_\Delta,
    \qquad
    P_0=\ket{\psi_0}\!\bra{\psi_0},
    \qquad
    \Pi_\Delta=I-P_0.
    \label{eq:ideal_filtered_generator}
\end{equation}
In the numerical evaluation, the compact-support OFT specified below defines a
Hermitian filtered generator $\widetilde G_{\bm b}$ satisfying
\begin{equation}
    \|\widetilde G_{\bm b}-G_{\bm b}\|
    \leq
    \varepsilon_{\rm OFT}^{[k]}.
    \label{eq:OFT_generator_error}
\end{equation}
Writing $T_k^{\rm kick}=-2\sigma^{[k]}\tau_k$, we have
\begin{equation}
    \left\|
    e^{-iT_k^{\rm kick}\widetilde G_{\bm b}}
    -e^{-iT_k^{\rm kick}G_{\bm b}}
    \right\|
    \leq
    |T_k^{\rm kick}|\varepsilon_{\rm OFT}^{[k]}.
    \label{eq:Duhamel_OFT_error}
\end{equation}
Tensor-product telescoping expansion therefore bounds the OFT error contribution to the
complete $P$-parallel call by
\begin{equation}
    P|T_k^{\rm kick}|\varepsilon_{\rm OFT}^{[k]}.
    \label{eq:P_parallel_OFT_error}
\end{equation}
If block-Hamiltonian simulation produces an
$\varepsilon_{\rm HS}^{[k]}$-precise block encoding of
$e^{-iT_k^{\rm kick}\widetilde G_{\bm b}}$, the overlap argument Eq.~\eqref{eq:first_term_eval_ghz_distance} in the proof
of Lemma~\ref{lem:catalytic_joint_phase_shifter} shows the error contribution from the block-Hamiltonian simulation error
\begin{equation}
    \sqrt{2P\varepsilon_{\rm HS}^{[k]}}.
    \label{eq:P_parallel_HS_error}
\end{equation}
Consequently, we have a refined evaluation of $\delta_{\rm phase}^{[k]}$, instead of Eq.~\eqref{eq:final_jointphaseshifterror},
\begin{equation}
    \delta_{\rm phase}^{[k]}= \sqrt{2P\varepsilon_{\rm HS}^{[k]}}
    +P|T_k^{\rm kick}|\varepsilon_{\rm OFT}^{[k]} + \sqrt{4\delta_{\rm obs}^{[k]}+(2\sigma^{[k]}\tau_k P \varepsilon_{\rm obs}^{[k]})^2}.
    \label{eq:refined_stage_error}
\end{equation}
The parameter $\varepsilon_{\rm obs}^{[k]}$ affects observable-block-encoding queries but not the Hamiltonian-evolution time, so we take it sufficiently small and omit its contribution from the present benchmark.

We then allocate the error budget to
\begin{align}
    \delta_{\rm joint}^{[k]}=
    f_{\rm joint}\delta_{\rm phase}^{[k]},~~~\delta_{\rm impl}^{[k]}
    &=
    (1-f_{\rm joint})\delta_{\rm phase}^{[k]},~~~\delta_{\rm obs}^{[k]}
    =
    \frac{\bigl(\delta_{\rm joint}^{[k]}\bigr)^2}{4},
\end{align}
and 
\begin{align}
    \sigma^{[k]}=
    \min\left\{
      M,
      \sqrt{\frac{\overline s}{2}
      \ln\left(
      \frac{8\overline r}{(\delta_{\rm joint}^{[k]})^2}
      \right)}
    \right\}.
\end{align}
Here, $\bar{s}$ and $\bar{r}$ denote the upper bounds of $\|S_M\|$ and ${\rm tr}[S_M]/\|S_M\|$, respectively:
\begin{equation}
    \left\|S_M\right\|\leq \bar{s},~~~\frac{{\rm tr}[S_M]}{\|S_M\|}\leq \bar{r}.
\end{equation}
Also, we split the remaining budget as
\begin{align}
    \varepsilon_{\rm HS}^{[k]}
=
    \frac{f_{\rm HS}^2(\delta_{\rm impl}^{[k]})^2}{2P},~~~
    \varepsilon_{\rm OFT}^{[k]}
=
    \frac{(1-f_{\rm HS})\delta_{\rm impl}^{[k]}}
    {P|T_k^{\rm kick}|}.
    \label{eq:OFT_allocation}
\end{align}
These choices satisfy
\begin{equation}
    2\sqrt{\delta_{\rm obs}^{[k]}}
    +\sqrt{2P\varepsilon_{\rm HS}^{[k]}}
    +P|T_k^{\rm kick}|\varepsilon_{\rm OFT}^{[k]}
    =
    \delta_{\rm phase}^{[k]}.
    \label{eq:stage_error_saturation}
\end{equation}
In the numerical evaluation, we optimize over
\begin{equation}
    f_{\rm joint}\in\{0.15,0.25,0.35,0.45,0.55\},~~~
    f_{\rm HS}\in\{0.20,0.40,0.60,0.80\}.
\end{equation}

\paragraph{Kaiser--Bessel OFT and GQSP simulation.}

We retain the OFT framework of
Lemma~\ref{lem:block_diagonalization} by Ref.~\cite{chen2025catalytic}, but in the numerical evaluation replace the Gaussian
filter by the positive, normalized, compact-support Kaiser--Bessel filter~\cite{1163349}.
In the Fourier convention of Lemma~\ref{lem:block_diagonalization},
we define
\begin{equation}
    f_{\beta,T}^{\rm KB}(t)
    :=
    \sqrt{2\pi}\,
    \frac{\beta}{2T\sinh\beta}
    I_0\!\left(
      \beta\sqrt{1-\frac{t^2}{T^2}}
    \right)
    \mathbf 1_{\{|t|\leq T\}},~~~\beta>0,~~~T>0,
    \label{eq:KB_filter}
\end{equation}
where $I_0$ is the modified Bessel function.
Positivity follows from
$I_0(x)>0$ for $x\geq0$.
With $x=t/T$ and
\begin{equation}
    \int_{-1}^{1}
    I_0\!\left(\beta\sqrt{1-x^2}\right)\,dx
    =
    \frac{2\sinh\beta}{\beta},
    \label{eq:KB_normalization_identity}
\end{equation}
we obtain
\begin{equation}
    \frac{1}{\sqrt{2\pi}}
    \int_{-T}^{T}f_{\beta,T}^{\rm KB}(t)\,dt
    =1.
    \label{eq:KB_normalization}
\end{equation}
Consequently, the corresponding LCU normalization is one, so no additional LCU-normalization factor multiplies
$|T_k^{\rm kick}|$.
The corresponding Fourier transform,
\begin{equation}
    \widehat f_{\beta,T}^{\rm KB}(\omega)
    :=
    \frac{1}{\sqrt{2\pi}}
    \int_{-T}^{T}
    f_{\beta,T}^{\rm KB}(t)e^{-i\omega t}\,dt,
    \label{eq:KB_Fourier_definition}
\end{equation}
is represented by~\cite{1163349}
\begin{equation}
    \widehat f_{\beta,T}^{\rm KB}(\omega)
    =
    \begin{cases}
      \displaystyle
      \frac{\beta}{\sinh\beta}
      \frac{\sinh\!\sqrt{\beta^2-(\omega T)^2}}
      {\sqrt{\beta^2-(\omega T)^2}},
      & |\omega|T\leq\beta,\\[3mm]
      \displaystyle
      \frac{\beta}{\sinh\beta}
      \frac{\sin\!\sqrt{(\omega T)^2-\beta^2}}
      {\sqrt{(\omega T)^2-\beta^2}},
      & |\omega|T\geq\beta.
    \end{cases}
    \label{eq:KB_Fourier_transform}
\end{equation}
At $|\omega|T=\beta$, both branches take the continuous-limit value
$\beta/\sinh\beta$, and
$\widehat f_{\beta,T}^{\rm KB}(0)=1$.

For stage $k$, we take the support range
\begin{equation}
    T_{\rm OFT}^{[k]}=\frac{\beta_k}{\Delta}
    \label{eq:KB_support_time}
\end{equation}
with $\beta_k>0$ specified below.
For every $|\omega|\geq\Delta$, setting
$y=\sqrt{(\omega T_{\rm OFT}^{[k]})^2-\beta_k^2}$ gives
\begin{equation}
    \left|
    \widehat f_{\beta_k,T_{\rm OFT}^{[k]}}^{\rm KB}(\omega)
    \right|
    =
    \frac{\beta_k}{\sinh\beta_k}
    \left|\frac{\sin y}{y}\right|
    \leq
    \frac{\beta_k}{\sinh\beta_k}.
    \label{eq:KB_stopband_bound}
\end{equation}
Because $W_{\bm b}$ is Hermitian, the ground--excited error has the off-diagonal
block form
\begin{equation}
    \hat{W}_{\bm b,f}
    -P_0W_{\bm b}P_0
    -\Pi_\Delta\hat W_{\bm b,f}\Pi_\Delta
    =
    \begin{pmatrix}
      0 & B_{\bm b}^{\dagger}\\
      B_{\bm b} & 0
    \end{pmatrix},
    \qquad
    B_{\bm b}=\Pi_\Delta\hat{W}_{\bm b,f}P_0.
    \label{eq:KB_offdiagonal_block}
\end{equation}
Its operator norm equals $\|B_{\bm b}\|$, and hence, using
$\|W_{\bm b}\|\leq1$,
\begin{equation}
    \left\|
    \hat W_{\bm b,f}
    -P_0W_{\bm b}P_0
    -\Pi_\Delta\hat W_{\bm b,f}\Pi_\Delta
    \right\|
    \leq
    \sup_{|\omega|\geq\Delta}
    \left|\widehat f_{\beta_k,T_{\rm OFT}^{[k]}}^{\rm KB}(\omega)\right|
    \leq
    \frac{\beta_k}{\sinh\beta_k}.
    \label{eq:KB_operator_error}
\end{equation}
We therefore choose $\beta_k$ as the unique positive solution of
\begin{equation}
    \frac{\beta_k}{\sinh\beta_k}
    =
    \varepsilon_{\rm OFT}^{[k]}.
    \label{eq:KB_beta_equation}
\end{equation}
Note that we set the discretization error to zero in the benchmark; this simplification almost preserves the actual Hamiltonian time.

It remains to count GQSP-based~\cite{motlagh2024generalized,berry2024doubling} block-Hamiltonian simulation queries.
For $\tau'=|T_k^{\rm kick}|$ and an integer $d>\tau'$, define the Kapteyn bound by Ref.~\cite{pocrnic2026improved}
\begin{equation}
    \mathcal K(\tau',d)
    :=
    \frac{
      2\exp\!\left[
      -d\operatorname{arcosh}(d/\tau')
      +\sqrt{d^2-(\tau')^2}
      \right]
    }{
      1-\exp[-\operatorname{arcosh}(d/\tau')]
    }.
    \label{eq:Kapteyn_bound}
\end{equation}
We define
\begin{equation}
    d_k
    :=
    \min\left\{
      d\in\mathbb Z_{>0}:
      d>|T_k^{\rm kick}|,\quad
      \mathcal K(|T_k^{\rm kick}|,d)
      \leq\varepsilon_{\rm HS}^{[k]}
    \right\}.
    \label{eq:Kapteyn_first_omitted_order}
\end{equation}
The GQSP Hamiltonian-simulation construction of
Ref.~\cite{berry2024doubling} uses $d_k+2$ queries to the filtered block-Hamiltonian encoding or its inverse.
Since one Kaiser--Bessel OFT requires the Hamiltonian time $2\beta_k/\Delta$, one stage-$k$ cost is evaluated as
\begin{equation}
    (d_k+2)\cdot 2\beta_k/\Delta.
\end{equation}
There are $2v_k\Upsilon_{\rm med}$ stage-$k$ calls.
The final catalytic readout cost is therefore
\begin{equation}
    \Delta T_{\rm obs}^{\rm cat}
    =
    4\Upsilon_{\rm med}
    \sum_{k=1}^{K}
    v_k(d_k+2)\beta_k.
    \label{eq:final_readout_cost}
\end{equation}

\subsection{Very optimistic preparation-cost baseline}

To get a meaningful baseline, we here combine the second-order
(discrete) Lindblad simulation cost of
Ref.~\cite{PhysRevResearch.6.033147} with the empirical rapid-mixing time fit observed in 1D transverse field Ising model under bulk dissipation in Ref.~\cite{wzb3-dbg9}.
We deliberately set the simulation-cost prefactor to one and omit multiple-jump, OFT-overhead, and all implementation-constant overheads, resulting in the sequential Hamiltonian time $\Delta^{-1}t_{\rm sim}^{3/2}\varepsilon_{\rm sim}^{-1/2}$ for the simulation time $t_{\rm sim}$ and error $\varepsilon_{\rm sim}$~\cite{PhysRevResearch.6.033147}.
This is therefore a \textit{very optimistic preparation-cost baseline}, which can be much smaller than an actual implementation cost.

Let $N$ be the number of 1D cites, let $\eta_{\rm mix}$ denote the fidelity-based mixing error~\cite{wzb3-dbg9}.
We extrapolate the observed rapid-mixing time in Ref.~\cite{wzb3-dbg9} at fixed error $\eta_{\rm mix}=1/2$ as
\begin{equation}
    \tau_{\rm mix}(N,\eta_{\rm mix})
    =
    5\ln\left(\frac{N}{2\eta_{\rm mix}}\right)-1.4,
    \label{eq:empirical_mixing_law}
\end{equation}
and define a very optimistic preparation time
\begin{equation}
    \Delta t_{\rm prep}^{\rm opt}(N,e)
    :=
    \min_{\substack{\varepsilon_{\rm sim},\eta_{\rm mix}>0\\
    \varepsilon_{\rm sim}+2\sqrt{\eta_{\rm mix}}\leq e}}
    \frac{\tau_{\rm mix}(N,\eta_{\rm mix})^{3/2}}
    {\sqrt{\varepsilon_{\rm sim}}}.
    \label{eq:optimistic_prep_cost}
\end{equation}
The same error $e$ is assigned to preparation and readout, for simplicity; we substitute
\begin{equation}
    e=
    \begin{cases}
      \eta_{\ell_1}, & \text{full distribution learning},\\
      \varepsilon, & \text{fidelity learning},\\
      \varepsilon, & \text{complex RDM entries},\\
      \varepsilon, & \text{local Pauli expectation values}.
    \end{cases}
    \label{eq:task_prep_accuracy}
\end{equation}
The mixing-time fit in Eq.~\eqref{eq:empirical_mixing_law} comes from a bulk-dissipative Lindbladian with many coupling operators~\cite{wzb3-dbg9}, whereas the unit-prefactor simulation cost comes from a single-jump Lindbladian.
This mismatch is intentional and makes the preparation baseline favorable to the conventional copy-consuming approaches in the following benchmarks.

\subsection{Tasks and specialized copy-consuming comparators}

To make a fair comparison, for each task we use an explicit
copy-consuming protocol specialized to the target observable set and exploiting its available algebraic structure.

\paragraph{Full computational-basis probability distribution learning.}

For a $d$-dimensional ground state, we take
\begin{equation}
    O_x=|x\rangle\langle x|,
    \qquad x=0,\ldots,d-1.
\end{equation}
Then
\begin{equation}
    M=d,
    \qquad
    \sum_x O_x^2=\bm{1},
    \qquad
    \overline s=1,
    \qquad
    \overline r=d.
    \label{eq:distribution_geometry_methods}
\end{equation}
To guarantee
$\|\widehat{\bm p}-\bm p\|_1\leq\eta_{\ell_1}$, we perform the catalytic readout with coordinate-wise additive error $\eta_{\ell_1}/d$ and obtain a distribution estimate $\widehat{\bm p}$.
As an optimistic direct-sampling comparator,
we use
\begin{equation}
    N_{\rm copy}^{\rm dist}
    =
    \left\lceil\frac{d}{\eta_{\ell_1}^2}\right\rceil.
\end{equation}

\paragraph{Stabilizer-state fidelity prediction.}

Let $\{\ket{\phi_j}\}_{j=1}^{M}$ be a catalog of independently sampled
$d$-dimensional stabilizer states.
The target quantities are
\begin{equation}
    F_j=\langle\psi_0|\Pi_j|\psi_0\rangle,
    \qquad
    \Pi_j=|\phi_j\rangle\langle\phi_j|.
    \label{eq:Clifford_targets_methods}
\end{equation}
For $S_M=\sum_j\Pi_j$, the property of state 1-design on the random state $\ket{{\phi}_j}$ gives $\mathbb E[S_M]=(M/d)\bm{1}$.
We then use the matrix-Chernoff tail~\cite{tropp2015introduction}
\begin{equation}
    \Pr\!\left[\|S_M\|\geq s\right]
    \leq
    d\,e^{s-M/d}
    \left(\frac{M/d}{s}\right)^s
\end{equation}
and choose the smallest $\overline s$ at catalog-set failure probability $1/d$, and use $\overline r=M$.

The copy-consuming comparator specialized to this task is the global-Clifford classical shadow~\cite{huang2020predicting}.
For a rank-one target projector $\Pi$, the global-Clifford
single-shot shadow estimator is
\begin{equation}
    \hat F=(d+1)\operatorname{tr}(\Pi \hat{Q})-1,
\end{equation}
where $\hat{Q}$ is the measured stabilizer-state projector.
Writing $F={\rm tr}[\rho \Pi]$ for the target state $\rho=|\psi_0\rangle\langle \psi_0|$, the projective three-design identity gives
\begin{equation}
    \operatorname{Var}_{\rho}(\widehat F)
    =
    \frac{d(1+2F)}{d+2}-F^2
    =
    \frac{2d(d+1)}{(d+2)^2}
    -
    \left(F-\frac{d}{d+2}\right)^2.
\end{equation}
Hence, the exact \textit{worst-case} variance is
\begin{equation}
    V_d
    :=
    \sup_{\rho}\operatorname{Var}_{\rho}(\widehat F)
    =
    \frac{2d(d+1)}{(d+2)^2}<2.
\end{equation}
This is tighter than applying the generic global-Clifford shadow-norm bound~\cite{huang2020predicting}, which approaches $3$ for a rank-one projector.
Under the common base failure probability and median convention, we have via Chebyshev inequality
\begin{equation}
    N_{\rm copy}^{\rm Cliff}
    =
    \Upsilon_{\rm med}(\delta_{\rm est})
    \left\lceil
      \frac{V_d}{p_\star\varepsilon^2}
    \right\rceil.
\end{equation}

\paragraph{Selected-orbital fermionic reduced density matrices.}

Let us consider a globally number-conserving ground state $\ket{\psi_0}$ on $N$ fermionic modes with
total particle number $\eta$.
Our goal here is to estimate the complete low-order
correlation inside a selected orbital subspace $A\subseteq [N]$ of $m$ modes.
Writing $E=A^{\rm c}$, the $\eta$-particle Hilbert space $\mathcal H_{\eta}$ and the reduced state $\rho^{(A)}:={\rm tr}_{\rm E}[|{\psi_0}\rangle\langle \psi_0|]$ can be decomposed into
\begin{align}
    \mathcal H_{\eta}
    \simeq
    \bigoplus_{\ell=\ell_{\min}}^{\ell_{\max}}
    \mathcal H_{\ell}^{(A)}\otimes
    \mathcal H_{\eta-\ell}^{(E)},~~~
    \rho^{(A)}
    =
    \bigoplus_{\ell=\ell_{\min}}^{\ell_{\max}}
    p_{\ell}\rho^{(A)}_{\ell},
    \label{eq:globla_structure_kRDM}
\end{align}
where $\mathcal{H}_{\ell}^{(A)}$ denotes the $\ell$-particle subspace in $A$, $\mathcal{H}_{\eta-\ell}^{(E)}$ denotes the $(\eta-\ell)$-particle subspace in $E$, $\rho^{(A)}_\ell$ is an $\ell$-particle density matrix on $A$, and 
\begin{equation}
    \ell_{\min}=\max\{0,\eta-(N-m)\},
    \qquad
    \ell_{\max}=\min\{\eta,m\}.
\end{equation}
Thus the number of particles in the reduced state can fluctuate with an (unknown) probability $p_\ell$ even though the global particle number is fixed.
Let $\mathcal I_{A,k}$ be the set of increasing $k$-tuples of mode indices in $A$.
For $\bm p,\bm q\in\mathcal I_{A,k}$, we define the $k$-body reduced density matrix ($k$-RDM) $D^{(k)}$ of a global state $\rho$ as
\begin{equation}
    A_{\bm p}^{\bm q}
    =a_{p_1}^{\dagger}\cdots a_{p_k}^{\dagger}
     a_{q_k}\cdots a_{q_1},
    \qquad
    D_{\bm p,\bm q}^{(k)}
    =\operatorname{tr}\!\left(\rho A_{\bm p}^{\bm q}\right),
    \label{eq:active_RDM_entry_methods}
\end{equation}
where $a_{p}^\dagger$ and $a_p$ denote the fermionic creation and annihilation operators on mode $p$, respectively.
To obtain all entries of this $k$-RDM, we take
Hermitian target observables as
\begin{equation}
    O_{\bm p}^{\rm diag}=A_{\bm p}^{\bm p},
    \qquad
    O_{\bm p,\bm q}^{\rm R}
    =\frac{A_{\bm p}^{\bm q}+A_{\bm q}^{\bm p}}{2},
    \qquad
    O_{\bm p,\bm q}^{\rm I}
    =\frac{A_{\bm p}^{\bm q}-A_{\bm q}^{\bm p}}{2i}
    \quad(\bm p\prec\bm q).
\end{equation}
for an arbitrary ordering $\prec$ of $\mathcal I_{A,k}$.
They return the diagonal entries and the real and imaginary parts of all off-diagonal entries.

We next derive the parameters $\bar{s}_{A,k}$ and $\bar{r}_{A,k}$ from those observables for the catalytic readout.
Let
\begin{equation}
    S_{A,k}
    :=
    \sum_{\bm p}\bigl(O_{\bm p}^{\rm diag}\bigr)^2
    +\sum_{\bm p\prec\bm q}
    \left[
      \bigl(O_{\bm p,\bm q}^{\rm R}\bigr)^2
      +\bigl(O_{\bm p,\bm q}^{\rm I}\bigr)^2
    \right].
\end{equation}
Consider arbitrary occupation-basis state $|X\rangle$ having $\ell$ particles on the selected modes.
The diagonal observable $A_{\bm p}^{\bm p}$ preserves $\ket{X}$ when $\bm{p}\subseteq X$; otherwise, the basis state is in the kernel. 
Hence,
\begin{equation}
    \sum_{\bm p}
    \bigl(O_{\bm p}^{\rm diag}\bigr)^2|X\rangle
    =\binom{\ell}{k}|X\rangle.
\end{equation}
For an unordered off-diagonal pair $\bm p\ne\bm q$, the two Hermitian
components follow
\begin{equation}
    \bigl(O_{\bm p,\bm q}^{\rm R}\bigr)^2
    +\bigl(O_{\bm p,\bm q}^{\rm I}\bigr)^2
    =
    \frac12
    \left(
      A_{\bm p}^{\bm q}A_{\bm q}^{\bm p}
      +A_{\bm q}^{\bm p}A_{\bm p}^{\bm q}
    \right)
\end{equation}
and then we have
\begin{equation}
    \sum_{\bm p\prec\bm q}
    \left[
      \bigl(O_{\bm p,\bm q}^{\rm R}\bigr)^2
      +\bigl(O_{\bm p,\bm q}^{\rm I}\bigr)^2
    \right]=\frac12\sum_{\bm p\prec\bm q}
    \left(
      A_{\bm p}^{\bm q}A_{\bm q}^{\bm p}
      +A_{\bm q}^{\bm p}A_{\bm p}^{\bm q}
    \right)=\frac12\sum_{\bm{p},\bm{q}\in \mathcal{I}_{A,k}, ~\bm{p\neq q}} 
      A_{\bm q}^{\bm p}A_{\bm p}^{\bm q}.
\end{equation}
To count their contribution, we first choose the annihilated tuple $\bm q\subseteq X$, which has $\binom{\ell}{k}$ choices.  
After annihilating $\bm q$, the creation tuple $\bm p$ may retain $k-r$ modes of $\bm q$ and replace the remaining $r$ modes by unoccupied modes in $A\setminus (X\cap A)$.
For fixed $r$, there are $\binom{k}{r}\binom{m-\ell}{r}$ choices.  Excluding $r=0$, which is the diagonal case $\bm p=\bm q$, we then obtain
\begin{equation}
    \left(\sum_{\bm{p},\bm{q}\in \mathcal{I}_{A,k}, ~\bm{p\neq q}} 
      A_{\bm p}^{\bm q}A_{\bm q}^{\bm p}\right)\ket{X}=\binom{l}{k}\left(\sum_{r=1}^{k}
    \binom{k}{r}\binom{m-\ell}{r}\right)\ket{X}=\binom{l}{k}\left(\binom{m-\ell+k}{k}-1\right)\ket{X}.
\end{equation}
Adding the diagonal
and off-diagonal contributions gives
\begin{equation}
    S_{A,k}|X\rangle
    =
    \left\{
      \binom{\ell}{k}
      +\frac12\binom{\ell}{k}
       \left[\binom{m-\ell+k}{k}-1\right]
    \right\}|X\rangle.
\end{equation}
The coefficient depends only on $\ell$, not on the occupation set $X$.
Thus $S_{A,k}$ is proportional to the identity in every local-number sector,
with coefficient
\begin{equation}
    s_{m,\ell,k}
    =
    \frac12\binom{\ell}{k}
    \left[
      \binom{m-\ell+k}{k}+1
    \right].
    \label{eq:RDM_sector_square_sum_methods}
\end{equation}
On the global fixed-$\eta$ sector, the full square-sum operator therefore has
the block form
\begin{equation}
    S_{A,k}
    =
    \bigoplus_{\ell=\ell_{\min}}^{\ell_{\max}}
    s_{m,\ell,k}\,
    \bm{1}_{\mathcal H_{\ell}^{(A)}}\otimes
    \bm{1}_{\mathcal H_{\eta-\ell}^{(E)}},
\end{equation}
thereby yielding 
\begin{equation}
    \overline s_{A,k}
    =\max_{\ell_{\min}\leq\ell\leq\ell_{\max}}s_{m,\ell,k},
\end{equation}
The $\ell$th block has dimension
$\binom{m}{\ell}\binom{N-m}{\eta-\ell}$, so inserting the block decomposition into the intrinsic-rank definition
${\rm tr}(S_{A,k})/\|S_{A,k}\|$ yields
\begin{align}
    \overline r_{A,k}
    &=
    \frac{1}{\overline s_{A,k}}\displaystyle
      \sum_{\ell=\ell_{\min}}^{\ell_{\max}}
      s_{m,\ell,k}
      \binom{m}{\ell}
      \binom{N-m}{\eta-\ell}.
    \label{eq:RDM_geometry_methods}
\end{align}

The copy-consuming comparator is the fermionic classical shadow tomography specialized to the particle number symmetry~\cite{low2022classical,koizumi2026provably}.
The protocol first performs a random orbital-rotation unitary and then measures the final state by the occupation basis.
Conditioning on the observed particle number $\ell$ on $A$, we can construct an estimator of $k$-RDM for $\rho_\ell^{(A)}$; averaging
over $\ell$ therefore gives an unbiased estimator of the target $k$-RDM due to the global structure Eq.~\eqref{eq:globla_structure_kRDM}.
For the order $k$, Ref.~\cite{koizumi2026provably} shows that the (conditioned) single-shot estimator has the element-wise variance at most
\begin{equation}
    V_{A,k}^{\rm OR}:=(k+1)\binom{\min{\{m,\eta\}}}{k}\cdot \sum_{j=1}^k 1/j.
\end{equation}
The remaining contribution to the overall variance is a small constant, so we ignore it.
Consequently, the number of copies in this copy-consuming protocol is, under the common base failure probability and median convention,
\begin{equation}
    N_{\rm copy}^{\rm RDM}(k)
    =
    \Upsilon_{\rm med}(\delta_{\rm est})
    \left\lceil
      \frac{V_{A,k}^{\rm OR}}{p_\star\varepsilon^2}
    \right\rceil.
\end{equation}

\paragraph{Local-Pauli expectation values.}

As a target observable set, let us take
\begin{equation}
    \mathcal P=\{P_1,\ldots,P_M\}
    \subseteq\{I,X,Y,Z\}^{\otimes N},
\end{equation}
a catalog of distinct $N$-qubit Pauli strings with
weight at most $k$.
Every Pauli string satisfies $P_j^2=\bm{1}$, and therefore the parameters $\bar{s}$ and $\bar{r}$ for the catalytic readout are given by
\begin{equation}
    \overline s=M,
    \qquad
    \overline r=2^N.
\end{equation}

We choose the local-Clifford classical shadow~\cite{huang2020predicting} as an explicit copy-consuming comparator: each qubit is measured
independently in a computational basis after performing a uniformly random local Clifford gate.
It is known that the corresponding single-shot estimator for one $k$-local Pauli expectation value has the variance at most $3^k$. 
Hence, using the same base failure probability and the common median
amplification gives the copy count
\begin{equation}
    N_{\rm copy}^{\rm Pauli}(k)
    =
    \Upsilon_{\rm med}(\delta_{\rm est})
    \left\lceil
      \frac{3^k}{p_\star\varepsilon^2}
    \right\rceil.
\end{equation}

\bibliographystyle{unsrt}
\bibliography{ref}

\end{document}